\documentclass[journal]{IEEEtran}
            
\ifCLASSOPTIONcompsoc 
  \usepackage[caption=false,font=normalsize,labelfont=sf,textfont=sf]{subfig}
\else
  \usepackage[caption=false,font=footnotesize]{subfig}
\fi
   
\usepackage{cite}
\usepackage{amssymb}
\usepackage{mathrsfs}
\usepackage{amsmath}
\usepackage{graphicx}
\usepackage{array}
\usepackage{multirow}
\usepackage{color} 
\usepackage{xcolor}
\usepackage{epstopdf}
\usepackage{cite}
\usepackage{amsfonts}

\usepackage{subeqnarray}
\usepackage{cases}

\usepackage[twocolumn,letterpaper]{geometry}

\newtheorem{theorem}{\bf{Theorem}}

\newtheorem{condition}{\bf{Assumption}}

\newtheorem{lemma}{\bf{Lemma}}

\newtheorem{problem}{\bf{Problem}}

\newtheorem{remark}{\bf{Remark}}

\begin{document} 

\title{\Large Fault-Tolerant Formation Tracking of Heterogeneous Multi-Agent Systems with Time-Varying Actuator Faults and Its Application to Task-Space Cooperative Tracking of  Manipulators} 
 
\author
{
Zhi~Feng and~Guoqiang~Hu 
\thanks{ This work was supported by Singapore Ministry of Education Academic Research Fund Tier 1 RG180/17 (2017-T1-002-158). 
Z. Feng and G. Hu are with the School of Electrical and Electronic Engineering, Nanyang Technological University, Singapore 639798 (E-mail: zhifeng, gqhu@ntu.edu.sg).
Part of results was presented in the IEEE Conference on Decision and Control, Korea, 2020. 
}
}


\maketitle

\begin{abstract}
This paper addresses a formation tracking problem for nonlinear multi-agent systems with time-varying actuator faults, in which only a subset of agents has access to the leader's information over the directed leader-follower network with a spanning tree. Both the amplitudes and signs of control coefficients induced by actuator faults are unknown and time-varying. The aforementioned setting improves the practical relevance of the problem to be investigated, and meanwhile, it poses technical challenges to distributed controller design and asymptotic stability analysis. By introducing a distributed estimation and control framework, a novel distributed control law based on a Nussbaum gain technique is developed to achieve robust fault-tolerant formation tracking for heterogeneous nonlinear multi-agent systems with time-varying actuator faults. It can be proved that the asymptotic convergence is guaranteed. In addition, the proposed approach is applied to task-space cooperative tracking of networked manipulators irrespective of the uncertain kinematics, dynamics, and actuator faults. Numerical simulation results are presented to verify the effectiveness of the proposed designs.     
\end{abstract}

\vspace*{-3pt} 

\IEEEpeerreviewmaketitle

\vspace*{-10pt}
\section{Introduction}
\vspace*{-3pt}
In coordination of multi-agent systems, formation tracking has attracted considerable attention during the past decades due to its broad potential applications such as cooperative localization \cite{AI06TRO}, surveillance \cite{Nigam12TCST}, target enclosing \cite{Sun15AT}, source seeking and mapping, etc. Although several classic formation control strategies (e.g., the leader-follower, virtual structure, and behavior-based ones in \cite{Ahn15AT}) are proposed to drive the states of all the agents to form a desired configuration, consensus-based formation control frameworks are presented via the local neighboring relative interaction for agents with single-/double-integrator \cite{Lin16TAC}, high-order linear \cite{Dong17TAC, Dong16AT}, and nonholonomic \cite{Liu18TRO} dynamics. 
Since many uncertainties and nonlinearities are unavoidable in the practical physical system, the approaches presented in aforementioned works cannot be directly applied to solve formation tracking issues for unknown nonlinear multi-agent systems. On the other hand, the increasingly equipped actuators, sensors, and other components of  multi-agent systems are inevitably subject to various faults that induce interruptions and lead to performance degradation or even instability. Recently, there have been some works reported on fault-tolerant consensus. Float actuator faults were considered in \cite{Wen14AT,Jiang13FS,Yang16IAC,Zhao20Tcyber}, while loss of effectiveness faults were studied in \cite{Wen16AT,Gang16Tcyber,Lewis15IE}. However, uniformly ultimately bounded (UUB) results are obtained for various faults. 
Overall, how to design algorithms to have fault-tolerant formation tracking of nonlinear multi-agent systems is challenging, and to the best of our knowledge, is still open.   

One common feature in aforementioned works assumes that the amplitudes and signs of control coefficients are known a priori for each agent. However, for nonlinear multi-agent systems subject to time-varying actuator faults, the control coefficient of each agent usually becomes unknown and time-varying \cite{Hu19AT}. Moreover, both their amplitudes and signs might not be available  a priori in many applications \cite{19ATChen}. For example, position of each robot influences its own controlling effect in the robot formation tracking problem. To handle this issue, there are mainly three systematic approaches: the switching detection method \cite{16KanACC}, the nonlinear proportional-integral (PI) scheme \cite{17HarisTAC,19WangTAC}, and the Nussbaum gain technique \cite{14TACWen1,Chen17TAC,Fan19TAC,15TACSu,Huang17TAC}. In particular, the switching mechanism is proposed in \cite{16KanACC} to deal with the unknown control sign, where the nonsmooth design may bring undesired chattering behaviors. The authors in \cite{17HarisTAC,19WangTAC} develops nonlinear PI schemes to solve this problem for a networked of single- and double-integrator agents over strongly connected graphs, 
where the design relies on the assumption that the control coefficients have to be constant. 
A class of Nussbaum-type functions are proposed in \cite{14TACWen1} to achieve consensus for first- and second-order systems with constant coefficients and identical signs. 
The piecewise Nussbaum function is developed in \cite{Chen17TAC} to allow nonidentical but partially unknown control signs. Moreover, the design is extended in \cite{Fan19TAC} to study output-constrained consensus with partially unknown control signs. Distributed cooperative output regulations are studied in \cite{15TACSu,Huang17TAC,15TACLiu}, where internal model schemes are used to allow each agent to have constant coefficients under connected and undirected graphs. 


This work focuses on the fault-tolerant formation tracking research of heterogeneous nonlinear multi-agent systems with time-varying actuator faults over the directed leader-follower network. In particular, the distributed estimation and control framework is developed to solve this fault-tolerant formation tracking problem.  
The main contributions of this paper are summarized as follows. \textbf{(a)} Asymptotic fault-tolerant formation tracking is achieved under the proposed distributed estimation and control algorithm over a directed graph with a spanning tree. To the best of our knowledge, this paper is the first attempt to solve this issue. 

\vspace*{-2pt}
\begin{itemize}
\item In contrast to existing formation/consensus works in \cite{Lin16TAC,Dong17TAC, Dong16AT,Liu18TRO,Wen14AT,Jiang13FS,Yang16IAC,Zhao20Tcyber,Wen16AT,Gang16Tcyber,Lewis15IE} with known and constant control coefficients, the presence of time-varying actuator faults in nonlinear multi-agent systems makes this control coefficient  time-varying with completely unknown nonidentical signs.

\item Although works in \cite{14TACWen1,Chen17TAC,Fan19TAC,15TACSu,15TACLiu,Huang17TAC} also adopt the Nussbaum gain technique to achieve consensus, the designs require constant control coefficients in \cite{14TACWen1,15TACSu,Huang17TAC,15TACLiu} or dynamic coefficients but with partially known signs in \cite{Chen17TAC,Fan19TAC}. 

\item Unlike aforementioned leaderless consensus works, we consider a directed leader-follower network to have fault-tolerant formation tracking asymptotically. The proposed distributed algorithm does not need the upper bound of faults as required in \cite{Wen14AT,Jiang13FS,Yang16IAC,Zhao20Tcyber,Wen16AT,Gang16Tcyber,Lewis15IE} where the UUB consensus is obtained. 
\end{itemize}
\textbf{(b)} Unlike existing works over undirected or strongly connected graphs, the developed framework enables agents to communicate over the directed graph with a spanning tree, which is nontrivial. Due to the time-varying control coefficients with completely unknown nonidentical signs in the directed leader-follower network, existing designs 
cannot be directly applied. 

\vspace{2pt}
\hspace{-1.2em}
\textbf{(c)} Although the preliminary version of this paper was presented in \cite{Hu20CDC}, this paper unifies and generalizes our conference contributions in \cite{Hu20CDC}. In particular, the proposed distributed algorithm is applied to solve the fault-tolerant task-space coordinated tracking problem for manipulators. 
As compared to most related works in \cite{WangTAC,Zhang17Tcyber,LiuTRO,WangAT,LiangTcyber,He20TNNLS}, the proposed algorithm can bring several advantages. Specifically, the proposed schemes in \cite{Zhang17Tcyber,WangAT,LiuTRO,WangTAC,LiangTcyber,He20TNNLS} require the global task reference for each robot to accomplish cooperative tasks. 
On the contrary, the distributed framework is developed to guarantee task-space coordination. The directed graph with a spanning tree is more general. 
Moreover, the proposed design does not require the upper bounds of unknown faults/uncertainties.

\vspace*{3pt}
This paper is organized as follows. Section II gives the problem formulation. The distributed estimation and control framework is developed in Section III to provide the main results. Section IV presents its application to task-space coordinated tracking control for networked manipulators. Simulation results are presented in Section V followed by the conclusions provided in Section VI.  


\vspace*{-5pt}
\section{Preliminaries and Problem Formulation}
\vspace*{-1pt}
\subsection{Notation \label{Notation}}
\vspace*{-2pt}
Let $0_{N}$ ($1_{N}$) be the $N\times 1$ vector with all the zeros (ones). Let col$(x_{1},...,x_{N})$ and diag$\{a_{1},...,a_{N}\}$ be the column vector with entries $x_{i}$ and the diagonal matrix with entries $a_{i}$, $i=1,\cdots,N$, respectively.  $\otimes $ and $\left\Vert \cdot \right\Vert $ represent the Kronecker product and the Euclidean norm, respectively. For $ x_{i} \in \mathbb{R}^{n} $, define $ \text{sig}^{\theta}(x_{i})=\text{col}(\text{sig}^{\theta}(x_{i1}),\cdots,\text{sig}^{\theta}(x_{in})) $, where $ \text{sig}^{\theta}(x_{ik})=\text{sgn}(x_{ik})|x_{ik}|^{\theta} $, $ k=1,\cdots,n $,  $\text{sgn}(x_{ik})$ is a signum function, and $ 0<\theta<1$. For a real symmetric matrix $M$, $ M>0 $ means that it is positive definite. Moreover, $\lambda _{\min }(M)$ and $\lambda _{\max }(M)$ are its minimum and maximum eigenvalues, respectively.   

\vspace*{-6pt}
\subsection{Graph Theory\label{Graph theory}}
\vspace*{-2pt}
Let $\mathcal{G}$ $=$ $\left\{ \mathcal{V},\mathcal{E}\right\} $ represent a digraph, where the set of vertices is defined as $\mathcal{V}$ $\in $ $\left\{ 1,...,N\right\} $, and the set of edges is $\mathcal{E}$ $\subseteq$ $ \mathcal{V\times V}$. 
$\mathcal{N}_{i}= \left\{ j\in \mathcal{V} | (j,i)\in \mathcal{E}\right\} $ denotes the neighborhood set of vertex $i$. For a directed graph $\mathcal{G}$, $(i,j)\in \mathcal{E}$ means that the
information of node $ i $ is accessible to node $ j $, but not conversely. $A$ $=$ $\left[ a_{ij}\right] $ is the adjacency matrix, where $a_{ij}>0$ if $(j,i)\in \mathcal{E}$, else $a_{ij}=0$. A matrix $\mathcal{L}$ $\triangleq $ $D-A$ is called the Laplacian matrix, where $D=\left[ d_{ii}\right] $ is a diagonal matrix with $d_{ii}=\sum\nolimits_{j=1}^{N}a_{ij}
$. Let $\mathcal{\bar{G}} = (\mathcal{\bar{V}}, \mathcal{\bar{E}}) $ be a directed graph of a leader-follower network, where  $\bar{\mathcal{E}} \subseteq \bar{\mathcal{V}} \times \bar{\mathcal{V}} $, \\ $ \mathcal{\bar{V}}=\{0,\cdots,N\} $, and the node $ 0 $ is associated with the leader. 
Clearly, $\mathcal{G}$ is a subgraph of $\mathcal{\bar{G}}$, where $ \mathcal{E} $ is obtained from $ \mathcal{\bar{E}}$ by removing all the edges between the node $0$ and the nodes in $ \bar{\mathcal{V}} $. Define the Laplacian matrix of $\mathcal{\bar{G}}$ as $\mathcal{\bar{L}}= [0, 0^{T}_{N}; -\mathcal{B}1_{N}, H]$
where $\mathcal{B}$ is a diagonal matrix with its $i$-th diagonal element being $a_{i0}$, 
(similarly, $a_{i0}>0$, if $(0,i) \in \mathcal{\bar{E}}$, and $a_{i0}=0$, otherwise), 
and $ H=[h_{ij}] \triangleq \mathcal{L}+\mathcal{B} $ is an information exchange matrix.

\subsection{Problem Formulation}
Consider a class of heterogeneous nonlinear multi-agent systems consisting of $ N $ followers labeled by agents $ 1, 2, \cdots, N $ and one leader labeled by agent $ 0 $. The dynamics of follower $ i $ are    
\vspace*{-3pt}
\begin{equation} \label{SystemModelFollower}
\left\{ 
\begin{array}{l}
\hspace{-0.3em} \dot{x}_{i,k} = x_{i,k+1}, \ k=1,2,\cdots,m-1,  \\
\hspace{-0.5em} \dot{x}_{i,m} = f^{T}_{i,m}(x_{i}) \theta_{i} +g_{i,m}(x_{i})u_{ai}+d_{i,m}(x_{i},t),  \\ 
\hspace{-0.5em} \ \ \ y_{i} = x_{i,1}, \ i=1,2,\cdots, N,
\end{array}
\right. 
\end{equation} 
where $ x_{i}=\text{col}(x_{i,1},\cdots,x_{i,m}) \in \mathbb{R}^{nm}$ is the state vector with $ x_{i,k} \in \mathbb{R}^{n}, k=1,\cdots,m $, $ m $ is the system order, $ n $ is the system dimension, $ u_{ai} \in \mathbb{R}^{n} $ denotes a control input with actuator
faults, $ y_{i} \in \mathbb{R}^{n} $ is the system output, $ \theta_{i} \in \mathbb{R}^{r} $ is the unknown constant parameter, $ f_{i,m}: \mathbb{R}^{nm}  \rightarrow \mathbb{R}^{r \times n} $ is the known nonlinear function, $ g_{i,m}: \mathbb{R}^{nm} \rightarrow \mathbb{R} $ is the unknown coefficient, and $ d_{i,m}: \mathbb{R}^{nm} \rightarrow \mathbb{R}^{n} $ are  uncertainties/disturbances that can be upper bounded by certain unknown constants. 
The sign of $g_{i,m}$ is unknown. 
 
The leader evolves with the following dynamics
\vspace*{-3pt}
\begin{equation} \label{SystemModelLeader}
\left\{ 
\begin{array}{l}
\hspace{-0.3em} \dot{x}_{0,k} = x_{0,k+1}, \ k=1,2,\cdots,m-1,  \\
\hspace{-0.5em} \dot{x}_{0,m} = o_{0,m}(x_{0},u_{0}),   \ y_{0} = x_{0,1},  
\end{array}
\right. 
\end{equation}
where $ x_{0}=\text{col}(x_{0,1},\cdots,x_{0,m}) \in \mathbb{R}^{nm}$ is the state vector with $ x_{0,k} \in \mathbb{R}^{n}, k=1,2,\cdots,m $, $ u_{0} \in \mathbb{R}^{n} $ is the control input of the leader, $ y_{0} \in \mathbb{R}^{n} $ is the leader's output, and $ o_{0,m}: \mathbb{R}^{nm} \times \mathbb{R}^{n} \rightarrow \mathbb{R}^{n} $ is an unknown and nonlinear input function.

In contrast to many existing works in distributed coordination based on healthy actuation of multi-agent systems, actuators with undetectable faults are considered. When actuation faults occur, there exists the  discrepancy between the actual control input $ u_{ai} $  and the designed control input $ u_{i} $ of the $ i $-th actuator. Thus, two types of actuator faults that may take place, are modeled as
\vspace*{-3pt}
\begin{equation}
u_{ai}= \phi_{i}(t) u_{i} + \psi_{i}(t), \ i=1,2,\cdots,N,  \label{Fault}
\end{equation}%
where $ u_{i} \in \mathbb{R}^{n} $, and $ \phi_{i}(t) \in \mathbb{R} $, $ \psi_{i}(t) \in \mathbb{R}^{n} $ denote the actuation loss of effectiveness fault and the float fault, respectively. 

\vspace*{2pt}
\begin{remark}
The multi-agent model in (\ref{SystemModelFollower}) involves the heterogeneous nonlinear dynamics and intrinsic unknown parameters. 
The system in (\ref{SystemModelLeader}) represents a class of non-autonomous leaders with an unknown and nonlinear input function. The fault formulation in (\ref{Fault}) describes a more generalized form that captures more failure processes than those in \cite{Wen14AT,Jiang13FS,Yang16IAC,Zhao20Tcyber,Wen16AT,Gang16Tcyber,Lewis15IE}, where either float faults or loss of effectiveness actuation faults are considered. It follows from (\ref{SystemModelFollower}) and (\ref{Fault}) that the control coefficients are thus time-varying with completely unknown nonidentical signs.  
\end{remark}

As aforementioned, the main technical challenges of this paper lie in handling the time-varying unknown control coefficient with nonidentical signs for nonlinear multi-agent systems (\ref{SystemModelFollower})-(\ref{SystemModelLeader}) under time-varying actuator faults  (\ref{Fault}) over the directed graph. 

The problem of this paper is stated as follows.
 
\begin{problem} \label{Problem}
Given the nonlinear multi-agent system composed of (\ref{SystemModelFollower})-(\ref{Fault}) and a directed graph $ \bar{\mathcal{G}} $, design a distributed control law $ u_{i} $ for an auxiliary state $ \hat{x}_{i} $ and two smooth functions $ h_{i1}, h_{i2} $,  
\vspace*{-3pt}
\begin{equation}
u_{i}=h_{i1}(x_{i}, \hat{x}_{i}),  \  \dot{\hat{x}}_{i}=h_{i2}(\hat{x}_{i}, \hat{x}_{j}), i \in \mathcal{V}, j \in \bar{\mathcal{N}}_{i},   \label{ControlLaw}
\end{equation}%
such that for each agent $ i $, it can track the nonlinear leader while maintaining a prescribed formation pattern in the sense that 
\vspace*{-3pt}
\begin{equation}
\underset{t\rightarrow \infty}{\text{lim}}  
(y_{i}(t)-y_{0}(t)) = \varDelta_{i} \ \text{and} \ \underset{t\rightarrow \infty }{\text{lim}}  (x_{i,k}(t)-x_{0,k}(t)) = 0_{n}, \label{Problem1}
\end{equation}%
where $ k=2,\cdots,m $ and $ \varDelta_{i} \in \mathbb{R}^{n} $ is a desired formation offset.     
\end{problem}

To solve Problem 1, the following assumptions and some useful lemmas are introduced to facilitate the development of distributed algorithm and Lyapunov stability analysis. 

\vspace*{2pt}
\begin{condition} \label{FollwerGains}
The sign of $ g_{i,m} \neq 0 $ in (\ref{SystemModelFollower}) is unknown, but there exist two continuous positive functions  
$ g^{-}_{i,m}(x_{i})$, $g^{+}_{i,m}(x_{i}) $ so that $ g^{-}_{i,m}(x_{i}) \leq |g_{i,m}| \leq  g^{+}_{i,m}(x_{i})$. 
\end{condition}

\begin{condition} \label{FaultAssumption}
The faults $ \phi_{i}(t) $ in (\ref{Fault}) satisfy: $ 0<\phi_{i}(t) \leq 1 $ and $ \psi_{i}(t) $ is bounded by certain unknown constants.
\end{condition}

\begin{condition} \label{LeaderAssumption}
The leader's input $ o_{0,m}(x_{0},u_{0}) $ in (\ref{SystemModelLeader}) and its time derivative are bounded by certain unknown constants. Only $ x_{0,k}$,  $k=1,\cdots,m $ are available to a subset of followers.  
\end{condition}

\begin{condition} \label{CommunicaitonAssumption}
The directed graph $\mathcal{\bar{G}}$ contains a  spanning tree with the leader $ 0 $ being the root. 
\end{condition}

\begin{remark}
Assumption \ref{FollwerGains} implies that the sign of coefficients is completely unknown, and allowed to be different for each agent, which is a much weaker assumption compared to existing related works in   \cite{14TACWen1,Chen17TAC,Fan19TAC,15TACSu,15TACLiu,Huang17TAC} considering constant/time-varying coefficients with known or partially known signs. Assumption \ref{FaultAssumption} is adopted for the robust design and is conventional in a single agent system with actuator faults. Assumption \ref{LeaderAssumption} is widely used in existing works for a nonlinear leader system. 
Assumption \ref{CommunicaitonAssumption} 
is a standard assumption for consensus works in the existing literature. 
\end{remark} 



\vspace*{2pt}
\begin{lemma} \label{young} \cite{Wen14AT}
Let $ \xi_{1},\xi_{2},\cdots, \xi_{N}\geq 0 $. Then, we can have $ (\sum_{i=1}^{N} \xi_{i} )^{p} \leq \sum_{i=1}^{N}\xi^{p}_{i}  \leq N^{1-p} (\sum_{i=1}^{N} \xi_{i} )^{p} $ for $ 0<p\leq 1$, and  $N^{1-p} ( \sum_{i=1}^{N} \xi_{i} )^{p} \leq \sum_{i=1}^{N}\xi^{p}_{i} \leq (\sum_{i=1}^{N} \xi_{i} )^{p}$ for $p>1$.	
\end{lemma}
 
\vspace*{2pt}
\begin{lemma}
	\label{DirectedGraphTheta}\cite{Hu16TCNS1} Under Assumption \ref{CommunicaitonAssumption}, $H$ is positive definite, and there exists a positive diagonal matrix $\Pi= \text{diag} \{\pi_{1},\pi_{2}, \cdots ,$ $\pi_{N}\}$ such that $\Xi =(\Pi H+H^{T}\Pi)/2 $ is symmetric and positive definite, where $\pi = \text{col} (\pi_{1}, \pi_{2},\cdots, \pi_{N}) =(H^{T})^{-1}1_{N}$.
\end{lemma}
 

\vspace*{1pt} 
\begin{lemma} \label{VaryingGain} \cite{15ATZuo} 
For $ \xi \in \mathbb{R} $ and $ \gamma >0 $, the following inequality holds: $ 0\leq |\xi|-\xi^{2}/ \sqrt{\xi^{2}+\gamma^{2}} \leq \gamma $. 
\end{lemma}

\vspace*{1pt}
\begin{lemma} \label{lemma2}  \cite{Hu19TCNS}
Consider a dynamic system $\dot{x} = f(x,t) $, $x \in \mathbb{R}^{n} $ with $f(0,t)=0_{n}$. Suppose that there exists a positive definite $\mathcal{C}^{1}$ Lyapunov function $V(x,t)$ defined on a neighborhood of the origin ($ \mathcal{D} \in \mathbb{R}^{n}$), and there are constants $a>0$, $ b\in (0,1) $, and an open neighborhood $ \mathcal{U} \subseteq \mathcal{D} $ so that $ \dot{V}(x,t)+a V^{b}(x,t) \leq 0, \ x \in \mathcal{U} \setminus {0}$, 
then the origin of this system is finite-time stable, and the settling time is described by $T \leq V^{1- b}(x(0),t) / (a (1-b))$.    		
\end{lemma} 

\begin{lemma} \label{Barbalat} (Barbalat's Lemma, \cite{BookKhail}): Let $ f(t): \mathbb{R}\rightarrow \mathbb{R}$ be a uniformly continuous function for $ t \geq 0 $. If 
$ \lim_{t \rightarrow \infty} \int_{0}^{t} f(\omega)d\omega $ exists and is finite, then $ \lim_{t\rightarrow \infty} f(t)=0$.
\end{lemma}

\vspace*{-5pt}
\section{Main Result} 
\vspace*{-2pt}
In this section, we first present a distributed nonlinear estimator to cooperatively estimate the states of the nonlinear leader system for each agent so that its output tracks the reference trajectory in a finite time. As mentioned in Section I, this estimator design is significant as it can not only handle the case when only a subset of followers has access to the leader's states, but also provide a distributed solution to adopt Nussbaum gain technique for each agent having unknown control coefficients with nonidentical signs. Each agent can use its associated estimator's output as a local reference. Then, this issue can be transformed into a simultaneous tracking problem. As a result, we propose a distributed adaptive controller to achieve  fault-tolerant formation tracking for a nonlinear leader-follower agent system subject to time-varying actuator faults and completely unknown nonidentical control signs.  

\subsection{Distributed Nonlinear Leader Estimator Design}
\vspace*{-1pt}
\textit{\underline{\textbf{Distributed Nonlinear Estimator}}}: for the nonlinear leader, motivated by the design in our preliminary work \cite{Hu19TCNS}, a  finite-time distributed estimator is designed for $ i\in\mathcal{V} $, $k=1,\cdots,m-1$, 
\vspace*{-4pt}
\begin{subequations}\label{Estimator}	
\begin{align}  
\hspace{-1.1em}
\dot{\hat{x}}_{i,k}&=\hat{x}_{i,k+1}+\kappa_{k i} \text{sig}^{\gamma}  (  \sum_{j=0}^{N}a_{ij}(\hat{x}_{j,k}-\hat{x}_{i,k}) ),  \hat{x}_{0,k}=x_{0,k},    \label{E0} \\ 
\hspace{-1.1em}
\dot{\hat{x}}_{i,m}&=\eta_{i}+\kappa_{mi} \text{sig}^{\beta}  (  \sum_{j=0}^{N}a_{ij}(\hat{x}_{j,m}-\hat{x}_{i,m}) ) ,  \hat{x}_{0,m}= x_{0, m},  \label{Ea} \\ 
\hspace{-1.1em} 
\dot{\eta}_{i}&=\kappa_{\eta i} [\text{sig}^{\alpha}  (e^{\eta}_{i} ) + \text{sgn} ( e^{\eta}_{i}) ],  e^{\eta}_{i}= \sum_{j=1}^{N}a_{ij}(\eta_{j}-\eta_{i}) + e^{\varrho}_{i},  \label{Eb} 
	\end{align}
\end{subequations}
where $ \kappa_{k i}$, $\kappa_{mi} $, $\kappa_{\eta i}  \in \mathbb{R}$ are positive gains, $\alpha$, $\beta$, $ \gamma $ $ \in (0.5,1)$,  and $ \hat{x}_{i,k}$, $ k=1,2,\cdots,m $, and $\eta_{i} $ are the state estimates of the leader's states $ x_{0,k} $, and $ \dot{x}_{0,m} $, respectively, for any $\hat{x}_{i,k}(0)$, $\eta_{i}(0) \in \mathbb{R}^{n} $. In (\ref{Eb}), $  e^{\varrho}_{i}= a_{i0}(\varrho_{i}-\eta_{i})$ where $ \varrho_{i} $, an estimate of unavailable $ \dot{x}_{0,m} $, is generated by the following estimator that for $ \kappa_{\xi i}$, $\kappa_{\varrho i}>0$, 
\vspace*{-3pt}
\begin{subequations}\label{filter} 
	\begin{align}
 \dot{\varrho}_{i}&=\kappa_{\varrho i}a_{i0}\text{sgn} \left(x_{0,m}-\xi_{i} \right), \ \varrho_{i}(0)=0_{n},  \label{Fb} \\
 	\dot{\xi}_{i}&=\varrho_{i}+\kappa_{\xi i}a_{i0}\text{sig}^{\frac{1}{2}} \left( x_{0,m}-\xi_{i} \right), \ \xi_{i}(0)=0_{n}. \label{Fa} 
	\end{align}
\end{subequations}

Define the local estimate errors $ \tilde{x}_{i,k}$,  $\tilde{\eta}_{i}$, $\tilde{\xi}_{i}$,  $ \tilde{\varrho}_{i} $  as    
\vspace*{-4pt} 
\begin{subequations}\label{LocalEerrors} 
\begin{align}
\tilde{x}_{i,k}&=\hat{x}_{i,k}-x_{0,k}, \ \tilde{\eta}_{i}= \eta_{i}-\dot{x}_{0,m}, \ k=1,\cdots,m, 
 \label{Eaa} \\ 
\tilde{\xi}_{i}&=\xi_{i}-x_{0,m}, \ \tilde{\varrho}_{i}= \varrho_{i}-\dot{x}_{0,m}, \ i =1,\cdots,N.  \label{Ebb} 
\end{align} 
\end{subequations}

 
Then, based on (\ref{LocalEerrors}), the estimated error dynamics can be derived under (\ref{Estimator}) and (\ref{filter}) for $ k=1,2,\cdots,m-1$, 
\vspace*{-4pt} 
\begin{subequations}\label{DefinedErrors}	
\begin{align} 	
\hspace{-0.4em}
\dot{\tilde{x}}_{i,k}&=\tilde{x}_{i, k+1}+\kappa_{k i} \text{sig}^{\gamma} ( \sum_{j=1}^{N}a_{ij}(\tilde{x}_{j,k}-\tilde{x}_{i,k}) - a_{i0} \tilde{x}_{i,k} ),   \label{D0} \\
\hspace{-0.4em}
\dot{\tilde{x}}_{i,m}&=\tilde{\eta}_{i}+\kappa_{mi} \text{sig}^{\beta} ( \sum_{j=1}^{N}a_{ij}(\tilde{x}_{j,m}-\tilde{x}_{i,m}) - a_{i0} \tilde{x}_{i,m} ),   \label{Da} \\
\hspace{-0.4em}  
\dot{\tilde{\eta}}_{i}&=\kappa_{\eta i}\text{sig}^{\alpha} ( \sum_{j=1}^{N}a_{ij}(\tilde{\eta}_{j}-\tilde{\eta}_{i}) + a_{i0}(\tilde{\varrho}_{i}-\tilde{\eta}_{i} ) )-\ddot{x}_{0,m}  \notag  \\
\hspace{-0.4em} 
& \ \ \ + \kappa_{\eta i}\text{sgn} ( \sum_{j=1}^{N}a_{ij}(\tilde{\eta}_{j}-\tilde{\eta}_{i}) + a_{i0}(\tilde{\varrho}_{i}-\tilde{\eta}_{i} ) ),  i\in\mathcal{V}, \label{Db} \\
\hspace{-0.4em}
\dot{\tilde{\xi}}_{i}&=\tilde{\varrho}_{i}-\kappa_{\xi i}a_{i0}\text{sig}^{\frac{1}{2}} (\tilde{\xi}_{i} ),  \   \dot{\tilde{\varrho}}_{i}=-\kappa_{\varrho i}a_{i0}\text{sgn} (\tilde{\xi}_{i} )-\ddot{x}_{0,m}. \label{Dc}
\end{align}
\end{subequations} 

Next, we define two disagreement estimate errors as    
\vspace*{-3pt}  
\begin{subequations}\label{CollectiveErrors}	
\begin{align}
\bar{x}_{i,k}&=\textstyle \sum_{j=1}^{N}a_{ij}(\tilde{x}_{j,k}-\tilde{x}_{i,k}) - a_{i0} \tilde{x}_{i,k},  k=1,\cdots,m, \label{Eaaa}  \\
\bar{\eta}_{i}&=\textstyle \sum_{j=1}^{N}a_{ij}(\tilde{\eta}_{j}-\tilde{\eta}_{i}) + a_{i0} (\tilde{\varrho}_{i}-\tilde{\eta}_{i} ),  i=1,\cdots,N. \label{Ebbb}
\end{align}  
\end{subequations} 

Substituting (\ref{CollectiveErrors}) into (\ref{DefinedErrors}) gives the following disagreement error dynamics for each order $ k=1,\cdots,m-1$,  
\vspace*{-3pt} 
\begin{subequations}\label{DisagreeErrors}	
\begin{align} 	
\hspace{-0.5em}
\dot{\bar{x}}_{i,k}&=\bar{x}_{i, k+1}+\kappa_{ki} e^{\gamma}_{i,s}, \  
\dot{\bar{x}}_{i,m}= \bar{\eta}_{i} -a_{i0}\tilde{\varrho}_{i} + \kappa_{mi} e^{\beta}_{i,m},  \label{DDa} \\
\hspace{-0.5em} 
\dot{\bar{\eta}}_{i}&=\kappa_{\eta i} e^{\alpha}_{i,\eta} + \kappa_{\eta i}e^{0}_{i,\eta}  - \kappa_{\varrho i} a^{2}_{i0} \text{sgn} (  \tilde{\xi}_{i} ), \ i=1,\cdots,N,   \label{DDb} \\
\hspace{-0.5em}
\dot{\tilde{\xi}}_{i}&=\tilde{\varrho}_{i}-\kappa_{\xi i}a_{i0}\text{sig}^{\frac{1}{2}} (\tilde{\xi}_{i} ),  \   \dot{\tilde{\varrho}}_{i}=-\kappa_{\varrho i}a_{i0}\text{sgn} (\tilde{\xi}_{i} )-\ddot{x}_{0,m}, \label{DDc}  
\end{align}
\end{subequations}
where $ e^{*}_{i,s} = \sum_{j=1}^{N}a_{ij}( \text{sig}^{*}(\bar{x}_{j,s})- \text{sig}^{*}(\bar{x}_{i,s})) - a_{i0}  \text{sig}^{*}(\bar{x}_{i,s} ) $ with $ *=\gamma$ for $ s=1,2,\cdots,m-1 $, $ *= \beta $ for $ s=m $, and $ *= \alpha $ for $ s=\eta $, and $e^{0}_{i,\eta}= \sum_{j=1}^{N}a_{ij}(\text{sgn}(\bar{\eta}_{j})-\text{sgn}(\bar{\eta}_{i})) - a_{i0}\text{sgn}(\bar{\eta}_{i}) $. 

\vspace*{5pt}
Next, finite-time convergence of the proposed distributed nonlinear estimator in (\ref{Estimator}) and (\ref{filter}) is presented as follows.

\begin{theorem} \label{theorem1}
Suppose that Assumptions \ref{LeaderAssumption} and \ref{CommunicaitonAssumption} hold. Under the proposed distributed nonlinear estimator (\ref{Estimator}) and (\ref{filter}), all state estimates are uniformly bounded, and the finite-time estimation is achieved in the sense that 
	$\lim_{t\to T_{1}} a_{i0}\tilde{\varrho}_{i}=0_{n}$,  $\lim_{t\to T_{2}} \tilde{\eta}_{i}=0_{n}$, $\lim_{t\to T_{3}} \tilde{x}_{i,m}=0_{n}$, and $\lim_{t\to T_{4}} \tilde{x}_{i,s}=0_{n}$, $ s=1,\cdots,m-1 $ for certain $ T_{i}>0, i=1,2,3,4$.
\end{theorem}

\vspace*{2pt}
\begin{proof}
The proof includes three steps: \\
\textbf{Step (i)}: prove that 	$\lim_{t\to T_{1}}a_{i0}\tilde{\varrho}_{i}=0_{n}$. Two cases are studied:

When $ a_{i0}=0 $, we obtain that $\lim_{t\to T_{1}}a_{i0}\tilde{\varrho}_{i}=0_{n}$; 
When $ a_{i0}=1 $, motivated by \cite{Moreno12TAC}, we can define an error variable $\zeta_{ik} = \\ \text{col}(\text{sig}^{\frac{1}{2}}(\tilde{\xi}_{ik}),\tilde{\varrho}_{ik}) $, where $ \tilde{\xi}_{ik}, \tilde{\varrho}_{ik}$ are the $k$th element of $ \tilde{\xi}_{i}$, $ \tilde{\varrho}_{i} $, $ k=1,\cdots,n $, respectively. Select a Lyapunov function candidate as $ V_{\zeta}(t)=\sum_{i=1}^{N}\sum_{k=1}^{n}V_{\zeta ik}(t)=\sum_{i=1}^{N} \sum_{k=1}^{n} \zeta^{T}_{ik}P_{ik}\zeta_{ik}$, where $ P_{ik} >0$ is a constant matrix. 
Let $x_{0k,m} $ be the $k$th element of  $ x_{0,m} $. Then, the time derivative of $ \zeta_{ik}  $ is given by  \cite{Moreno12TAC}
\vspace*{-8pt}
\begin{equation}
\hspace{-0.6em}
\dot{\zeta}_{ik}=  \frac{1}{2} |\tilde{\xi}_{ik}|^{-\frac{1}{2}}  \left[
\begin{array}{c}
-\kappa_{\xi i} \text{sig}^{\frac{1}{2}}(\tilde{\xi}_{ik})+\tilde{\varrho}_{ik} \\
-2[\kappa_{\varrho i}-\ddot{x}_{0k,m}\text{sgn}(\tilde{\xi}_{ik})]\text{sig}^{\frac{1}{2}}(\tilde{\xi}_{ik})
\end{array}%
\right].  \notag 
\end{equation}


\vspace*{-8pt}
Then, the time derivative of $ V_{\zeta} $ along the system (\ref{DDc}) is given by  $\dot{V}_{\zeta} = \sum_{i=1}^{N} \sum_{k=1}^{n}|\tilde{\xi}_{ik}|^{-\frac{1}{2}}  \zeta^{T}_{ik} \left( R^{T}_{ik}P_{ik}+P_{ik}R_{ik} \right) \zeta_{ik} $,
where $ R_{ik}=\begin{bmatrix} -\frac{1}{2}\kappa_{\xi i}  & \frac{1}{2}  \\
-[\kappa_{\varrho i}- \ddot{x}_{0k,m}\text{sign}(\tilde{\xi}_{ik})] & 0 \\
\end{bmatrix}$ is Hurwitz if and only if $ \kappa_{\xi i} $ $ >0 $ and $ \kappa_{\varrho i}> \sup_{t\in (0,\infty)} \{ \| \ddot{x}_{0,m} \|_{\infty}\}+1 $ based on Assumption \ref{LeaderAssumption}. Since $ R_{ik}$ is Hurwitz, for each matrix $ \Gamma_{ik}>0$, we can find a matrix $ P_{ik} >0$ to the following algebraic Lyapunov inequality: $ R^{T}_{ik}P_{ik}+P_{ik}R_{ik} \leq -\Gamma_{ik}, i=1,\cdots,N, k=1,\cdots,n $, such that for the constructed strict Lyapunov function $ V_{\zeta} $ \cite{Moreno12TAC}, we have that $ \dot{V}_{\zeta} =- \sum_{i=1}^{N}\sum_{k=1}^{n} |\tilde{\xi}_{ik}|^{-\frac{1}{2}}  \zeta^{T}_{ik} \Gamma_{ik}\zeta_{ik} \leq 0$. On the other hand, we have $ |\tilde{\xi}_{ik}|^{\frac{1}{2}} = |\text{sig}^{\frac{1}{2}} (\tilde{\xi}_{ik} )| \leq  | \zeta_{ik} | \leq  \lambda^{-\frac{1}{2}}_{min}(P_{ik})V^{\frac{1}{2}}_{1ik}$. Then, defining $\epsilon_{0}= \text{min}_{i, k} \lbrace \lambda^{\frac{1}{2}}_{min}(P_{ik})\lambda_{min}(\Gamma_{ik}) /\lambda_{max}({P_{ik}}) \rbrace>0$ and using Lemma \ref{young} yield the following inequality  
\vspace*{-5pt}
\begin{equation}
\hspace*{-0.5em}
\dot{V}_{\zeta} \leq-\sum_{i=1}^{N} \sum_{k=1}^{n} \lambda^{\frac{1}{2}}_{min}(P_{ik})V^{-\frac{1}{2}}_{\zeta ik} \frac{\lambda_{min}(\Gamma_{ik})}{\lambda_{max}(P_{ik})} V_{\zeta ik} \leq -\epsilon_{0}V^{\frac{1}{2}}_{\zeta}. \label{A2} 
\end{equation}

\vspace*{-5pt}
Based on Lemma \ref{lemma2}, $ V_{\zeta} \in \mathcal{L}_{\infty}$ and $\tilde{\varrho}_{i}$ converges to zero in a finite time, e.g.,  $\lim_{t\to T_{1}} \tilde{\varrho}_{i}=0_{n}$ with $T_{1}= \frac{2}{\epsilon_{0}}V^{\frac{1}{2}}_{\zeta}(0)$. Overall, we obtain that $\lim_{t\to T_{1}}a_{i0}\tilde{\varrho}_{i}=0_{n}$ for $ a_{i0}= 0 \ \text{or} \ 1$.

\vspace*{2pt}	
\textbf{Step (ii)}: prove that $\lim_{t\to T_{2}} \tilde{\eta}_{i}=0_{n}$. 

\vspace*{2pt}
Consider a nonnegative Lyapunov function candidate  $ V_{\eta}(t)= \sum_{i=1}^{N} \sum_{k=1}^{n} \pi_{i} (|\bar{\eta}_{ik}|+\frac{1}{\alpha+1}|\bar{\eta}_{ik}|^{\alpha+1}) $, where $ \bar{\eta}_{ik}$ is the $ k $th entry of $ \bar{\eta}_{i} $ in (\ref{Ebbb}) and $ \pi_{i} $ is given in Lemma \ref{DirectedGraphTheta}. The 
upper bound of $ V_{\eta}(t) $ is derived by Lemma \ref{young} in two cases:    

\vspace{2pt}
Case a) For $ \sum_{i=1}^{N}\sum_{k=1}^{n} \pi_{i} |\bar{\eta}_{ik}|>1 $, $ V_{\eta} \leq  \sum_{i=1}^{N} \sum_{k=1}^{n} \pi_{i} |\bar{\eta}_{ik} \\ |   + \frac{1}{\alpha+1} \sum_{i=1}^{N} \sum_{k=1}^{n} \pi_{i}  |\bar{\eta}_{ik}|^{\alpha+1} \leq \bar{\pi} \frac{\alpha+2}{\alpha+1} (\sum_{i=1}^{N} \sum_{k=1}^{n}    \bar{\eta}^{2}_{ik})^{\frac{\alpha+1}{2}}$ with $ \bar{\pi}$$=\max \{\pi_{i}\} $. Thus, we have 
\vspace*{-5pt}
\begin{equation}
\sum_{i=1}^{N}\sum_{k=1}^{n} |\bar{\eta}_{ik}| \geq \left( \frac{(\alpha+1)V_{\eta}}{\bar{\pi}(\alpha+2)}  \right)^{\frac{1}{\alpha+1}}.  \label{U1}  
\end{equation}

\vspace*{-3pt} 
Case b) For $ \sum_{i=1}^{N}\sum_{k=1}^{n} \pi_{i} |\bar{\eta}_{ik}|\leq 1 $, $ V_{\eta} \leq \bar{\pi}(1+1/(\alpha+1))$ $  \sum_{i=1}^{N} \sum_{k=1}^{n}    |\bar{\eta}_{ik}| $. Thus, we obtain 
\vspace*{-6pt}
\begin{equation}
\sum_{i=1}^{N}\sum_{k=1}^{n} |\bar{\eta}_{ik}| \geq \left( \frac{(\alpha+1)V_{\eta}}{\bar{\pi}(\alpha+2)}  \right).  \label{U2}  
\end{equation}

\vspace*{-3pt} 
Let $ \bar{\eta} $ be the stacked column  vector of $ \bar{\eta}_{i} $ and  $ \kappa_{\eta}=\text{diag}\{\kappa_{\eta i}\} $. Then, the time derivative of $ V_{\eta}(t) $ along (\ref{DDb}) is given by 
\vspace*{-5pt}
\begin{align} 
\hspace{-0.5em}
\dot{V}_{\eta} 
&=\sum_{i=1}^{N}\sum_{k=1}^{n} \pi_{i} ( \text{sig}^{\alpha}(\bar{\eta}_{ik}) +\text{sgn}(\bar{\eta}_{ik}) ) \times [-a^{2}_{i0} \kappa_{\varrho i} \text{sgn} (  \tilde{\xi}_{ik} )   \notag  \\ 
& \ \ \ +  \kappa_{\eta i} (\sum_{j=1}^{N}a_{ij} (\text{sig}^{\alpha}(\bar{\eta}_{jk}) -\text{sig}^{\alpha}(\bar{\eta}_{ik}) )- a_{i0}\text{sig}^{\alpha}(\bar{\eta}_{ik})) \notag   \\ 
& \ \ \ +  \kappa_{\eta i} (\sum_{j=1}^{N}a_{ij} (\text{sgn}(\bar{\eta}_{jk}) -\text{sgn}(\bar{\eta}_{ik}) )- a_{i0}\text{sgn}(\bar{\eta}_{ik}))] \notag   \\ 
& = - [ \text{sig}^{\alpha}(\bar{\eta}) +\text{sgn}(\bar{\eta} ) ]^{T} (\kappa_{\eta}  \Pi H \otimes I_{n})  [ \text{sig}^{\alpha}(\bar{\eta}) +\text{sgn}(\bar{\eta} ) ]  \notag   \\  
& \ \ \ - \sum_{i=1}^{N}\sum_{k=1}^{n} \kappa_{\varrho i}  \pi_{i}  a^{2}_{i0}  ( \text{sig}^{\alpha}(\bar{\eta}_{ik}) +\text{sgn}(\bar{\eta}_{ik}) )  \text{sgn} (  \tilde{\xi}_{ik} )   \notag   \\ 
& \leq - \lambda_{min}(\Xi)  \underline{\kappa}^{2}_{\eta} \sum_{i=1}^{N}\sum_{k=1}^{n} \left( 1+ 2 |\bar{\eta}_{ik} |^{\alpha} + |\bar{\eta}_{ik} |^{2\alpha} \right) + Nn\bar{\pi} \bar{k}_{\varrho}   \notag   \\ 
& \ \ \ + Nn \bar{\pi} \bar{\kappa}_{\varrho} \sum_{i=1}^{N}\sum_{k=1}^{n} |\bar{\eta}_{ik} |^{\alpha} \leq - \bar{\pi} \bar{\kappa}_{\varrho} \sum_{i=1}^{N}\sum_{k=1}^{n}  |\bar{\eta}_{ik} |^{\alpha} ,   \label{U3} 
\end{align} 
where $ \Xi=(\Pi H+H^{T}\Pi)/2 $ is given in Lemma \ref{DirectedGraphTheta}, $ \underline{\kappa}_{\eta}=\min \{\kappa_{\eta i}\} $ \\ ,  $ \bar{\kappa}_{\varrho}=\max \{\kappa_{\varrho i}\} $ and $ \underline{\kappa}^{2}_{\eta} \geq \bar{\pi} \bar{\kappa}_{\varrho} / \lambda_{min}(\Xi) $. 

By Lemma \ref{young}, $ \dot{V}_{\eta}  \leq -  \bar{\pi} \bar{\kappa}_{\varrho} \left(  \sum_{i=1}^{N}\sum_{k=1}^{n} |\bar{\eta}_{ik} | \right)^{\alpha}  \leq - \epsilon_{1} V^{\frac{\alpha}{\alpha+1}}_{\eta} $ by (\ref{U1}) and  $ \dot{V}_{\eta}  \leq -  \bar{\pi} \bar{\kappa}_{\varrho}   \left(  \sum_{i=1}^{N}\sum_{k=1}^{n} |\bar{\eta}_{ik} | \right)^{\alpha} \leq  - \epsilon_{2} V^{\alpha}_{\eta} $ by (\ref{U2}), 
where $ \epsilon_{1}= \bar{\pi} \bar{\kappa}_{\varrho} (\frac{\alpha+1}{\bar{\pi}(\alpha+2)})^{\frac{\alpha}{\alpha+1}} $ and $ \epsilon_{2}= \bar{\pi} \bar{\kappa}_{\varrho}  (\frac{\alpha+1}{\bar{\pi}(\alpha+2)})^{\alpha} $. 

\vspace*{3pt}
Since $ 0.5<\alpha<1 $, exploiting Lemma \ref{lemma2} yields $ V_{\eta} \in \mathcal{L}_{\infty}$ and $\bar{\eta}_{i}$ converges to zero in a finite time, e.g.,  $\lim_{t\to T_{\eta}} \bar{\eta}_{i}=0_{n}$ with $T_{\eta}=$ $\max \{\frac{1+\alpha}{\epsilon_{1}} V^{\frac{1}{1+\alpha}}_{\eta}(0), \frac{1}{\epsilon_{2}(1-\alpha)} V^{1-\alpha}_{\eta}(0) \}$. As $\lim_{t\to T_{1}}a_{i0}\tilde{\varrho}_{i}=0_{n}$ from Step (i), we get that for $ t \geq T_{1} $, it follows from (\ref{Ebbb}) and Assumption \ref{CommunicaitonAssumption}  that $\lim_{t\to T_{2}} \tilde{\eta}_{i}=0_{n}$, $T_{2}=T_{\eta}+T_{1}$.

\textbf{Step (iii)}: prove that $\lim_{t\to T_{3}} \tilde{x}_{i,m}=0_{n}$. 

Consider a nonnegative Lyapunov function candidate  $ V_{m}(t)= \sum_{i=1}^{N} \sum_{k=1}^{n} \pi_{i} \frac{1}{\beta+1}|\bar{x}_{ik,m}|^{\beta+1} $, where $ \bar{x}_{ik,m}$ is the $ k $th element of $ \bar{x}_{i,m} $. By Lemma \ref{young}, $ (Nn)^{-\beta}\frac{\underline{\pi}}{\beta+1} (\sum_{i=1}^{N} \sum_{k=1}^{n} |\bar{x}_{ik,m}|)^{\beta+1}  \leq  V_{m}(t)  \leq \frac{\bar{\pi}}{\beta+1} (\sum_{i=1}^{N} \sum_{k=1}^{n} |\bar{x}_{ik,m}|)^{\beta+1} $. 

Let $ \bar{x}_{m} $ be a stacked vector of $ \bar{x}_{i,m} $ and  $ \underline{\kappa}_{m}=\min \{\kappa_{m i}\} $. The time derivative of $ V_{m}(t) $ along (\ref{DDa}) can be expressed as 
\vspace*{-6pt}
\begin{align} 
\hspace{-0.5em}
\dot{V}_{m} 
&=\sum_{i=1}^{N}\sum_{k=1}^{n} \pi_{i} \text{sig}^{\beta}(\bar{x}_{ik,m}) (\bar{\eta}_{ik} -a_{i0} \tilde{\varrho}_{ik}) + \sum_{i=1}^{N}\sum_{k=1}^{n} \pi_{i} \text{sig}^{\beta}(\bar{x}_{ik,m})     \notag   \\ 
&  \ \times \kappa_{m i}  (\sum_{j=1}^{N}a_{ij} (\text{sig}^{\beta}(\bar{x}_{jk,m}) -\text{sig}^{\beta}(\bar{x}_{ik,m}) )- a_{i0}\text{sig}^{\beta}(\bar{x}_{ik,m}))    \notag   \\ 
& \leq  \sum_{i=1}^{N}\sum_{k=1}^{n} [- \lambda_{min}(\Xi)  \underline{\kappa}_{m}  | \bar{x}_{ik,m}|^{2\beta} + \bar{\pi} |\bar{x}_{ik,m}|^{^{\beta}} ( |\bar{\eta}_{ik} |+ | \tilde{\varrho}_{ik} |) ].   \notag 
\end{align}  

By using the fact that $\lim_{t\to T_{1}}a_{i0}\tilde{\varrho}_{i}=0_{n}$ and  $\lim_{t\to T_{2}} \bar{\eta}_{i}=0_{n}$ in Steps (i) and (ii), respectively, we get for $ t\geq T_{2} $, 
\vspace*{-5pt}
\begin{equation} 
\dot{V}_{m} \leq - \lambda_{min}(\Xi) \underline{\kappa}_{m}  (  \sum_{i=1}^{N}\sum_{k=1}^{n}  | \bar{x}_{ik,m}| ) ^{2\beta} \leq  -\epsilon_{3} V^{\frac{2\beta}{\beta+1}}_{m},  \label{U7} 
\end{equation}
where $ \epsilon_{3}= \lambda_{min}(\Xi) \underline{\kappa}_{m}  (\frac{\beta+1}{\bar{\pi}})^{\frac{2\beta}{\beta+1}} $. Thus, by Lemma \ref{lemma2}, $ V_{m} \in \mathcal{L}_{\infty}$ and $\bar{x}_{i,m}$ converges to zero in a finite time, e.g., $\lim_{t\to T_{3}} \bar{x}_{i,m}=0_{n}$ with $T_{3}= \frac{\beta+1}{\epsilon_{3}(1-\beta)}V^{\frac{1-\beta}{\beta+1}}_{m}(T_{2}) + T_{2}$. Further, it follows from (\ref{Eaaa}) that $\lim_{t\to T_{3}} \tilde{x}_{i,m}=0_{n}$ by Assumption \ref{CommunicaitonAssumption}.	

Next, we show that $ \tilde{x}_{i,m} $ is uniformly bounded for $ t< T_{3} $. 

In Steps (i) and (ii), it has been shown that
$\lim_{t\to T_{1}}a_{i0}\tilde{\varrho}_{i}=0_{n}$ and  $\lim_{t\to T_{2}} \tilde{\eta}_{i}=0_{n}$. Besides, $ \tilde{\varrho}_{i} $, $ \tilde{\eta}_{i} $ are uniformly bounded for $ t< T_{2} $. Thus, $ |\tilde{\varrho}_{ik}| \leq \varrho_{0} $, $ |\tilde{\eta}_{ik}| \leq \eta_{0} $ for certain constants $\eta_{0} $, $ \varrho_{0} $. Next, we verify that for $ t< T_{3} $ and certain constant $ c_{0}>0 $,
\vspace*{-4pt}
\begin{equation}
\dot{V}_{m} \leq  \sum_{i=1}^{N}\sum_{k=1}^{n} \bar{\pi} |\bar{x}_{ik,m}|^{^{\beta}} ( |\bar{\eta}_{ik} |+ | \tilde{\varrho}_{ik} |) \leq c_{0} V^{\frac{\beta}{\beta+1}}_{m}, \label{U8}
\end{equation} 
which implies that $ \tilde{x}_{i,m} $ cannot escape in a finite time.  
Therefore,  $ \hat{x}_{i,m}$ is uniformly bounded at any finite time interval if $x_{0,m} $ will not escape to infinity in a finite time by Assumption \ref{LeaderAssumption}.

\vspace*{2pt}
\textbf{Step (iv)}: prove that $\lim_{t\to T_{4}} \tilde{x}_{i,s}=0_{n}, s=1,\cdots,m-1$. 

\vspace*{2pt}
Select the nonnegative Lyapunov function candidate as $ V_{s}(t)= \sum_{i=1}^{N} \sum_{k=1}^{n} \pi_{i} \frac{1}{\gamma+1}|\bar{x}_{ik,s}|^{\gamma+1} $ where $ \bar{x}_{ik,s}$ is the $ k $th entry of $ \bar{x}_{i,s} $. Similar to Step (iii), we get $\dot{V}_{s} \leq -\epsilon_{4} V^{\frac{2\gamma}{\gamma+1}}_{s} $ and $\lim_{t\to T_{ts}} \bar{x}_{i,s}=0_{n}$ for certain positive constants $ \epsilon_{4} $ and $ T_{ts} $, $s=1,\cdots,m-1$. Similarly, $\bar{x}_{i,s}$ is uniformly bounded for $ t< T_{4} $, $ T_{4}=\max\{T_{ts}\} $. Thus, $\lim_{t\to T_{4}} \tilde{x}_{i,s}=0_{n}$ by Assumption \ref{CommunicaitonAssumption}.
 
\vspace*{2pt} 
To conclude, estimates $ \varrho_{i} $, $ \eta_{i} $, $ x_{i,s} $, $ s=1,\cdots,m $  are bounded, and finite-time estimation is achieved, i.e., 	$\lim_{t\to T_{1}} a_{i0}\tilde{\varrho}_{i}=0_{n}$,  $\lim_{t\to T_{2}} \tilde{\eta}_{i}=0_{n}$, $\lim_{t\to T_{3}} \tilde{x}_{i,m}=0_{n}$, and $\lim_{t\to T_{4}} \tilde{x}_{i,s}=0_{n}$, $ s=1,\cdots,m-1 $ for certain $ T_{i}>0, i=1,2,3,4$.
\end{proof}

\vspace*{1pt}
\begin{remark}
Theorem 1 implies that each agent can accurately estimate the state of the nonlinear leader after a finite time under the proposed nonlinear distributed estimator. Then, each agent can use its associated estimate as a local reference in the distributed control based on a cascaded structure. As a result, this distributed estimator will facilitate the following distributed controller design so that Nussbaum gains are decoupled for each agent 
as only the local reference instead of neighboring information is utilized. On the other hand, with the adopted function $ \text{sgn}(\cdot) $ in the proposed estimator, the right-hand sides of $ \dot{\tilde{\eta}}_{i} $ and $ \dot{\tilde{\varrho}}_{i} $ in (\ref{Db}) and (\ref{Dc}) are discontinuous, and their solutions can be investigated in terms of differential inclusions based on the nonsmooth analysis. To avoid symbol redundancy, the differential inclusion is not applied.
\end{remark}


\vspace*{-5pt} 
\subsection{Fault-Tolerant Formation Tracking Control}
\vspace*{-1pt} 
With the estimated information in Theorem 1, a new distributed adaptive controller will be developed in this subsection to achieve formation tracking of nonlinear multi-agent systems. 

In particular, the leader's states $ x_{0,k}, k=1,\cdots,m$  have been reconstructed via distributed estimator (\ref{Estimator}) and (\ref{filter}) for each agent. The rest is to propose a novel distributed adaptive control law to solve Problem 1. Specifically, we can transform the coordinated trajectory tracking problem into a simultaneous tracking problem. For $ i\in \mathcal{V} $, we define the following tracking errors as
\vspace*{-4pt} 
\begin{equation}
e_{i,1}=x_{i,1}-x_{0,1}-\varDelta_{i}=y_{i}- \hat{y}_{i} +\hat{y}_{i}-y_{0}-\varDelta_{i}= z_{i,1}+\tilde{y}_{i},  
\label{B1}  
\end{equation}
\vspace*{-12pt}
\begin{equation}
e_{i,k}=x_{i,k}-x_{0,k}=z_{i,k}+\tilde{x}_{i,k}, z_{i,k}=x_{i,k}-\hat{x}_{i,k}, k=2,\cdots,m, \notag 
\end{equation}
where $ \hat{y}_{i}=\hat{x}_{i,1} $, $ z_{i,1}=y_{i}- \hat{y}_{i} -\varDelta_{i}$, $ \tilde{y}_{i}= \hat{y}_{i}-y_{0}$, $ \varDelta_{i} $ is the formation offest defined in (\ref{Problem1}), and 
$ \tilde{x}_{i,k}, k=1,\cdots,m $ is the estimated error defined in (\ref{Eaa}). Then, based on the fact that $\lim_{t\to \infty} \tilde{y}_{i}=0_{n}$ and $\lim_{t\to \infty} \tilde{x}_{i,k}=0_{n}$ from Theorem 1, the coordinated tracking control objective in (\ref{Problem1}) can be transformed into the simultaneous tracking objective as  
\vspace*{-3pt}
\begin{equation}
\lim_{t\to \infty} z_{i,k} =0_{n}, \ k=1,2,\cdots,m, \ i \in \mathcal{V}. \label{transform}
\end{equation} 
 
From (\ref{DefinedErrors}), the time derivative of $ z_{i,k}, \ k=1,\cdots,m-1$ is  
\vspace*{-3pt} 
\begin{subequations}\label{FollowerErrors}	
\begin{align} 	
\hspace{-0.8em}
\dot{z}_{i,k}&=z_{i, k+1}-\kappa_{ki} \text{sig}^{\gamma} ( \sum_{j=1}^{N}a_{ij}(\tilde{x}_{j,k}-\tilde{x}_{i,k}) - a_{i0} \tilde{x}_{i,k} ),   \label{FE1} \\
\hspace{-0.8em}
\dot{z}_{i,m}&=\dot{x}_{i,m}-\eta_{i}-\kappa_{mi} \text{sig}^{\beta} ( \sum_{j=1}^{N}a_{ij}(\tilde{x}_{j,m}-\tilde{x}_{i,m}) - a_{i0} \tilde{x}_{i,m} ) \notag \\
&= f^{T}_{i,m}(x_{i}) \theta_{i} +g_{i,m}(x_{i}) [\phi_{i}(t) u_{i} + \psi_{i}(t)] +d_{i,m}(x_{i},t)  \notag \\ 
& \ \ \ - \eta_{i} -\kappa_{mi} \text{sig}^{\beta} ( \sum_{j=1}^{N}a_{ij}(\tilde{x}_{j,m}-\tilde{x}_{i,m}) - a_{i0} \tilde{x}_{i,m} ) \notag \\
& = f^{T}_{i,m}(x_{i}) \theta_{i} + G_{i,m}(x_{i},t)u_{i} +D_{i,m}(x_{i},t) -\tilde{\eta}_{i}  \notag \\ 
&\ \ \ - \kappa_{mi} \text{sig}^{\beta} ( \sum_{j=1}^{N}a_{ij}(\tilde{x}_{j,m}-\tilde{x}_{i,m}) - a_{i0} \tilde{x}_{i,m} ), \label{FE2} 
\end{align}
\end{subequations} 
where $  G_{i,m}(x_{i},t)= g_{i,m}(x_{i})\phi_{i}(t) \in \mathbb{R} $ and $ D_{i,m}(x_{i},t)=d_{i,m} ( \\ x_{i},t)+g_{i,m}(x_{i},t)\psi_{i}(t) -\dot{x}_{0,m} \in \mathbb{R}^{n} $. It follows from Assumptions \ref{FollwerGains} and \ref{FaultAssumption} that the sign of the control coefficient $  G_{i,m}(x_{i},t) \neq 0 $ is unknown and there exist positive functions  
$ G^{-}_{i,m}(x_{i})$, $G^{+}_{i,m}(x_{i}) $ so that $ G^{-}_{i,m}(x_{i}) \leq |G_{i,m}(x_{i},t)| \leq  G^{+}_{i,m}(x_{i})$. Define the positive constant $ \varepsilon_{i}= \text{sup}_{t\geq 0}\{\|D_{i,m}(x_{i},t)\| \} $. Then, $ \hat{\varepsilon}_{i} $, representing the estimate of this unknown bound vector $ \varepsilon_{i} 1_{n} $, is to be determined later, and $ \tilde{\varepsilon}_{i} =\varepsilon_{i} 1_{n}-\hat{\varepsilon}_{i} $ is its estimated error. 

\vspace*{4pt}
Next, we introduce a Nussbaum gain technique to deal with the time-varying control coefficients with completely unknown signs via a smooth function $ N(k) $ that satisfies the following properties: $ \underset{k \rightarrow \infty }{\text{lim}} \text{sup}  \frac{1}{k} \int^{k}_{0} N(s)ds  =+\infty $, $\underset{k \rightarrow \infty }{\text{lim}} \text{inf}  \frac{1}{k} \int^{k}_{0}N(s)ds=-\infty $. Throughout this paper, select $ N_{i}(\kappa_{i})=\text{exp}(\kappa^{2}_{i})  \text{cos}((\pi/2)  \kappa_{i})+1$ in \cite{Hu19AT} with nonidentical Nussbaum gains $ \kappa_{i} $ for each agent $i=1,\cdots,N$. 
Let $ s_{ij}$, $i=1,\cdots,N$, $j=1,\cdots,m$ represent the $ j $-th element of certain vector $ s_{i} \in \mathbb{R}^{n}$, and define $\text{diag} \{s_{i1},\cdots, s_{im}\} $ as a diagonal matrix with its main diagonal being $ s_{ij}$. Similarly, define $\text{diag} \{\frac{s_{i1}}{\sqrt{s^{2}_{i1}+\delta^{2}_{i1}(t)}},\cdots, \frac{s_{im}}{\sqrt{s^{2}_{im}+\delta^{2}_{im}(t)}}\} $ as a diagonal matrix with its main diagonal being $ \frac{s_{ij}}{\sqrt{s^{2}_{ij}+\delta^{2}_{ij}(t)}}$, where  $  \delta_{ij}(t)>0 $ is an integrable function so that $ \int_{0}^{\infty} \delta_{ij}(\omega)d\omega \leq \delta^{*}_{ij} $
for  $ \delta^{*}_{ij}>0$.

\vspace*{4pt}
Based on the selected Nussbaum function, we propose a novel distributed adaptive controller as follows. 
 
\vspace*{3pt}
\textbf{\underline{\textit{Fault-Tolerant Distributed Adaptive Controller}}}: in light of the Nussbaum function and local estimates $ \hat{x}_{i,k} $ from (\ref{Estimator}) and (\ref{filter}), a fault-tolerant adaptive controller is proposed as
\vspace*{-3pt}
\begin{subequations}\label{Controller}
\begin{align}
\hspace*{-0.86em}
u_{i}&= N_{i}(\kappa_{i})\bar{u}_{i}, \   N_{i}(\kappa_{i})=\text{exp}(\kappa^{2}_{i})\text{cos}((\pi/2)\kappa_{i})+1, \label{S11}  \\
\hspace*{-0.6em}
\bar{u}_{i}&=  \bar{k}_{mi}\tilde{z}_{i,m} +\tilde{z}_{i,m-1} -\dot{z}^{*}_{i,m} +f^{T}_{i,m} \hat{\theta}_{i} + \text{diag} \{\tilde{z}_{\delta m} \}   \hat{\varepsilon}_{i},   \label{S12}  \\
\hspace*{-0.6em}
 \dot{\hat{\varepsilon}}_{i} &=\varGamma_{\varepsilon i}\text{diag} \{\tilde{z}_{\delta m}\} \tilde{z}_{i,m},  \tilde{z}_{\delta m}=\tilde{z}_{ij,m} / \sqrt{\tilde{z}^{2}_{ij,m}+\delta^{2}_{ij}(t)},  \label{S13} \\
\hspace*{-0.6em} \dot{\kappa}_{i}&=k_{\kappa i}\tilde{z}^{T}_{i,m} \bar{u}_{i},   \ \dot{\hat{\theta}}_{i}=\varGamma_{\theta i}f_{i,m}\tilde{z}_{i,m}, \ i=1,2,\cdots,N,
  \label{S14} 
\end{align} 
\end{subequations}
where $ N_{i}(\kappa_{i}) $ is the agent $ i $' Nussbaum function,  $f_{i,m}=f_{i,m}(x_{i}) $,  $ k_{\kappa i} $ is the positive constant, $ \varGamma_{\varepsilon i}, \varGamma_{\theta i}>0$ are two constrant gain matrices,   $\hat{\theta}_{i} $ is the estimate of unknown vector $\theta_{i}$, $ \hat{\varepsilon}_{i} $ is an estimate of the unknown vector $\varepsilon_{i}1_{n} $, and $ \tilde{z}_{i,m}, z^{*}_{i,m} $ are the subsequently defined error terms to be determined.     

Now, we are ready to present the robust fault-tolerant formation tracking result for a class of nonlinear multi-agent systems under time-varying actuator faults over the directed graph.

\vspace*{2pt} 
\begin{theorem}
Consider a class of nonlinear multi-agent systems consisting of $ N $ followers in (\ref{SystemModelFollower}) and the leader in (\ref{SystemModelLeader}) with time-varying actuator faults in (\ref{Fault}). Under Assumptions \ref{FollwerGains}-\ref{CommunicaitonAssumption}, the proposed controller in (\ref{Controller}) combined with the distributed estimator in (\ref{Estimator}) and (\ref{filter}) ensures that Problem \ref{Problem} is solvable in the sense that:  $ \underset{t\rightarrow \infty}{\text{lim}}  
(y_{i}(t)-y_{0}(t)) = \varDelta_{i} \ \text{and} \ \underset{t\rightarrow \infty }{\text{lim}}  (x_{i,k}(t)-x_{0,k}(t)) = 0_{n}$. 
\end{theorem}  

\vspace*{4pt} 
\begin{proof} 
Define an estimated error $ \hat{e}_{i,k} =\text{sig}^{\gamma} ( \sum_{j=1}^{N} a_{ij} (\tilde{x}_{j,k} \\ -\tilde{x}_{i,k}) - a_{i0} \tilde{x}_{i,k}), k=1,\cdots,m$.  Then, for clarity and conciseness, a step by step procedure is presented as follows. 
 

\vspace*{2pt}  
\textbf{Step 1:} Introduce two error variables
\vspace*{-3pt} 
\begin{subequations} \label{Step1} 
\begin{align} 
\tilde{z}_{i,1}&=z_{i,1}=y_{i}-\hat{y}_{i}-\varDelta_{i}, \ i =1,2,\cdots,N,  
\label{S1} \\ 
\tilde{z}_{i,2}&= z_{i,2}-z^{*}_{i,2},  \ z^{*}_{i,2}=-\bar{k}_{1i} \tilde{z}_{i,1},  \label{S2} 
	\end{align} 
\end{subequations}
where $ z^{*}_{i,2} $ is a virtual control input for the first subsystem in (\ref{FollowerErrors}) and $ \bar{k}_{1i}>0 $ is a constant gain. 

Choose the Lyapunov candidate as $ V_{i,1}=\frac{1}{2} \tilde{z}^{T}_{i,1}\tilde{z}_{i,1} $. Then, the time derivative of $ V_{i,1} $ along the first subsystem in (\ref{FollowerErrors}) is  
\vspace*{-3pt} 
\begin{equation}
\dot{V}_{i,1}=\tilde{z}^{T}_{i,1}(z_{i,2}-\kappa_{1i} \hat{e}_{i,1} )=\tilde{z}^{T}_{i,1}(\tilde{z}_{i,2}-\bar{k}_{1i}\tilde{z}_{i,1} -\kappa_{1i} \hat{e}_{i,1} ). \label{S3} 
\end{equation}

\textbf{Step 2:} In this step, we analyze the error $ \tilde{z}_{i,2} $, i.e., $ z_{i,2} $ tracks the virtual control input $ z^{*}_{i,2} $. Definite two error variables 
\vspace*{-3pt} 
\begin{subequations} \label{Step2} 
	\begin{align}  
	\dot{\tilde{z}}_{i,2}&= z_{i,3}-\kappa_{2i}\hat{e}_{i,2} -\dot{z}^{*}_{i,2}, \ \dot{z}^{*}_{i,2} =- \bar{k}_{1i} \dot{\tilde{z}}_{i,1},  
	\label{S4} \\ 
	\tilde{z}_{i,3}&= z_{i,3}-z^{*}_{i,3},  \ z^{*}_{i,3}=-\bar{k}_{2i} \tilde{z}_{i,2}-\tilde{z}_{i,1}+ \dot{z}^{*}_{i,2},   \label{S5} 
	\end{align} 
\end{subequations}
where $ z^{*}_{i,3} $ is a virtual control input and $ \bar{k}_{2i}>0 $.

Choose the augmented Lyapunov candidate as $ V_{i,2}=V_{i,1}+\frac{1}{2} \tilde{z}^{T}_{i,2}\tilde{z}_{i,2} $. Then, the time derivative of $ V_{i,2} $ is derived as 
\vspace*{-3pt} 
\begin{align}
\dot{V}_{i,2}&=-\bar{k}_{1i}\tilde{z}^{T}_{i,1}\tilde{z}_{i,1}  +\tilde{z}^{T}_{i,1}\tilde{z}_{i,2}-\kappa_{1i} \tilde{z}^{T}_{i,1}\hat{e}_{i,1} + \tilde{z}^{T}_{i,2} \dot{\tilde{z}}_{i,2}   \label{S6}  \\
&= -\bar{k}_{1i}\tilde{z}^{2}_{i,1} -\bar{k}_{2i}\tilde{z}^{2}_{i,2}  +\tilde{z}^{T}_{i,2}\tilde{z}_{i,3}-\kappa_{1i} \tilde{z}^{T}_{i,1}\hat{e}_{i,1} -\kappa_{2i} \tilde{z}^{T}_{i,2}\hat{e}_{i,2}. \notag 
\end{align}

\textbf{Step} $ \mathbf{q \ (3\leq q \leq m-1)} $: Similarly, define error variables 
\vspace*{-3pt} 
\begin{subequations} \label{Step3} 
	\begin{align}  
	\hspace{-1.0em}
	\dot{\tilde{z}}_{i,s}&= z_{i,q+1}-\kappa_{qi}\hat{e}_{i,q}-\dot{z}^{*}_{i,q},   \dot{z}^{*}_{i,q}= -\sum_{s=1}^{q-1} \bar{k}_{si} \dot{\tilde{z}}^{*}_{i,s},  i \in \mathcal{V},  
	\label{S7} \\ 
	\hspace{-1.0em}
	\tilde{z}_{i,q+1}&= z_{i,q+1}-z^{*}_{i,q+1},   z^{*}_{i,q+1}=-\bar{k}_{qi} \tilde{z}_{i,q}-\tilde{z}_{i,q-1}+ \dot{z}^{*}_{i,q},   \label{S8} 
	\end{align} 
\end{subequations}
where $ z^{*}_{i,q+1} $ is a virtual control input and $ \bar{k}_{qi}>0 $. 

\vspace*{3pt}
Choose the augmented Lyapunov candidate as $ V_{i,q}=V_{i,q-1}+\frac{1}{2} \tilde{z}^{T}_{i,q}\tilde{z}_{i,q} $. Then, the time derivative of $ V_{i,q} $ can be derived as 
\vspace*{-3pt} 
\begin{equation}
\dot{V}_{i,q}=-\sum_{s=1}^{q-1} \bar{k}_{si} \tilde{z}^{T}_{i,s}\tilde{z}_{i,s} + \tilde{z}^{T}_{i,q}\tilde{z}_{i,q+1}-\sum_{s=1}^{q-1} \kappa_{si} \tilde{z}^{T}_{i,s}\hat{e}_{i,s}.   \label{S9}    
\end{equation}

\textbf{Step m:} According to the analysis in Step $ q $, we have that for this final step, the time derivative of $ \tilde{z}_{i,m} $ by (\ref{FollowerErrors}) is 
\begin{align} 
\dot{\tilde{z}}_{i,m}
& = f^{T}_{i,m} \theta_{i} + G_{i,m}(x_{i},t)u_{i} +D_{i,m}(x_{i},t) -\tilde{\eta}_{i} - \kappa_{mi}  \notag \\ 
&\ \ \times \text{sig}^{\beta} ( \sum_{j=1}^{N}a_{ij}(\tilde{x}_{j,m}-\tilde{x}_{i,m}) - a_{i0} \tilde{x}_{i,m} ) -\dot{z}^{*}_{i,m}. \label{S10} 
\end{align} 

Next, we select the augmented Lyapunov candidate as $ V_{i,m}=V_{i,m-1}+\frac{1}{2} \tilde{z}^{T}_{i,m}\tilde{z}_{i,m} $. Then, the time derivative of $ V_{i,m} $ using (\ref{S10}) can be further expressed as 
\vspace*{-3pt} 
\begin{align}
\dot{V}_{i,m}&=-\sum_{s=1}^{m-1} \bar{k}_{si} \tilde{z}^{T}_{i,s}\tilde{z}_{i,s} -\sum_{s=1}^{m-1} \kappa_{si} \tilde{z}^{T}_{i,s}\hat{e}_{i,s}  + \tilde{z}^{T}_{i,m} [f^{T}_{i,m} \theta_{i}   \label{S99}   \\
\hspace{-1.0em}
&\  + G_{i,m}(x_{i},t)u_{i} +D_{i,m}(x_{i},t) + \tilde{z}_{i,m-1} -\tilde{\eta}_{i} -\dot{z}^{*}_{i,m} ].  \notag 
\end{align}

Substituting the control law (\ref{S11}) and (\ref{S12}) into (\ref{S10}) yields 
\vspace*{-5pt} 
\begin{align} 
\dot{\tilde{z}}_{i,m}
&= f^{T}_{i,m} \tilde{\theta}_{i} - \bar{k}_{m}\tilde{z}_{i,m}  -\tilde{z}_{i,m-1} -\text{diag} \{\tilde{z}_{\delta m} \}   \hat{\varepsilon}_{i} -\tilde{\eta}_{i}  \label{S15} \\  
& \ \ +( G_{i,m}(x_{i},t)N_{i}(\kappa_{i}) +1 ) \bar{u}_{i} +D_{i,m}(x_{i},t)  - \kappa_{mi} \hat{e}_{i,m}. \notag 
\end{align} 

Select $ V_{i}=V_{i,m}+\frac{1}{2} \tilde{\theta}^{T}_{i} \varGamma^{-1}_{\theta i} \tilde{\theta}_{i} + \frac{1}{2} \tilde{\varepsilon}^{T}_{i} \varGamma^{-1}_{\varepsilon i} \tilde{\varepsilon}_{i} $ for each agent $ i $, the time derivative of $ V_{i}$ using (\ref{S99}) is written as  
\vspace*{-5pt} 
\begin{align} 
\dot{V}_{i}
& = -\sum_{s=1}^{m-1} \bar{k}_{si} \tilde{z}^{T}_{i,s}\tilde{z}_{i,s} -\sum_{s=1}^{m-1} \kappa_{si} \tilde{z}^{T}_{i,s}\hat{e}_{i,s} +   \tilde{z}^{T}_{i,m} f^{T}_{i,m} \tilde{\theta}_{i} \notag \\ 
& \ \ \  + \tilde{z}^{T}_{i,m} ( G_{i,m}(x_{i},t)N_{i}(\kappa_{i}) +1 ) \bar{u}_{i}- \tilde{\theta}^{T}_{i} f_{i,m} \tilde{z}_{i,m}  \notag \\ 
& \ \ \ + \tilde{z}^{T}_{i,m} (D_{i,m}(x_{i},t)-\text{diag} \{\tilde{z}_{\delta  m} \}   \hat{\varepsilon}_{i})  -\tilde{\varepsilon}^{T}_{i}\text{diag} \{\tilde{z}_{\delta m}\} \tilde{z}_{i,m}   \notag \\ 
& \ \ \  + \tilde{z}^{T}_{i,m}(-\bar{k}_{mi}\tilde{z}_{i,m} - \kappa_{mi} \hat{e}_{i,m} -\tilde{\eta}_{i}).      \label{S16} 
\end{align}

Notice that it follows from Lemma \ref{VaryingGain} that $ |\tilde{z}_{ij,m}|-\tilde{z}_{ij,m}   \tilde{z}_{ij,m} / \\  \sqrt{\tilde{z}^{2}_{ij,m}+\delta^{2}_{ij}(t)}  \leq \delta_{ij}(t) $, where $ \tilde{z}_{ij,m} $ denotes the $j$-th element of $ \tilde{z}_{i,m} $, $ i=1,\cdots,N, j=1,\cdots,n $. Let $\varepsilon_{ij} $ be the $ j $th element of $ \varepsilon_{i}= \text{sup}_{t\geq 0}\{\|D_{i,m}(x_{i},t)\| \} $. Then, we have   
\vspace*{-5pt} 
\begin{align} 
\dot{V}_{i}
& \leq  -\sum_{s=1}^{m-1} \bar{k}_{si} \tilde{z}^{T}_{i,s}\tilde{z}_{i,s} -\sum_{s=1}^{m-1} \kappa_{si} \tilde{z}^{T}_{i,s}\hat{e}_{i,s}- \tilde{z}^{T}_{i,m} (\bar{k}_{mi}\tilde{z}_{i,m} +  \tilde{\eta}_{i}  \notag \\ 
&\ \ \  + \kappa_{mi} \hat{e}_{i,m} ) + \sum^{n}_{j=1}  \left[  |\tilde{z}_{ij,m}|\varepsilon_{ij}- \frac{ \tilde{z}^{2}_{ij,m} \varepsilon_{ij}} { \sqrt{\tilde{z}^{2}_{ij,m}+\delta^{2}_{ij}(t)}}   \right]    \notag \\
& \ \ \ + \tilde{z}^{T}_{i,m} ( G_{i,m}(x_{i},t)N_{i}(\kappa_{i}) +1) \bar{u}_{i} \notag \\
& \leq -\sum_{s=1}^{m} \bar{k}_{si} \tilde{z}^{T}_{i,s}\tilde{z}_{i,s} -\sum_{s=1}^{m} \kappa_{si} \tilde{z}^{T}_{i,s}\hat{e}_{i,s} + \sum^{n}_{j=1} \delta_{ij}(t) \varepsilon_{ij}  \notag \\  
& \ \ \ - \tilde{z}^{T}_{i,m}\tilde{\eta}_{i}  +  G_{i,m}(x_{i},t)(N_{i}(\kappa_{i}) +1)  \dot{\kappa}_{i}/k_{\kappa i}.      \label{S17} 
\end{align}

\vspace*{-5pt}
Thus, the overall Lyapunov function candidate is selected as $ V=\sum_{i=1}^{N} V_{i}$. Then, its time derivative is described by 
\vspace*{-3pt} 
\begin{align} 	
\dot{V}
& \leq -\sum_{i=1}^{N}  ( \sum_{s=1}^{m} \bar{k}_{si} \tilde{z}^{T}_{i,s}\tilde{z}_{i,s}+\sum_{s=1}^{m} \kappa_{si} \tilde{z}^{T}_{i,s}\hat{e}_{i,s} + \tilde{z}^{T}_{i,m}\tilde{\eta}_{i} )   \label{S18}  \\  
& \ \ \ + \sum_{i=1}^{N}    \sum^{n}_{j=1}  \delta_{ij}(t) \varepsilon_{ij}   + \sum_{i=1}^{N} \frac{1}{k_{\kappa i}} G_{i,m}(x_{i},t)(N_{i}(\kappa_{i}) +1)  \dot{\kappa}_{i} .   \notag    
\end{align}

\vspace*{-5pt}
Then, integrating both sides of (\ref{S18}) gives rise to  $ V(t)\leq  V(0) - \int_{0}^{t}W(s)ds + I_{A}+I_{B} $, 
where $ W(t) = \sum_{i=1}^{N} ( \sum_{s=1}^{m} \bar{k}_{si} \tilde{z}^{T}_{i,s}\tilde{z}_{i,s} +\sum_{s=1}^{m} \kappa_{si} \tilde{z}^{T}_{i,s}\hat{e}_{i,s} + \tilde{z}^{T}_{i,m}\tilde{\eta}_{i} )$, $ I_{A}=\int_{0}^{t} \sum_{i=1}^{N} \sum^{n}_{j=1} \delta_{ij}(\omega) |\varepsilon_{ij}|d\omega \leq  \sum_{i=1}^{N} \sum^{n}_{j=1}  \delta^{*}_{ij} |\varepsilon_{ij}| < \delta^{*} $ for certain positive scalarss $ \delta^{*} $, and  $ I_{B}= \int_{0}^{t} \sum_{i=1}^{N} k^{-1}_{\kappa i} |(G_{i,m}N_{i}  (\kappa_{i} (\omega))  +1 )  \dot{\kappa}_{i}(\omega)|d\omega $ can be bounded by seeking a contradiction under Assumptions \ref{FollwerGains} and \ref{FaultAssumption} with a similar spirit of arguments in \cite{Wen98TAC}. 
Then, $ \int_{0}^{t}W(s)ds $ can be upper bounded. Thus, the existence of $ \lim_{t\rightarrow \infty} \int_{0}^{t}W(s)ds$ can be guaranteed and it is finite. According to Theorem 1 and the input-to-state stability theory, we can obtain that $ \hat{e}_{i,s}, \tilde{\eta}_{i} $ are bounded and moreover, all signals in $ V_{i} $ are bounded. Hence, the linear analysis can be applied to show that $ \tilde{z}_{i,s}, s=1,\cdots,m $, and $ \dot{\tilde{z}}_{i,m} $ are bounded. Since $ \int_{0}^{t}W(s)ds $ is bounded, we further get  that  $ \lim_{t\rightarrow \infty} \int_{0}^{t} \sum_{i=1}^{N} \sum_{s=1}^{m} \bar{k}_{si} \tilde{z}^{T}_{i,s}(\tau)\tilde{z}_{i,s}(\tau) d\tau $ can be upper bounded as $ \int_{0}^{t} \sum_{i=1}^{N} (\sum_{s=1}^{m} \tilde{z}^{T}_{i,s}(\tau)\hat{e}_{i,s}(\tau) + \tilde{z}^{T}_{i,m}(\tau)\tilde{\eta}_{i}(\tau)) d\tau$ are bounded due to $ \lim_{t\rightarrow T_{2}} \tilde{\eta}_{i} =0_{n} $ and  $  \lim_{t\rightarrow T_{3}} \hat{e}_{i,s} =0_{n} $. Thus, $ \tilde{z}_{i,s} $ is square integrable. The Barbalat's Lemma in Lemma \ref{Barbalat} is used to conclude  $  \lim_{t\to \infty} \tilde{z}_{i,s}=0_{n}$. That is, $ \lim_{t\to \infty} z_{i,s} =0_{n}$ via the defined virtual control inputs. Hence, $ \underset{t\rightarrow \infty}{\text{lim}}  
(y_{i}(t)-y_{0}(t)) = \varDelta_{i} \ \text{and} \ \underset{t\rightarrow \infty }{\text{lim}}  (x_{i,k}(t)-x_{0,k}(t)) = 0_{n}$, $k=2,\cdots,m $.
\end{proof}

\section{Application to Task-space Cooperative Tracking}
\vspace*{-2pt}
\subsection{Networked Manipulator Model}
\vspace*{-3pt} 
Consider a group of $N$ manipulators, where the kinematics and dynamics of each manipulator $ i \in \mathcal{V} $ are governed by 
\vspace*{-4pt}
\begin{align}  
&x_{i}=S_{i}(q_{i}), \ \dot{x}_{i}=J_{i}(q_{i}) \dot{q}_{i}, \ i=1,2,\cdots,N,   \label{M1} \\
&M_{i}(q_{i})\ddot{q}_{i}+C_{i}(q_{i},\dot{q}_{i})\dot{q}%
_{i}+G_{i}(q_{i})+F_{i}(q_{i},\dot{q}_{i})=g_{i}(t)\tau_{ai}+d_{i}(t),   \notag 
\end{align}
where $x_{i} \in \mathbb{R}^{n} $ is a generalized end-effector configuration, $ S_{i} (q_{i}) \\ :  \mathbb{R}^{l} \rightarrow \mathbb{R}^{n} $ is a nonlinear mapping from the joint space to the task space, and $ J_{i}(q_{i})=\partial{S_{i}(q_{i})}/\partial{q_{i}} \in \mathbb{R}^{n\times l} $ is its Jacobian matrix, $ q_{i}$ and $\dot{q}_{i} \in \mathbb{R}^{l} $ denote generalized position and velocity vectors, respectively, $M_{i}(q_{i})$ $\in $ $\mathbb{R}^{l\times l}$ is an inertia matrix, $C_{i}(q_{i},\dot{q}_{i})\dot{q}_{i} \in \mathbb{R}^{l} \\ $ is a Coriolis  centrifugal force vector, $G_{i}(q_{i})\in \mathbb{R}^{l}$ is a gravity vector,  $F_{i}(q_{i},\dot{q}_{i}) \in \mathbb{R}^{l}$ are uncertain dynamics 
(e.g., $ F_{i}(q_{i},\dot{q}_{i})=F_{vi}\tanh (\dot{q}_{i}) + F_{ci}\text{sgn}(\dot{q}_{i})$ for matrices $F_{vi}$ and  $F_{ci} $), $\tau _{ai} \in \mathbb{R}^{l}$ are control torques with faults, $d_{i} \in \mathbb{R}^{l}$ are external disturbances, and $ g_{i}(t) \neq 0$ are unknown control coefficients relating $ \tau_{ai} $ to  torques $ \tau_{i} $ to be designed. The sign of $ g_{i}(t) $ is unknown. 

\vspace*{2pt}
Similar to (\ref{Fault}), two types of actuator faults are described by  
\vspace*{-3pt} 
\begin{equation}
\tau_{ai}= \phi_{i}(t) \tau_{i} + \psi_{i}(t), \ i=1,2,\cdots,N,  \label{FaultManipulator}
\end{equation}%
where $ 0<\phi_{i}(t)\leq 1  $ and $ \psi_{i}(t) \in \mathbb{R}^{l} $ represent the actuation loss of effectiveness fault and the float fault, respectively. 




\vspace*{1pt}
\textit{Property 1:}\label{M_proverty} $M_{i}(q_{i})$ is symmetric and positive definite.

\textit{Property 2:}\label{Skewsymmetric} $  \dot{M}_{i}(q_{i})-2C_{i}(q_{i},\dot{q}_{i}) $ is skew symmetric.

\textit{Property 3:}\label{Adaptive} 
For $x, y\in \mathbb{R}^{n}$, $M_{i}(q_{i})y+C_{i}(q_{i},\dot{q}_{i})x+G_{i}(q_{i})=Y_{i}(q_{i},\dot{q}_{i},x,y)\theta_{i}$, where $Y_{i}(\cdot)\in \mathbb{R}^{l\times p}$ denotes a known dynamic regression matrix, and $\theta_{i} \in \mathbb{R}^{p}$ is an unknown parameter vector.

\textit{Property 4:}\label{TaskAdaptive} The kinematics (\ref{M1}) relies linearly on a kinematic parameter vector $ a_{i}\in \mathbb{R}^{r} $, i.e., $ \dot{x}_{i}=J_{i}(q_{i}) \dot{q}_{i}= Z_{i}(q_{i}, \dot{q}_{i}) a_{i} $, where $ Z_{i}(q_{i}, \dot{q}_{i}) \in \mathbb{R}^{n \times r} $  is a known kinematic regression matrix.



\vspace*{2pt} 
\begin{problem} \label{Problemm}
Consider the robots' dynamics and kinematics in (\ref{M1}) with actuator faults in (\ref{FaultManipulator}). Given a directed graph $\mathcal{\bar{G}}$, design a distributed controller $ \tau_{i} $ so that each robot achieves   
\vspace*{-3pt} 
\begin{equation}
\underset{t\rightarrow \infty}{\text{lim}}  
(x_{i}-x_{d}) = 0_{n} \ \text{and} \ \underset{t\rightarrow \infty }{\text{lim}}  (\dot{x}_{i}-\dot{x}_{d}) = 0_{n}, \ i\in \mathcal{V}, \label{M3}
\end{equation}%
where $ x_{d} $ is a desired global task reference and $ \dot{x}_{d} $ is its velocity. 
\end{problem}  

\vspace*{-8pt} 
\subsection{Task-space Coordinated Tracking of Networked Manipulators} 
\vspace*{-2pt}
Unlike works in \cite{Zhang17Tcyber,WangAT,LiuTRO,WangTAC,LiangTcyber} requiring available task information ($ x_{d}$, $\dot{x}_{d} $, and $ \ddot{x}_{d}$) to all robots,  we present a distributed estimation framework to reconstruct this global information for each robot. 
 
\vspace*{2pt} 
\textit{\underline{\textbf{Distributed Nonlinear Estimator}}}: similar to (\ref{Estimator}), the following finite-time distributed estimator is developed for each robot $ i $
\vspace*{-8pt}
\begin{subequations}\label{EEstimator}	
	\begin{align}  
	\hspace{-0.8em}	\dot{\chi}_{i}&=\vartheta_{i}+k_{\chi i} \text{sig}^{\gamma}  (  \sum_{j=1}^{N}a_{ij}(\chi_{j}-\chi_{i}) + a_{i0}(x_{d}-\chi_{i})  ) ,   \label{E0} \\ 
	\hspace{-0.8em} \dot{\vartheta}_{i}&=\eta_{i}+k_{\vartheta i} \text{sig}^{\beta}  (  \sum_{j=1}^{N}a_{ij}(\vartheta_{j}-\vartheta_{i}) + a_{i0}(\dot{x}_{d}-\vartheta_{i})  ) ,   \label{EEa} \\ 
	\hspace{-0.8em} 
	\dot{\eta}_{i}&=k_{\eta i} [\text{sig}^{\alpha}  (e^{\eta}_{i} ) + \text{sgn} ( e^{\eta}_{i}) ],  e^{\eta}_{i}= \sum_{j=1}^{N}a_{ij}(\eta_{j}-\eta_{i}) + e^{\varrho}_{i},  \label{EEb} \\
	\hspace{-0.8em} \dot{\xi}_{i}&=\varrho_{i}+k_{\xi i}a_{i0}\text{sig}^{\frac{1}{2}} \left( \dot{x}_{d}-\xi_{i} \right) , \ e^{\varrho}_{i}= a_{i0}(\varrho_{i}-\eta_{i}), \label{FFa} \\  
	\hspace{-0.8em} 
	\dot{\varrho}_{i}&=k_{\varrho i}a_{i0}\text{sgn} \left( \dot{x}_{d}-\xi_{i} \right) ,  \xi_{i}(0)=0_{n},  \varrho_{i}(0)=0_{n}, \label{FFb}
	\end{align}
\end{subequations}
where $ k_{\chi i}, k_{\vartheta i}, k_{\eta i}, k_{\xi i}, k_{\varrho i}>0$, $\alpha, \beta, \gamma \in (0.5,1)$, $ \chi_{i}$,  $\vartheta_{i}$, $\eta_{i} $ are the estimates of $ x_{d} $, $ \dot{x}_{d} $, $ \ddot{x}_{d} $, respectively. 

Define $\tilde{\chi}_{i} =\chi_{i}-x_{d}$, $\tilde{\vartheta}_{i} =\vartheta_{i}-\dot{x}_{d}$, $\tilde{\eta}_{i}=\eta_{i}-\ddot{x}_{d}$, $\tilde{\xi}_{i}=\xi_{i}-\dot{x}_{d}$, and $\tilde{\varrho}_{i}=\varrho_{i}-\ddot{x}_{d}$. Then, we get the estimation result below.

\vspace*{1pt}
\begin{theorem} \label{theorem3}
	Under the proposed distributed estimator in (\ref{EEstimator}), all these state estimates are uniformly bounded, and the finite-time estimation is achieved: 
	$\lim_{t\to T_{1}} a_{i0}\tilde{\varrho}_{i}=0_{n}$,  $\lim_{t\to T_{2}} \tilde{\eta}_{i}=0_{n}$, $\lim_{t\to T_{3}} \tilde{\vartheta}_{i}=0_{n}$, $\lim_{t\to T_{4}} \tilde{\chi}_{i}=0_{n}$ for $ T_{i}, i=1,2,3,4 $.
\end{theorem}

\begin{proof}
	It is similar to that in Theorem 1 and is omitted.	
\end{proof}


\vspace*{5pt}
Next, we can transform the coordinated tracking problem into the simultaneous tracking problem for the decoupled robot group. Specifically, denote the position tracking error $ \bar{x}_{i}=x_{i}-x_{d}$ and velocity tracking error $  \dot{\bar{x}}_{i}=\dot{x}_{i}-\dot{x}_{d}$, $  i \in \mathcal{V} $. Then, based on the estimations $ \chi_{i} $ and $ \vartheta_{i} $ obtained from the distributed estimator in (\ref{EEstimator}), we have  $
\bar{x}_{i}= x_{i} - \chi_{i} + \tilde{\chi}_{i}$ and $
\dot{\bar{x}}_{i}= \dot{x}_{i} - \vartheta_{i} + \tilde{\vartheta}_{i}$ where $ \tilde{\chi}_{i} =\chi_{i}-x_{d}  $ and $ \tilde{\vartheta}_{i} =\vartheta_{i}-\dot{x}_{d}  $. As it follows from Theorem \ref{theorem3} that $\lim_{t\to \infty} \tilde{\chi}_{i}=0_{n}$ and $\lim_{t\to \infty} \tilde{\vartheta}_{i}=0_{n}$, $\forall i\in \mathcal{V} $, the coordinated tracking control objective in (\ref{M3}) can be transformed into the simultaneous tracking
objective in the sense that 
\vspace*{-4pt}
\begin{equation}
\lim_{t\to \infty} (x_{i} - \chi_{i}) =0_{n} \ \text{and} \ \lim_{t\to \infty} ( \dot{x}_{i} - \vartheta_{i}) =0_{n}. \label{transform}
\end{equation}

First, let us define a task-space sliding variable $ s_{xi} $ 
\vspace*{-3pt}
\begin{equation}
s_{xi}=e_{vi} + \alpha_{xi}e_{xi}, \ e_{xi}= x_{i}-\chi_{i},  
\	e_{vi}= \dot{x}_{i}-\vartheta_{i}, \ i\in \mathcal{V}, \label{SlidingError}
\end{equation}
where $ e_{xi}$ and $ e_{vi} $ are the task-space position and velocity tracking errors of the $ i $th manipulator and $ \alpha_{xi}>0 $ is a scalar. 

In the presence of uncertain kinematics, $ J_{i}(q_{i}) $ becomes unknown and satisfies Property 4. Using the estimate of $ J_{i}(q_{i}) $ and the sliding vector $ s_{xi} $, we define a joint-space reference velocity 
\vspace*{-5pt}
\begin{equation}
\dot{q}_{ri}=\hat{J}^{+}_{i}(q_{i})  ( \vartheta_{i} -\alpha_{xi}e_{xi}-\alpha_{ri} \int_{0}^{t}s_{xi}(\omega)d\omega ), \ i\in \mathcal{V}, \label{SS2}
\end{equation}
where $ \hat{J}^{+}_{i}(q_{i}) =\hat{J}^{T}_{i}(q_{i}) (\hat{J}_{i}(q_{i}) \hat{J}^{T}_{i}(q_{i}) )^{-1} $ is a generalized inverse of the approximate Jacobian matrix, $ \alpha_{ri}>0 $ is a scalar, $ \vartheta_{i} $ is obtained from  (\ref{EEstimator}), and $ \hat{J}_{i}(q_{i}) $ is the estimate of $ J_{i}(q_{i}) $, which is obtained by replacing $a_{i} $ in $ J_{i}(q_{i}) $ with $\hat{a}_{i} $ that defines the estimate of the unknown  kinematic parameter $a_{i} $ in Property 4.   


Differentiating (\ref{SS2}) gives the joint-space reference acceleration
\vspace*{-5pt}
\begin{equation}
\ddot{q}_{ri}=\hat{J}^{+}_{i}(q_{i})  ( \dot{\vartheta}_{i} -\alpha_{xi} \dot{e}_{xi} -\alpha_{ri} s_{xi}) + \dot{\hat{J}}^{+}_{i}(q_{i}) \hat{J}_{i}(q_{i})  \dot{q}_{ri} .  \label{SS3}
\end{equation}

Based on the reference velocity given in (\ref{SS2}), we define a joint-space sliding vector $ s_{i}=\dot{q}_{i}-\dot{q}_{ri} $ and an estimated end-effector velocity $ \dot{\hat{x}}_{i} =\hat{J}_{i}(q_{i})\dot{q}_{i} $. Then, it follows from (\ref{SS2}) that  
\vspace*{-2pt}
\begin{align}
s_{i}&=\dot{q}_{i}-\dot{q}_{ri} =\hat{J}^{+}_{i}(q_{i})  ( \dot{\hat{x}}_{i} -\dot{x}_{i} + \dot{x}_{i}  ) -\dot{q}_{ri}  \label{SS5}  \\
&= \hat{J}^{+}_{i}(q_{i})  [ (\hat{J}_{i}(q_{i})-J_{i}(q_{i}))\dot{q}_{i} + s_{xi}
+ \alpha_{ri} \int_{0}^{t}s_{xi}(\omega)d\omega ],  \notag 
\end{align}
which yields the following relation between the joint-space sliding vector $ s_{i} $ and the task-space sliding vector $ s_{xi} $ 
\vspace*{-3pt}
\begin{equation}
\hat{J}_{i}(q_{i})s_{i}=s_{xi}+ \alpha_{ri} \int_{0}^{t}s_{xi}(\omega)d\omega +Z_{i}(q_{i},\dot{q}_{i}) \tilde{a}_{i},  \ i \in \mathcal{V},  \label{SS6} 
\end{equation}
where $ \tilde{a}_{i} = \hat{a}_{i} -a_{i} $ is the kinematic parameter estimation error.

\vspace*{2pt}
By Property 3, $ M_{i}(q_{i})\ddot{q}_{ri} +C_{i}(q_{i},\dot{q}_{i})\dot{q}_{ri}+G_{i}(q_{i})$$=Y_{i}\theta_{i} $. 
Then, for $ b_{i}(t)=g_{i}\phi_{i} \in \mathbb{R} $ and $ D_{i}(t)=g_{i}\psi_{i}-F_{i}(q_{i},\dot{q}_{i}) +d_{i} \in \mathbb{R}^{l}$,   
\vspace*{-3pt}
\begin{equation}
M_{i}(q_{i})\dot{s}_{i}+C_{i}(q_{i},\dot{q}_{i})s_{i}=b_{i}(t)\tau_{i}+D_{i}(t)  -Y_{i}\theta_{i}. \label{OpenLoop}
\end{equation}

Similarly, let $ \varepsilon_{i}= \text{sup}_{t\geq 0}\{\|D_{i}(t)\| \} $ and $ \hat{\varepsilon}_{i} $ denotes an estimate of this unknown bound vector $ \varepsilon_{i} 1_{n} $. Then, $ \tilde{\varepsilon}_{i} =\varepsilon_{i} 1_{n}-\hat{\varepsilon}_{i} $.

\textbf{\underline{\textit{Fault-Tolerant Distributed Adaptive Controller}}}: in light of the Nussbaum function and local estimates $ \chi_{i} $ and $ \vartheta_{i} $ from (\ref{EEstimator}), a fault-tolerant adaptive controller is proposed as 
\vspace*{-2pt}
\begin{subequations}\label{CController}
\begin{align}
\hspace*{-0.9em}
\tau_{i}&= N_{i}(\kappa_{i})u_{i}, \   N_{i}(\kappa_{i})=\text{exp}(\kappa^{2}_{i})\text{cos}((\pi/2)\kappa_{i})+1, \label{B7}  \\
\hspace*{-0.9em}
u_{i}&=  Y_{i}\hat{\theta}_{i} -\hat{J}^{T}_{i} K_{si}  \hat{J}_{i}s_{i} - \text{diag} \{s_{\delta ij} \}   \hat{\varepsilon}_{i}, \ \dot{\kappa}_{i}=-k_{\kappa i} s^{T}_{i} u_{i},  \label{B77}  \\
\hspace*{-0.9em}
\dot{\hat{\varepsilon}}_{i}&= \varGamma_{\varepsilon i} \text{diag} \{s_{\delta ij} \}s_{i}, \ s_{\delta ij} = s_{ij}/ \sqrt{s^{2}_{ij}+\delta^{2}_{ij}(t)}, \ i \in \mathcal{V},   \label{B777}
\end{align} 
\end{subequations}
where 
$ k_{\kappa i}>0 $ is a constant, and $ K_{si}, \varGamma_{\varepsilon i}>0$ are matrices.  Then, for matrices $ \varGamma_{\theta i}, \varLambda_{i}>0 $, the dynamic  adaptation law for $ \hat{\theta}_{i} $ and the kinematic adaptation law for $ \hat{a}_{i} $  are updated by  
\vspace*{-2pt}
\begin{subequations}\label{KDAdaptiveLaw}
\begin{align}
\dot{\hat{\theta}}_{i}&= -\varGamma_{\theta i}   Y^{T}_{i}(q_{i},\dot{q}_{i},\dot{q}_{ri},\ddot{q}_{ri})s_{i}, \ i \in \mathcal{V}, \label{AdaptiveLawY} \\
\dot{\hat{a}}_{i}&= \varLambda_{i} Z^{T}_{i}(q_{i},\dot{q}_{i}) (s_{xi}+ \alpha_{ri} \int_{0}^{t}s_{xi}(\omega)d\omega -\hat{J}_{i}s_{i}) . \label{AdaptiveLawZ2}
\end{align} 
\end{subequations}


Let $ \tilde{\theta}_{i}=\hat{\theta}_{i}- \theta_{i}$. Then, substituting (\ref{CController}) into (\ref{OpenLoop}) yields 
\begin{align}
\hspace*{-0.9em}
M_{i}(q_{i})\dot{s}_{i}&=-C_{i}(q_{i},\dot{q}_{i})s_{i} + (b_{i}(t) N_{i}(\kappa_{i})-1)u_{i} +Y_{i} \tilde{\theta}_{i}  \notag \\
& \ \ \ -\hat{J}^{T}_{i}K_{s i} \hat{J}_{i}s_{i} +D_{i}(t)- \text{diag} \{ s_{\delta ij} \} \hat{\varepsilon}_{i}. \label{ClosedLoop}
\end{align}

The combination of (\ref{SlidingError}), (\ref{SS6}) and (\ref{ClosedLoop}) yields a cascade system 
\vspace*{-5pt}
\begin{equation} \label{CascadeSystem}
\left\{ 
\begin{array}{l}
\hspace{-0.75em} \ \  s_{xi}=e_{vi} + \alpha_{xi}e_{xi}, \ i \in \mathcal{V},  \\ 
\hspace{-0.75em} \ \hat{J}_{i}s_{i}=s_{xi}+ \alpha_{ri} \int_{0}^{t}s_{xi}(\omega)d\omega +Z_{i}(q_{i},\dot{q}_{i}) \tilde{a}_{i},    \\
\hspace{-0.75em} M_{i}\dot{s}_{i}=-C_{i}s_{i} + (b_{i}(t) N_{i}(\kappa_{i})-1)u_{i} -\hat{J}^{T}_{i}K_{s i} \hat{J}_{i}s_{i}   \\ 
\hspace{-0.72em}  \ \ \ \ \ \ \ \ \ \ +Y_{i}\tilde{\theta}_{i}  + D_{i}(t) - \text{diag} \{ s_{\delta ij}\} \hat{\varepsilon}_{i}.
\end{array}
\right.
\end{equation} 

Now, we have the fault-tolerant task-space coordination result.

\vspace*{2pt} 
\begin{theorem}
Consider a group of networked manipulators subject to faults in (\ref{FaultManipulator}). Under the proposed distributed controller in (\ref{CController}) with the  nonlinear estimator in (\ref{EEstimator}), Problem \ref{Problemm} is solvable, i.e., $ \underset{t\rightarrow \infty }{\text{lim}}  (x_{i}-x_{d})=0_{n}$ and $ \underset{t\rightarrow \infty }{\text{lim}}  (\dot{x}_{i}-\dot{x}_{d})=0_{n}$. 
\end{theorem}  

\vspace*{3pt} 
\begin{proof} 
	The proof includes three steps: 
	
	\textbf{Step (i):} prove $ \hat{J}_{i}s_{i} \in \mathcal{L}_{2} $ 
	and $\underset{t\rightarrow \infty }{\text{lim}} \hat{J}_{i}s_{i} =0_{n}$, $ i\in \mathcal{V} $.
	
Construct the following Lyapunov function candidate: 
	\vspace*{-3pt}
	\begin{equation}
	\hspace{-0.69em}
	V_{si}=\frac{1}{2} \left(  s_{i}^{T}M_{i}(q_{i})s_{i}+   \tilde{\theta}_{i}^{T} \varGamma^{-1}_{\theta i} \tilde{\theta}_{i}   +\tilde{\varepsilon}^{T}_{i} \varGamma^{-1}_{\varepsilon i} \tilde{\varepsilon}_{i} \right).  \label{CS5}
	\end{equation}
	
	Then, the time derivative of $ V_{si} $ along 
	(\ref{CascadeSystem}) is given by
	\vspace*{-2pt}
	\begin{align}
	\dot{V}_{si}&=   s^{T}_{i}M_{i}(q_{i})\dot{s}_{i}+\frac{1}{2}s^{T}_{i}\dot{M}_{i}(q_{i})s_{i}+ \tilde{\theta}^{T}_{i} \varGamma^{-1}_{\theta i}  \dot{\tilde{\theta}}_{i}   +\tilde{\varepsilon}^{T}_{i}\varGamma^{-1}_{\varepsilon i} \dot{\tilde{\varepsilon}}_{i}    \notag \\
	&= s^{T}_{i}  (  Y_{i}\tilde{\theta}_{i}-\hat{J}^{T}_{i}  K_{s i} \hat{J}_{i} s_{i} )  +  s^{T}_{i}  [ \frac{1}{2} \dot{M}_{i}(q_{i})- C_{i}(q_{i},\dot{q}_{i}) ]  s_{i}  \notag \\
	& \ \ \  +s^{T}_{i} [ D_{i}(t)- \text{diag} \{ s_{\delta ij} \}  \hat{\varepsilon}_{i} ] -\tilde{\varepsilon}^{T}_{i} \text{diag} \{ s_{\delta ij} \} s_{i} \notag \\
	&\ \ \ - \tilde{\theta}^{T}_{i} Y^{T}_{i} s_{i}  +  s^{T}_{i} \left( b_{ i}(t)N(\kappa_{i}) -1 \right) u_{i}. \label{CS6} 
	\end{align}   
	
Since $ \sum^{l}_{j=1} ( |s_{ij}|\varepsilon_{ij}-s_{ij} s_{ij} \varepsilon_{ij} / \sqrt{s^{2}_{ij}+\delta^{2}_{ij}(t)} ) \leq  \sum^{l}_{j=1} \delta_{ij}(t) \\ \varepsilon_{ij}$ 
	by Lemma \ref{VaryingGain}, we have $	\dot{V}_{si} \leq  -\frac{1}{k_{\kappa i}}   ( b_{i}(t)  N_{i} (\kappa_{i}(t))$ $  -1 ) \dot{\kappa}_{i}(t) - s^{T}_{i} \hat{J}^{T}_{i}  K_{s i} \hat{J}_{i} s_{i} +\sum^{l}_{j=1} \delta_{ij}(t) \varepsilon_{ij}$.
Then, integrating this inequality yields $ V_{si}(t) \leq V_{si}(0) - \int_{0}^{t} s^{T}_{i}(\omega)  \hat{J}^{T}_{i}(\omega)  K_{s i}(\omega) \hat{J}_{i}(\omega) s_{i}(\omega) d\omega + I^{'}_{A} +I^{'}_{B}, $
where $ I^{'}_{A}=\int_{0}^{t} k^{-1}_{\kappa i} |(1-b_{i}(\omega) N_{i}(\kappa_{i}(\omega))) \dot{\kappa}_{i}(\omega)| d\omega$ and $ I^{'}_{B}= \int_{0}^{t} \sum^{l}_{j=1} \delta_{ij}(\omega) |\varepsilon_{ij}| d\omega$. Since $ \int_{0}^{t} \delta_{ij}(\omega) d\omega \leq \delta^{*}_{ij} $, it is not difficult to obtain that $ I^{'}_{B} \leq \sum^{l}_{j=1}\delta^{*}_{ij}  |\varepsilon_{ij} | $, which can be upper bounded by certain constant. Moreover, with the similar spirit of the argument in \cite{Wen98TAC}, the boundness of $ I^{'}_{A} $ can be obtained by seeking a contradiction. Thus, it can be concluded from (\ref{CS5}) that  $ s_{i} \in \mathcal{L}_{\infty} $, $ \tilde{\theta}_{i} \in \mathcal{L}_{\infty} $, and $ \tilde{\varepsilon}_{i} \in \mathcal{L}_{\infty} $. 
	
	In addition, define $ W_{si}=s^{T}_{i} \hat{J}^{T}_{i}  K_{s i}\hat{J}_{i}s_{i} $ and we can obtain 
	\vspace*{-4pt}
	\begin{equation}
	\int_{0}^{t} W_{si}(\omega) d \omega  \leq  V_{si}(0) -V_{si}(t) + I^{'}_{A} + I^{'}_{B}, \ i \in \mathcal{V}, \label{CS9}
	\end{equation}
	which means that $ \int_{0}^{t} W_{si}(\omega) d \omega $ can be upper bounded and the existence of $ \underset{t\rightarrow \infty }{\text{lim}} \int_{0}^{t} W_{si}(\omega) d \omega $ can be guaranteed and it is finite. Then, according to the definition of $ W_{si} $, we have that $\hat{J}_{i}s_{i}$ is square integrable. Hence, the Barbalat's Lemma in Lemma \ref{Barbalat} is used to conclude that $ \hat{J}_{i}s_{i} \in \mathcal{L}_{2} $ 
	and $\underset{t\rightarrow \infty }{\text{lim}} \hat{J}_{i}s_{i} =0_{n}$.  
	
	\textbf{Step (ii): } prove $\underset{t\rightarrow \infty }{\text{lim}} (Z_{i}(q_{i},\dot{q}_{i}) \tilde{a}_{i}) =0_{n}$, $\underset{t\rightarrow \infty }{\text{lim}} e_{xi}=\underset{t\rightarrow \infty }{\text{lim}} (x_{i}-\chi_{i}) =0_{n}$, and  $\underset{t\rightarrow \infty }{\text{lim}} e_{vi}=\underset{t\rightarrow \infty }{\text{lim}} (\dot{x}_{i}-\vartheta_{i})=0_{n}$. 
	
	\vspace*{2pt}
	Since $ \hat{J}_{i}s_{i} \in \mathcal{L}_{2} $ by Step (i), there exists a constant $ I_{si} $ so that $ I_{si}= 2\int_{0}^{\infty} s^{T}_{i}(\omega) \hat{J}^{T}_{i}(\omega) \hat{J}_{i}(\omega) s_{i}(\omega) d\omega < \infty $, $ i\in \mathcal{V} $. Then, for the second subsystem in (\ref{CascadeSystem}), we can select  
	\vspace*{-4pt}
	\begin{align}
	\hspace{-0.8em}
	V_{zi}&= \frac{1}{2}e^{T}_{xi}e_{xi} + ( I_{si}- 2\int_{0}^{t} s^{T}_{i}(\omega) \hat{J}^{T}_{i}(\omega)  \hat{J}_{i}(\omega) s_{i}(\omega) d\omega )   \notag \\
	&+2\tilde{a}^{T}_{i} \Lambda^{-1}_{i} \tilde{a}_{i}+ \alpha_{ri} ( \int_{0}^{t}s_{xi}(\omega) d\omega  )^{T}  ( \int_{0}^{t}s_{xi}(\omega) d\omega  ).  \label{CS10} 
	\end{align}
	
	\vspace*{-3pt} 
	By using $ \hat{J}_{i}s_{i}=s_{xi}+ \alpha_{ri} \int_{0}^{t}s_{xi}(\omega)d\omega +Z_{i}(q_{i},\dot{q}_{i}) \tilde{a}_{i} $ from the second subsystem of (\ref{CascadeSystem}), the time derivative of $ V_{zi} $ is
	\vspace*{-5pt}
	\begin{align}
	\hspace{-1.2em} 
	\dot{V}_{zi}
	&=-\alpha_{xi}e^{T}_{xi}e_{xi} +e^{T}_{xi}  ( \hat{J}_{i}s_{i}- \alpha_{ri} \int_{0}^{t}s_{xi}(\omega)d\omega -Z_{i}(q_{i},\dot{q}_{i}) \tilde{a}_{i} )   \notag \\
	& \ \ \ - 2 (s_{xi}+ \alpha_{ri} \int_{0}^{t}s_{xi} d\omega )^{T}  (s_{xi}+ \alpha_{ri} \int_{0}^{t}s_{xi} d\omega  )  \notag \\
	& \ \ \ - 2 \tilde{a}^{T}_{i} Z^{T}_{i}(q_{i},\dot{q}_{i}) Z_{i}(q_{i},\dot{q}_{i}) \tilde{a}_{i} - 2\alpha_{ri} s^{T}_{xi} ( \int_{0}^{t}s_{xi}(\omega) d\omega ). \label{CS11}
	\end{align}
	
	\vspace*{-2pt}
	Let $ V_{szi}= V_{si}+V_{zi}$. Then, using Young inequality and (\ref{CS11}),
	\vspace*{-4pt}
	\begin{align}
	\dot{V}_{szi}&\leq -(\alpha_{xi}-\frac{3}{4})e^{T}_{xi}e_{xi} - s^{T}_{i} \hat{J}^{T}_{i}  (K_{si} -I_{l})\hat{J}_{i} s_{i} -s^{T}_{xi}s_{xi} \notag \\
	&\ \ \  -\left(s_{xi}+ \alpha_{ri} \int_{0}^{t}s_{xi} d\omega\right)^{T} \left(s_{xi}+ \alpha_{ri} \int_{0}^{t}s_{xi} d\omega \right)  \notag \\
	&\ \ \  +\sum^{l}_{j=1} \delta_{ij}(t) \varepsilon_{ij}- \frac{1}{k_{\kappa i}} (b_{i}(t) N_{i}(\kappa_{i}) -1) \dot{\kappa}_{i} \notag \\ 
	&\ \ \  -2 \tilde{a}^{T}_{i} Z^{T}_{i}(q_{i},\dot{q}_{i}) Z_{i}(q_{i},\dot{q}_{i}) \tilde{a}_{i}.  \label{CS12}  
	\end{align}
	
	\vspace*{-3pt}
	Following the similar analysis in (\ref{CS9}),   $ \underset{t\rightarrow \infty }{\text{lim}} Z_{i}(q_{i}, \dot{q}_{i}) \tilde{a}_{i}=0_{n} $,  $ \underset{t\rightarrow \infty }{\text{lim}} \hat{J}_{i}(q_{i})s_{i}=0_{n} $, and $ \underset{t\rightarrow \infty }{\text{lim}} e_{xi}=0_{n} $, provided that $ \alpha_{xi}>\frac{3}{4} $ and $ K_{si}>I_{l} $. 
	By (\ref{CascadeSystem}), $ s_{xi}+ \alpha_{ri} \int_{0}^{t}s_{xi}(\omega)d\omega=0_{n} $. From the input-output property of exponentially stable and strictly proper linear systems, $ \underset{t\rightarrow \infty }{\text{lim}} \int_{0}^{t}s_{xi}(\omega) d\omega =0_{n} $ and $ \underset{t\rightarrow \infty }{\text{lim}} s_{xi}=0_{n} $. By Step (ii), $ \underset{t\rightarrow \infty }{\text{lim}} e_{xi}=0_{n} $. Hence, $ \underset{t\rightarrow \infty }{\text{lim}} e_{vi}=0_{n} $. 
	
	\vspace{3pt}
	\textbf{Step (iii):} prove  $ \underset{t\rightarrow \infty }{\text{lim}}  (x_{i}-x_{d})=0_{n}$ and $ \underset{t\rightarrow \infty }{\text{lim}}  (\dot{x}_{i}-\dot{x}_{d})=0_{n}$. 
	
	Based on Theorem \ref{theorem3}, we have $ \underset{t\rightarrow \infty }{\text{lim}} \tilde{\chi}_{i}=0_{n}$ and $ \underset{t\rightarrow \infty }{\text{lim}} \tilde{\vartheta}_{i} =0_{n}$. From Step (ii), we have  $\underset{t\rightarrow \infty }{\text{lim}} e_{xi} =0_{n}$, and  $\underset{t\rightarrow \infty }{\text{lim}} e_{vi}=0_{n}$, $ i\in \mathcal{V} $. Hence, we get  the following expression
	\hspace{-10pt}
	\begin{align}
	\underset{t\rightarrow \infty }{\text{lim}} \bar{x}_{i}&= \underset{t\rightarrow \infty }{\text{lim}} (x_{i}-x_{d})=\underset{t\rightarrow \infty }{\text{lim}} (x_{i}- \chi_{i}+\chi_{i}-x_{d}) \notag \\
	&=\underset{t\rightarrow \infty }{\text{lim}}e_{xi}+\underset{t\rightarrow \infty }{\text{lim}}\tilde{\chi}_{i}  = 0_{n} \label{TranXErrors},  \\
	\underset{t\rightarrow \infty }{\text{lim}} \dot{\bar{x}}_{i}&= \underset{t\rightarrow \infty }{\text{lim}} (\dot{x}_{i}-\dot{x}_{d})=\underset{t\rightarrow \infty }{\text{lim}} (\dot{x}_{i}- \vartheta_{i}+\vartheta_{i}-\dot{x}_{d}) \notag \\
	&=\underset{t\rightarrow \infty }{\text{lim}} e_{vi}+\underset{t\rightarrow \infty }{\text{lim}}\tilde{\vartheta}_{i}=0_{n},
	\end{align}
	and the proof is thus completed.
\end{proof}

\section{Numerical Simulations}
In this section, two examples and numerical simulation results are provided to show the effectiveness of the proposed methods. 

\vspace{-6pt}
\subsection{Fault-Tolerant Formation Tracking of Second-Order Nonlinear Multi-Agent Systems with Unknown Control Signs} 
Consider a second-order multi-agent system consisting of six followers and one leader with the followers' dynamics given by 
\vspace{-4pt}
\begin{equation} \label{SFollower}
\left\{ 
\begin{array}{l}
\hspace{-0.3em}  \dot{x}_{i,1} = x_{i,2},  \  y_{i} = x_{i,1}, \ i=1,2,\cdots,6,  \\
\hspace{-0.3em}  \dot{x}_{i,2} = f^{T}_{i}(x_{i}) \theta_{i} +g_{i}(x_{i})u_{ai}+d_{i}(x_{i},t),  \\
\end{array}
\right. 
\end{equation} 
where $ x_{i}=\text{col}(x_{i,1},x_{i,2}) \in \mathbb{R}^{4}$, the nonlinear function $ f_{i}(x_{i}) \in \mathbb{R}^{2\times 2}$ and the unknown parameter $ \theta_{i} \in \mathbb{R}^{2}$ are described by  
\vspace{-3pt}
\begin{equation} \label{SFollower}
f_{i}(x_{i})=\left[ 
\begin{array}{cc} 
-\sin(x^{1}_{i,1}) & x^{2}_{i,2} \\
x^{1}_{i,2} & -x^{2}_{i,1}
\end{array}
\right], 
\ \theta_{i}=\left[ 
\begin{array}{l} 
\theta_{i,1} \\
\theta_{i,2}
\end{array}
\right],
\end{equation} 
where $ x^{k}_{i,1} $, $ x^{k}_{i,2} $, $ k=1,2 $ denote the $ k $-th elements of their states, respectively, $ \theta_{i,1}=0.3i $, $ \theta_{i,1}=0.5i $, $ i=1,2,\cdots,6 $, the control coefficients are $ g_{i}(x_{i})=p_{i}(\cos(x^{T}_{i,1}x_{i,1}+x^{T}_{i,2}x_{i,2}))$ with $ p_{i}=(-1)^{i}*0.1i $,  the time-varying actuator faults $ u_{ai} $ are given by  
\vspace{-3pt}
\begin{equation} \label{gfunction}
u_{ai}=\left\{ 
\begin{array}{l}
\hspace{-0.5em} (0.2\sin(t)+0.4) u_{i}+[2,2\cos(t)]^{T},  \  3 \leq t < 6, \\
\hspace{-0.5em} (0.3\cos(t)+0.6) u_{i}+ [\sin(0.1t),3]^{T},  \ 6 \leq t,  
\end{array} 
\right. 
\end{equation} 
and the bounded uncertainties/disturbances are provided as 
\vspace{-3pt}
\begin{equation} \label{disturbanc}
d_{i}(x_{i},t)=
\left[ 
\begin{array}{l} 
0.1\sin(x^{1}_{i,1}+x^{2}_{i,1})-0.3\cos(0.3t) \\
0.2\cos(x^{1}_{i,2}x^{2}_{i,2})+0.5\sin(0.5t)
\end{array}
\right].  
\end{equation}

The leader's dynamics are described by 
\vspace{-4pt}
\begin{equation} \label{leaderd}
\dot{x}_{0,1} = x_{0,2},  \
\dot{x}_{0,2} = o_{0}(x_{0},u_{0},t),    \ y_{0} = x_{0,1},  
\end{equation}
where $ x_{0}=\text{col}(x_{0,1},x_{0,2})$, $ u_{0}=\text{col}(0.8\sin(t),0.8\cos(t))$ and 
\vspace{-3pt}
\begin{equation} \label{leaderinput} 
o_{0}(x_{0},u_{0},t)=\left[ \begin{array}{l}
0.1\cos(0.1x^{1}_{0,1}+x^{2}_{0,2})+0.8\sin(t)  \\ 
0.2\sin(x^{2}_{0,1}+ 0.2x^{1}_{0,2})+0.8\cos(t) 
\end{array}
\right].  
\end{equation}

\begin{figure}[!t]
	\centering
	\includegraphics[width=3.6cm,height=1.5cm]{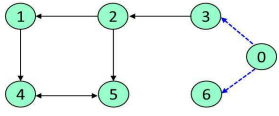}
	\caption{A directed communication topology for the leader-following network. } 
	\label{elsix_agent_undirected}
\end{figure}

The initial states are set as  
$ x_{1,1} = \text{col}(-0.3, -0.5)$, $ x_{2,1} = \text{col}(-2, -1.6)$, $ x_{3,1} = \text{col}(1, -3)$, $ x_{4,1} = \text{col}(0.2, 0.8)$, $ x_{5,1} = \text{col}(2, -1.5)$, $ x_{6,1} = \text{col}(2.5, 1.8)$, $ x_{i,2} = \text{col}(0, 0)$, $ x_{01}=\text{col}(0,-2) $ and $ x_{02}=\text{col}(1,0) $. The communication graph is depicted in Fig. \ref{elsix_agent_undirected}. The prescribed hexagonal formation is $ \varDelta=\text{col}(-1, 0, -\frac{1}{2}, \frac{\sqrt{3}}{2}, \frac{1}{2}, \frac{\sqrt{3}}{2}, 1, 0, \frac{1}{2},  -\frac{\sqrt{3}}{2}, -\frac{1}{2}, -\frac{\sqrt{3}}{2})$. The proposed distributed algorithm in (\ref{Controller}) with estimator in  (\ref{Estimator})-(\ref{filter}) is performed with  parameters selected as $ \kappa_{1i}=15 $, $ \kappa_{2i}=5 $, $ \kappa_{\eta i}=8 $, $ \kappa_{\xi i}=6 $, $ \kappa_{\varrho i} =4$, $ \bar{k}_{1i}= 0.8$,  $ \bar{k}_{2i}= 80$,  $\delta_{i}= 0.05$,  $k_{\kappa i}= 1$, $\varGamma_{\kappa i}= I_{2}$, $\varGamma_{\varepsilon i}= I_{2}$, and $\varGamma_{\theta i}= 10I_{2}$. The simulation results are obtained as illustrated in Figs. \ref{d}-\ref{AgentKapa}. In particular, Fig. \ref{d} depicts the position and velocity estimates of the leader's states, respectively. 
Then, the formation trajectory of the six agents is shown in Fig. \ref{Formation}, where the agents' initial positions are marked by circles and their final positions make the hexagonal formation. In addition, the position and velocity tracking errors between the leader and the followers are shown in Fig. \ref{TrackingErrors}. From Fig. \ref{AgentKapa}, the adaptive parameters $ \hat{\varepsilon}_{i} $ and $  \kappa_{i} $ are bounded. Thus, it can be seen that the formation tracking is achieved for nonlinear multi-agent systems with time-varying actuator faults over the directed graph.

\begin{figure}[!t]  
	\centering
	\hspace*{-1.5em}
	\begin{tabular}{cc}
		\subfloat [Agents' estimates $ \hat{x}_{i,1} $]
		{\includegraphics[width=4.8cm,height=2.5cm]{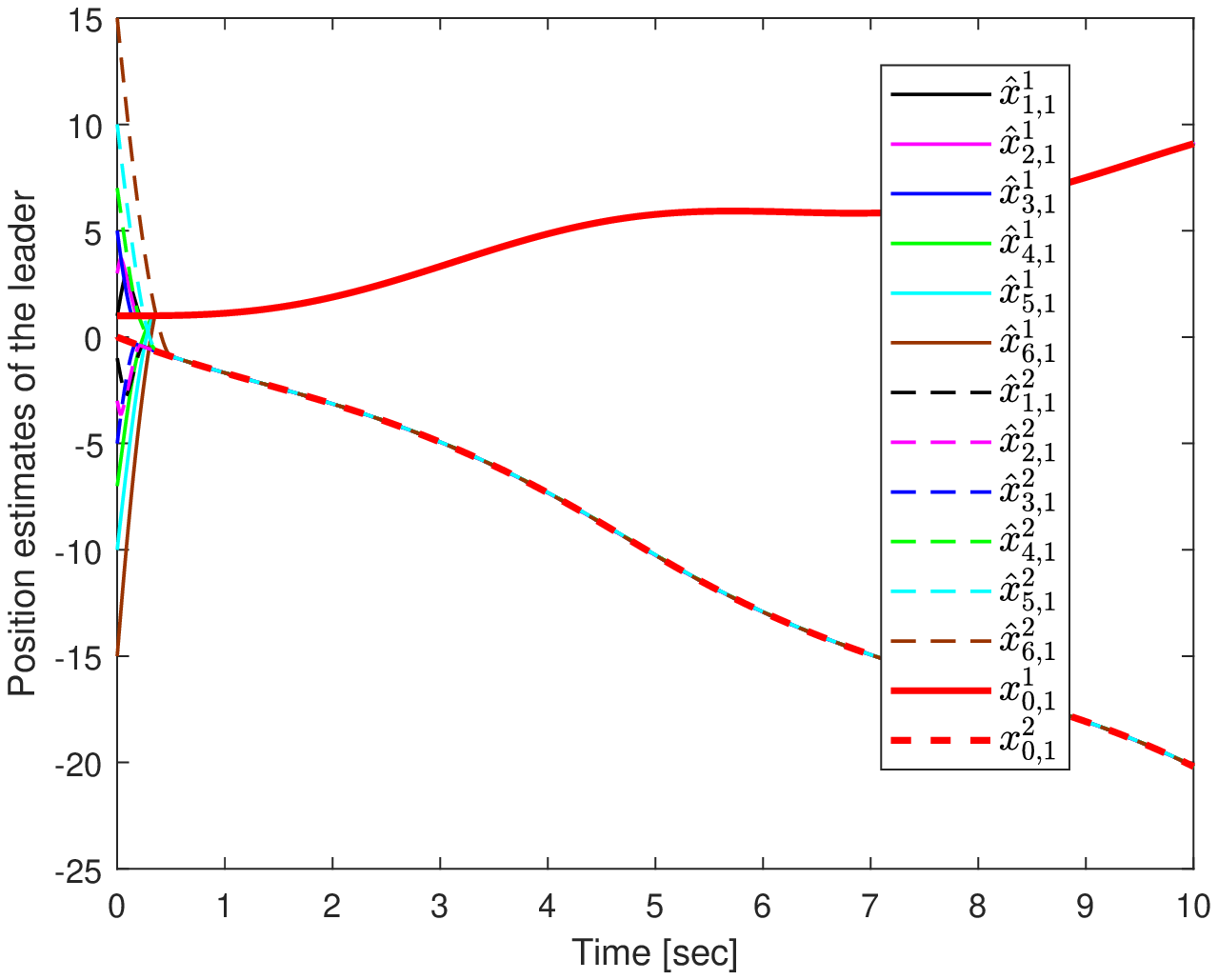} \label{d1}}
		
		\hspace*{-1.5em}
		\subfloat [Agents' estimated errors $ \tilde{x}_{i,1} $.] 
		{\includegraphics[width=4.8cm,height=2.5cm]{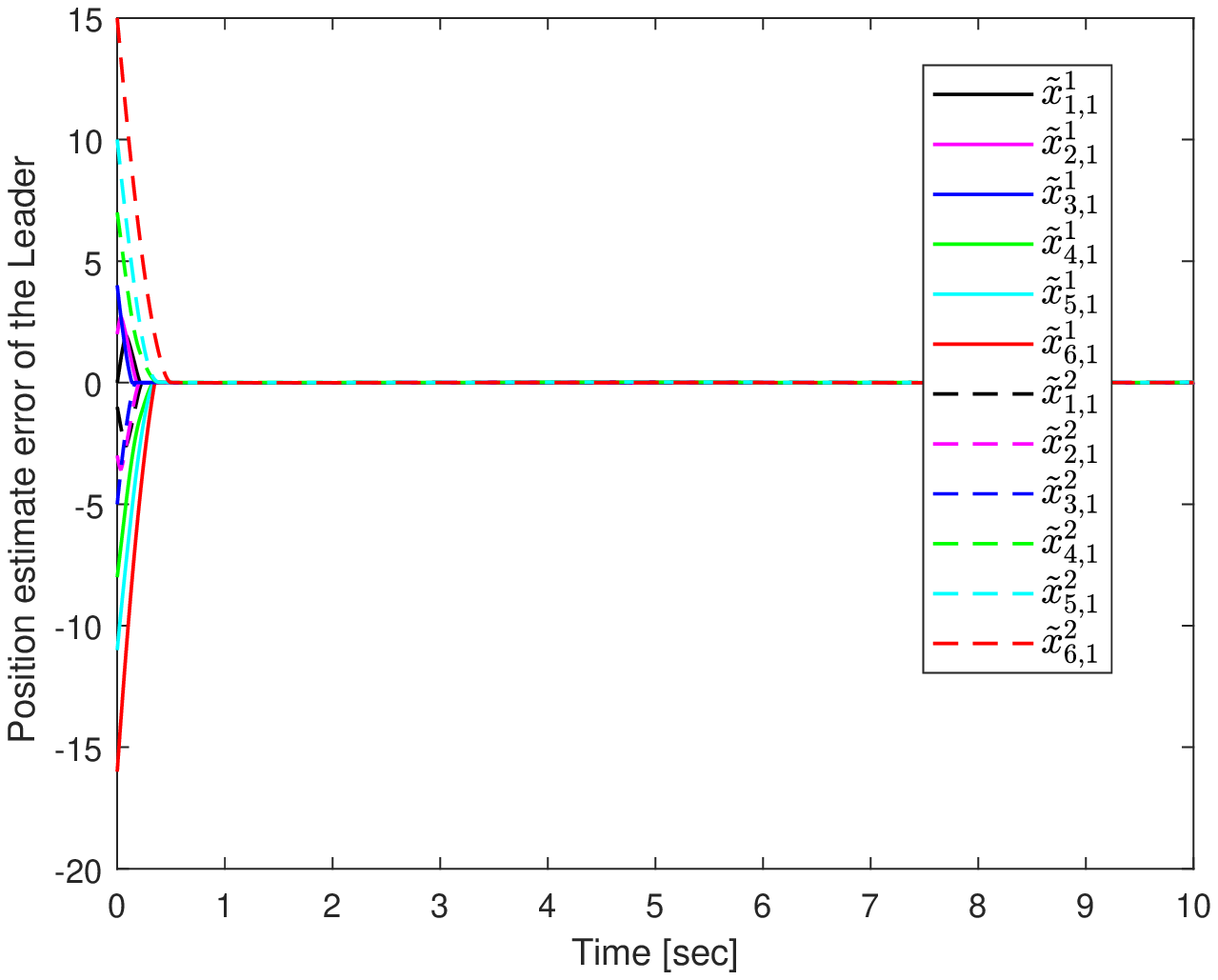}  \label{d2}} 
		
		\\

		\subfloat [Agents' estimates $ \hat{x}_{i,2} $]
				{\includegraphics[width=4.8cm,height=2.5cm]{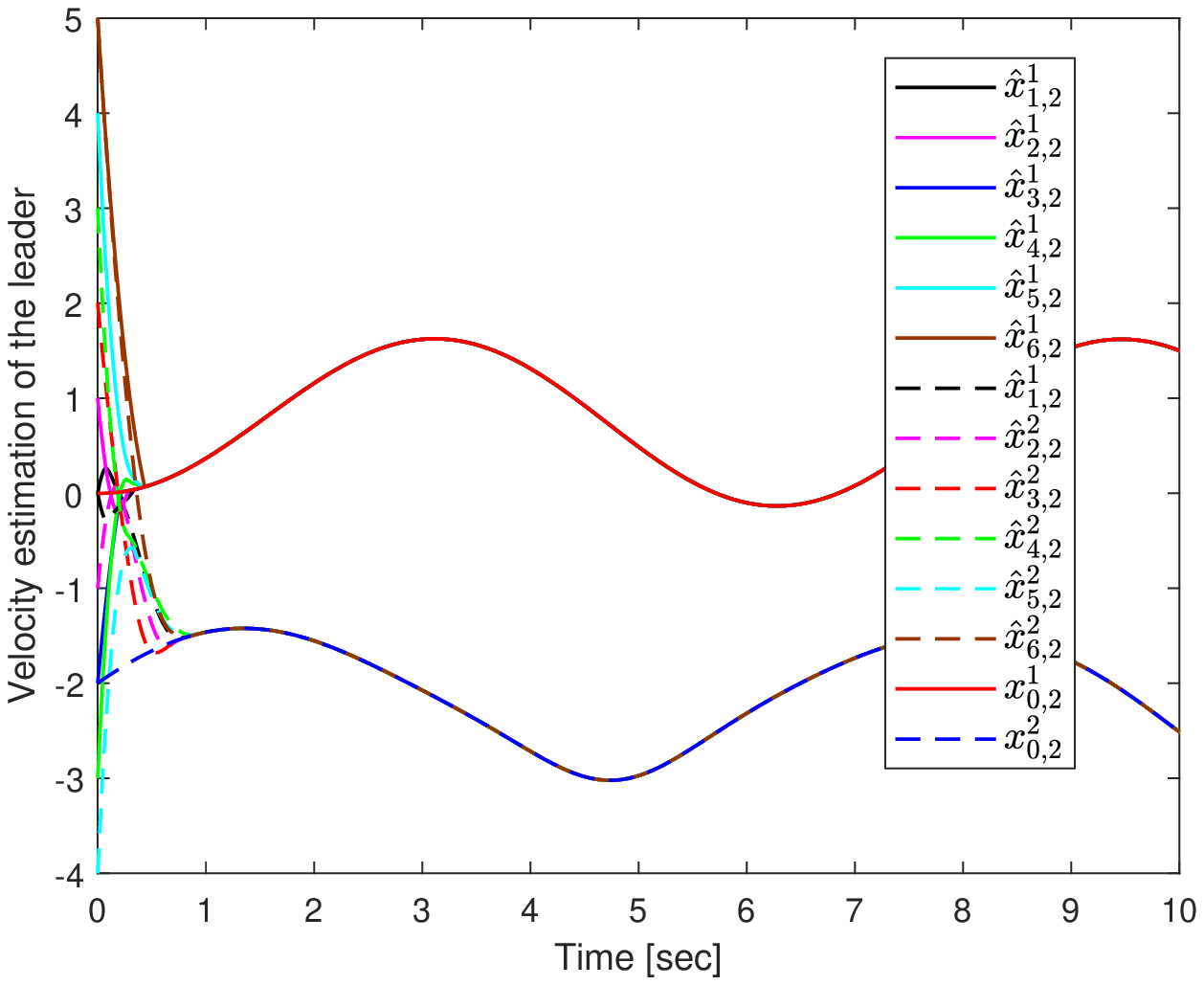} \label{dd1}}
				
		\hspace*{-1.5em}
		\subfloat [Agents' estimated errors $ \tilde{x}_{i,2} $.] 
				{\includegraphics[width=4.8cm,height=2.5cm]{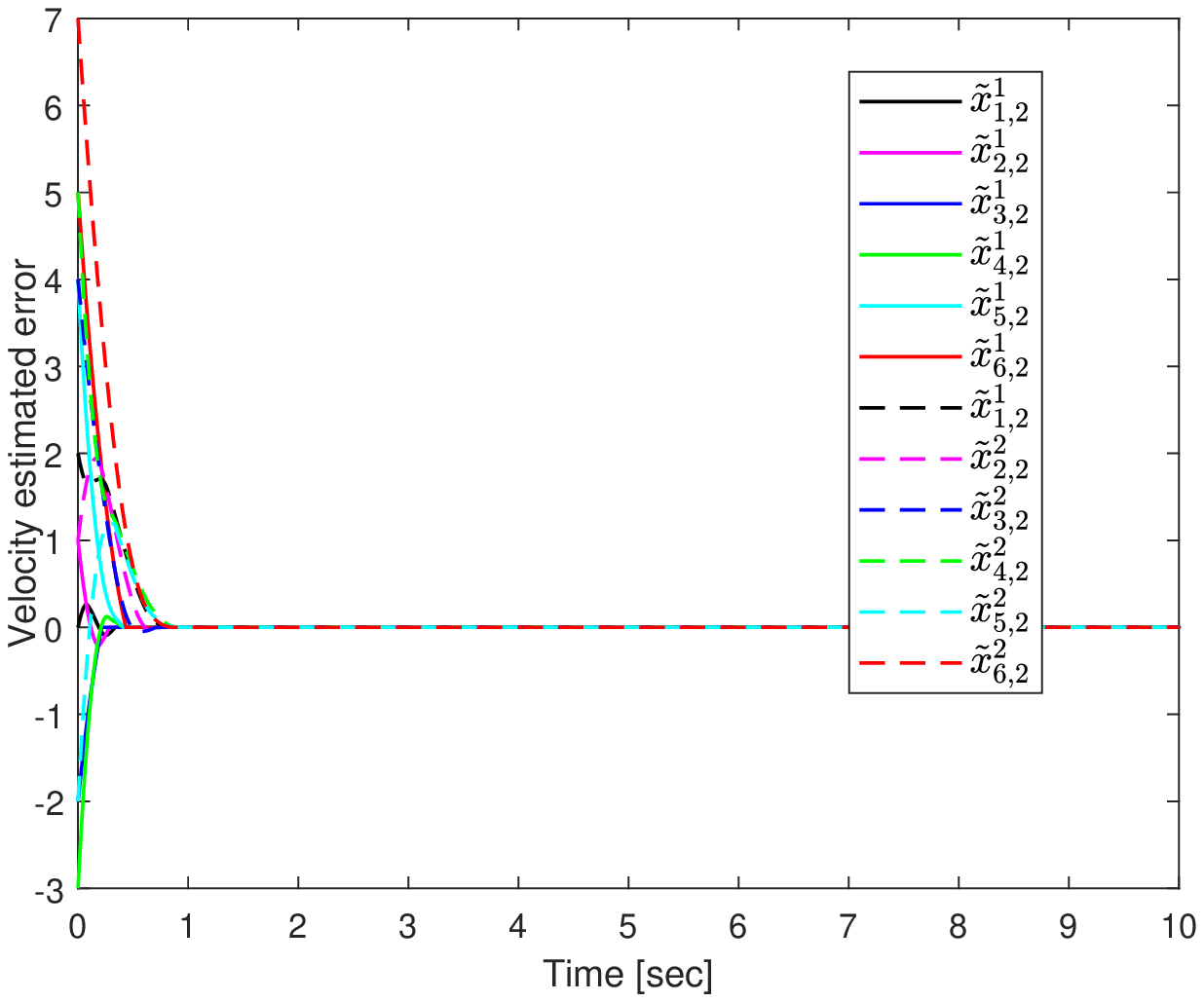}  \label{dd2}} 	
	\end{tabular}
	\vspace{-5pt}
	\caption{The estimates of the leader's states under the estimator (\ref{Estimator})-(\ref{filter}).}
	\label{d}	
\end{figure}

\begin{figure}[!t]
	\centering
	\includegraphics[width=8.0cm,height=5.0cm]{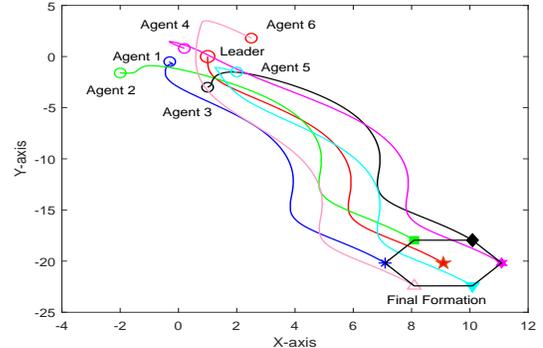}
	\vspace{-5pt}
	\caption{The formation trajectory of the six agents under the controller (\ref{Controller})} 
	\label{Formation}
\end{figure}


 
\begin{figure}[t!] 
	\centering
	\hspace*{-0.5em}
	\begin{tabular}{cccc}	
		\hspace*{-1.0em}		
		\subfloat [Formation tracking error $e_{i,1}$]  {\includegraphics[width=4.8cm,height=2.5cm]{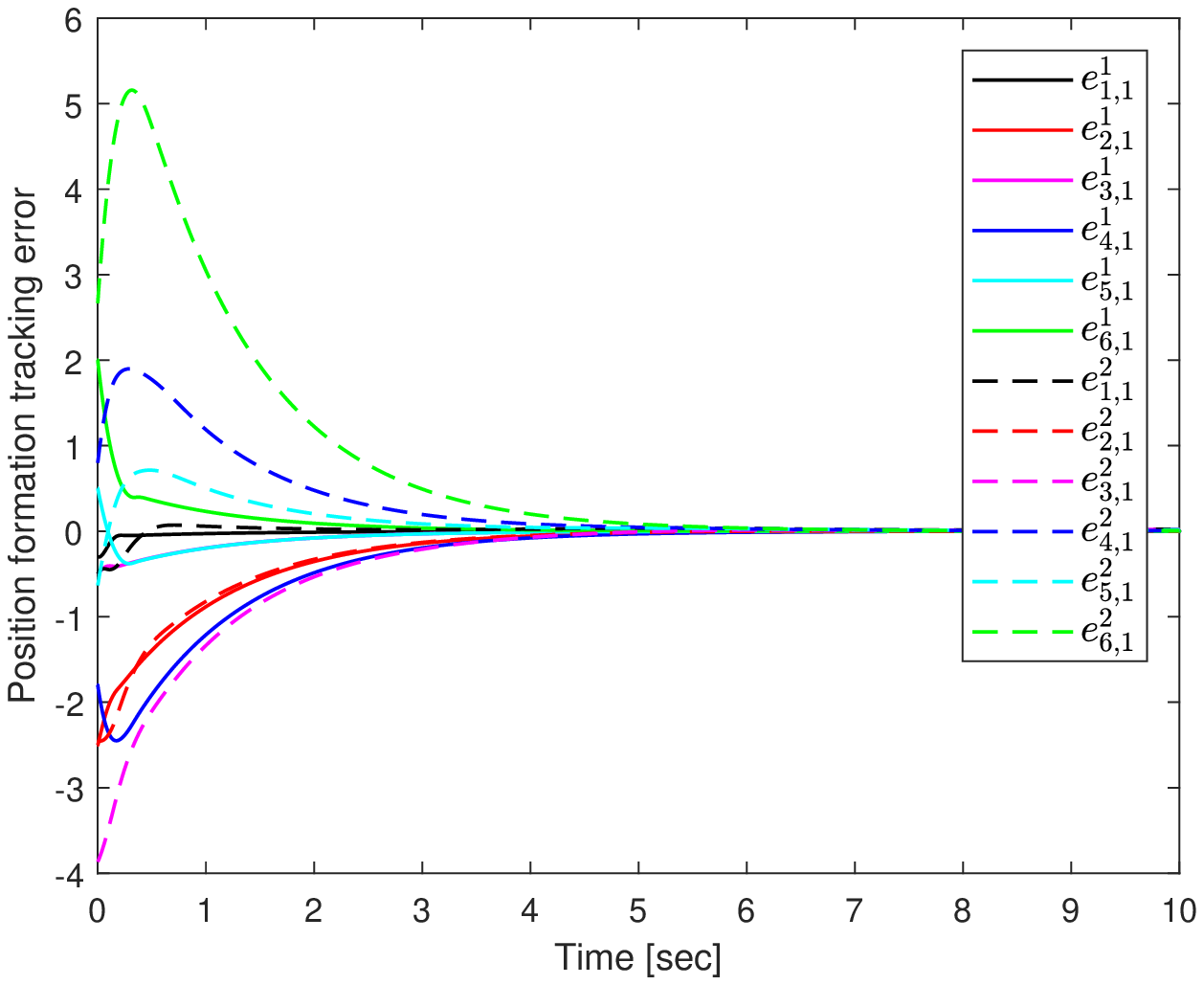} \label{Position_formation_tracking_error}}
		
		\hspace*{-1.0em}		
		\subfloat [Velocity tracking error $ e_{i,2} $ ] {\includegraphics[width=4.8cm,height=2.5cm]{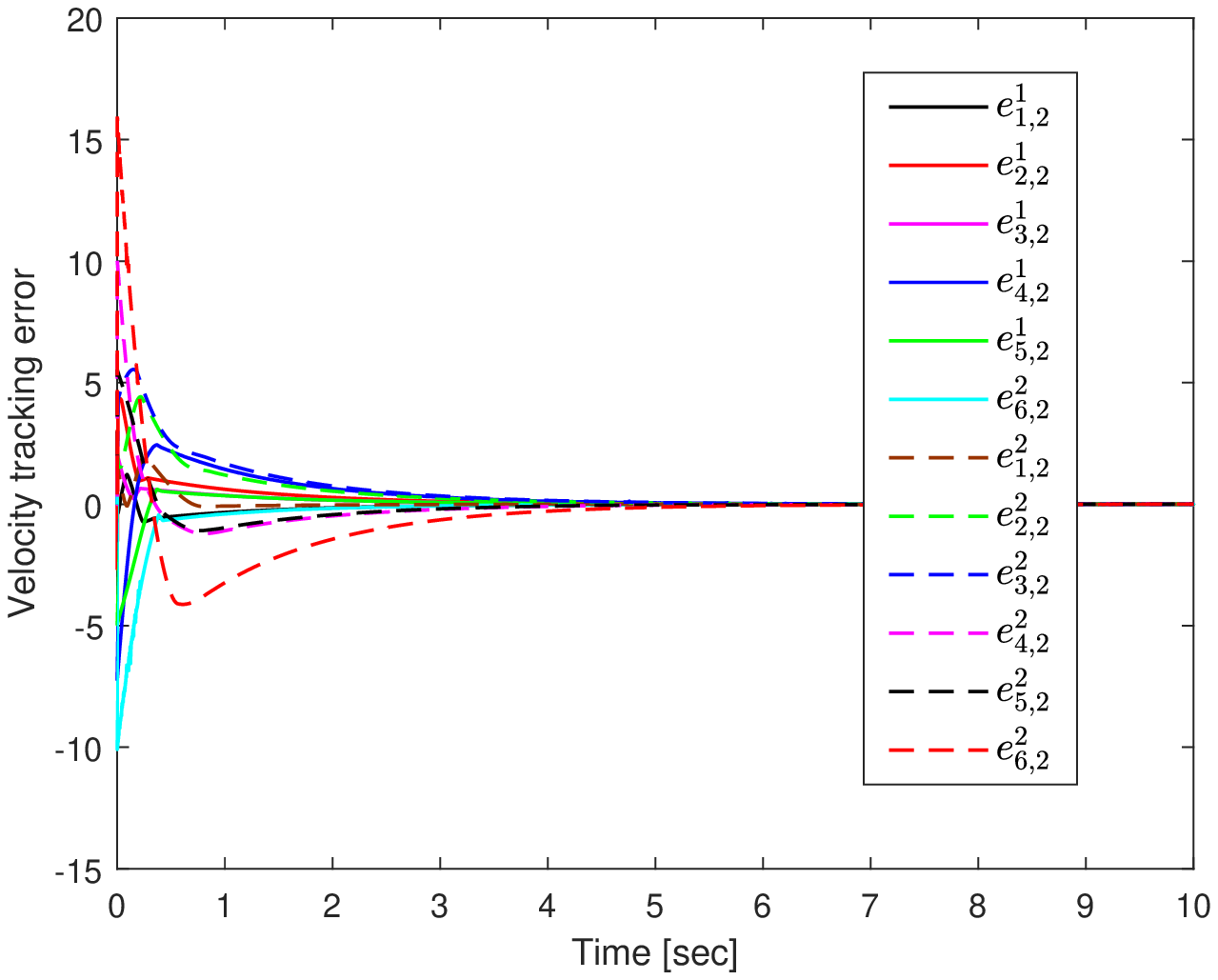} \label{VelocityError}} 
		
	\end{tabular}
	\vspace*{-3pt}
	\caption{The trajectories of position and velocity tracking errors under (\ref{Controller}).}
	\label{TrackingErrors}
\end{figure}

\begin{figure}[t!] 
	\centering
	\hspace*{-0.5em}
	\begin{tabular}{cccc}	
		\hspace*{-1.0em}		
		\subfloat [$ \hat{\varepsilon}_{i} $]  {\includegraphics[width=4.8cm,height=2.5cm]{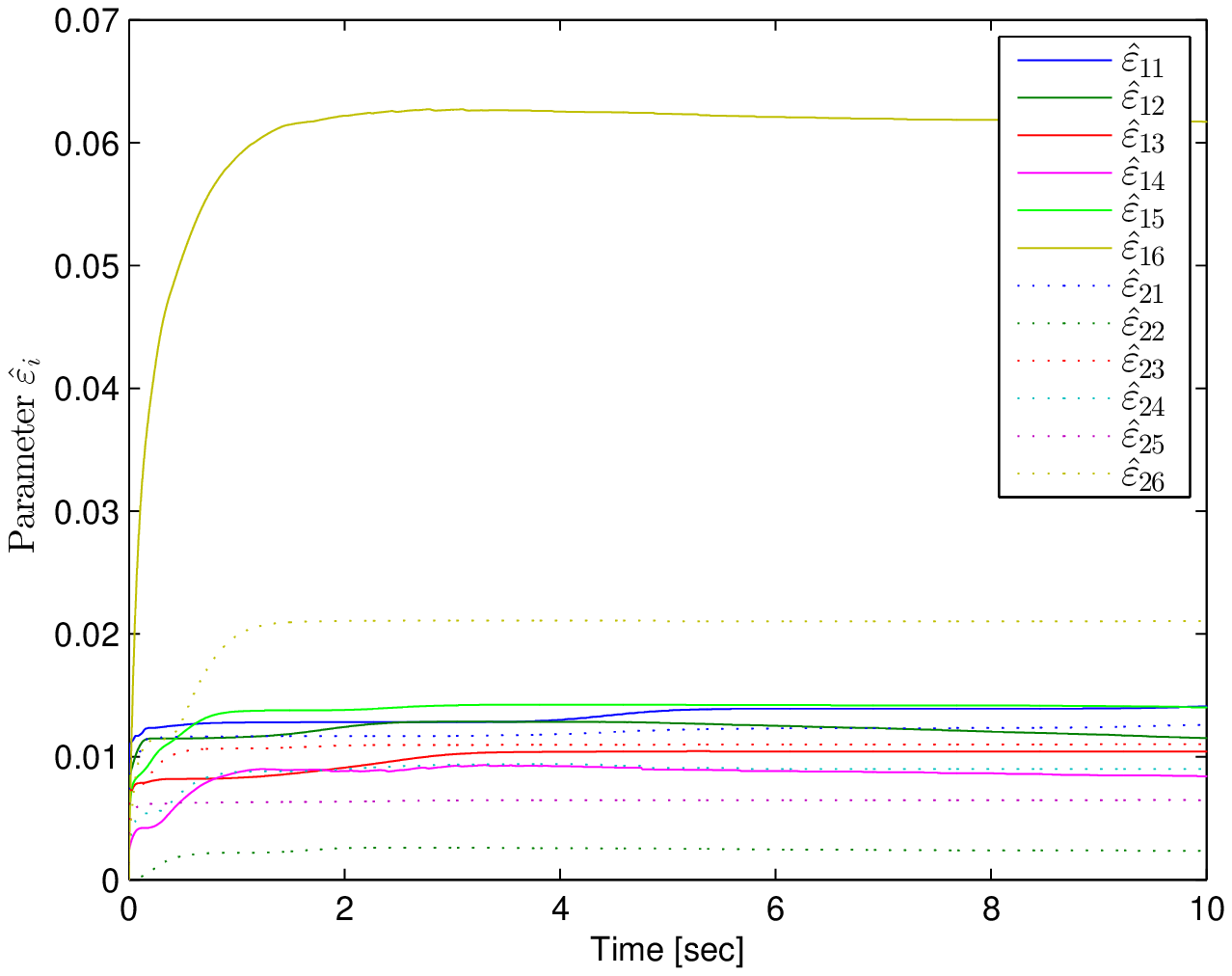} \label{E1}}
 		
		\hspace*{-1.0em}		
		\subfloat [$  \kappa_{i} $ ] {\includegraphics[width=4.8cm,height=2.5cm]{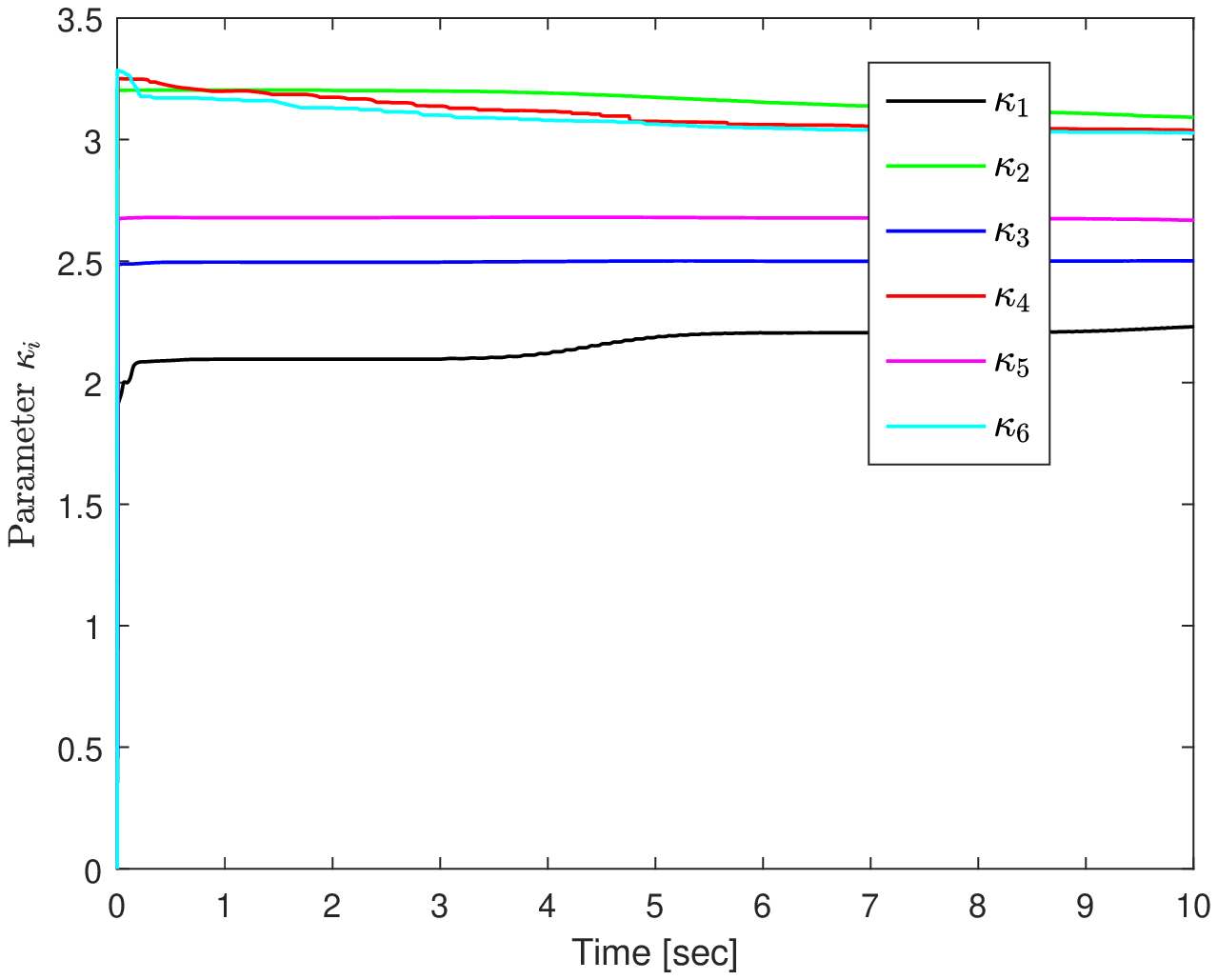} \label{E2}} 
		
	\end{tabular}
	\vspace*{-3pt}
	\caption{The trajectories of adaptive parameters under (\ref{Controller}).}
	\label{AgentKapa}
\end{figure}

\subsection{Fault-Tolerant Task-Space Coordination of Manipulators}
\vspace*{-2pt}
In this simulation, six networked two-link planar manipulators are investigated. Since all robots are planar manipulators, the main objective is to synchronize the end-effectors, while ensuring that they follow the desired trajectory in the task space \cite{Zhang17Tcyber,WangAT,LiuTRO,WangTAC,LiangTcyber}. 

The classic manipulator's dynamics are described by 
\vspace*{-3pt} 
\begin{equation}
M_{i}(q_{i})\ddot{q}_{i}+C_{i}(q_{i},\dot{q}_{i})\dot{q}%
_{i}+G_{i}(q_{i})+F_{i}(q_{i},\dot{q}_{i})=g_{i}(t)\tau_{ai}+d_{i},  \notag
\end{equation}%
where $ q_{i}=\text{col}(q_{1i}, q_{2i}) $ denote the joint angles, and 
\vspace*{-3pt} 
\begin{equation}
M_{i}(q_{i}) =\left( 
\begin{array}{cc}
\theta_{1i}+\theta_{2i}+2\theta_{3i}\cos (q_{2i}) & \theta_{2i}+\theta_{3i}\cos (q_{2i}) \\
\theta_{2i}+\theta_{3i}\cos (q_{2i}) & \theta_{2i}%
\end{array}%
\right), \notag 
\end{equation}%
\vspace*{-5pt} 
\begin{equation}
C_{i}(q_{i},\dot{q}_{i}) =\left( 
\begin{array}{cc}
-\theta_{3i}\sin (q_{2i})\dot{q}_{2i} & -\theta_{3i}\sin (q_{2i})(\dot{q}_{1i}+\dot{q}_{2i}) \\
\theta_{3i}\sin (q_{2i})\dot{q}_{1i} & 0
\end{array}
\right),  \notag 
\end{equation}%
\vspace*{-5pt}  
\begin{equation}
G_{i}(q_{i}) =\left( 
\begin{array}{c}
\theta_{4i}\mathrm{g}\cos (q_{1i}) +\theta_{5i}\mathrm{g} \cos(q_{1i}+q_{2i}) \\
\theta_{5i}\mathrm{g} \cos(q_{1i}+q_{2i})
\end{array}
\right),  \notag 
\end{equation}%
with $ \theta_{i} =\text{col} (\theta_{1i}, \theta_{2i}, \theta_{3i},\theta_{4i},\theta_{5i}) $, where $ \theta_{1i}=I_{1i}+m_{1i}l^{2}_{c1i}+m_{2i}l^{2}_{1i} $,  $ \theta_{2i}=I_{2i} +m_{2i}l^{2}_{c2i} $, $\theta_{3i}= m_{2i}l_{1i}l_{c2i} $,  $ \theta_{4i}=(m_{1i}+m_{2i})l_{1i} $, and $\theta_{5i}=m_{2i}l_{2i}$. The physical parameters of the six robotic manipulators are listed in Table I.
The term $ F_{i}(q_{i},\dot{q}_{i})=F_{vi}\tanh (\dot{q}_{i}) + F_{ci}\text{sgn}(\dot{q}_{i})$ with $F_{vi}=[1 \ 0; 0 \ 1] $, $F_{ci}=[1 \ 1; 1 \ 1] $ are friction forces, and $ d_{i} $ are disturbances described by  
\vspace*{-5pt}
\begin{align}
d_{i}(t)&=\left[
\begin{array}{c}
0.2\text{sin}(\frac{\pi}{10i}t)+0.4\text{sin}(\frac{\pi}{20i}t) \\
0.4\text{cos}(\frac{\pi}{10i}t)+0.8\text{cos}(\frac{\pi}{20i}t)
\end{array}%
\right], i=1,2,\cdots,6.  
\end{align}

In addition, the control coefficients are $ g_{i}(t)=p_{i}(\cos(t)+1.2)$ with $ p_{i}=(-1)^{i}*0.1i $ and the actuator faults $ \tau_{ai} $ are given by
\vspace{-5pt}
\begin{equation} \label{ggfunction}
\tau_{ai}=\left\{ 
\begin{array}{l}
\hspace{-0.5em} (0.2\sin(t)+0.4) \tau_{i}+[2,2\cos(t)]^{T},  \  3 \leq t < 6, \\
\hspace{-0.5em} (0.3\cos(t)+0.6) \tau_{i}+ [\sin(0.1t),3]^{T},  \ 6 \leq t.  
\end{array} 
\right. 
\end{equation} 

In this simulation, the desired end-effector reference trajectory for the 
manipulators is a circle in the task space \cite{LiuTRO}, 
\vspace*{-5pt}
\begin{equation} 
x_{d}(t)= \left[
\begin{array}{c}
1.2+0.5\text{sin}(0.6t) \\
1.0+0.5\text{cos}(0.6t)
\end{array}
\right].  \label{Desiredtrajectory}   
\end{equation}


%

\begin{table}[t!] 
	\hspace*{-1.0em}	
	\centering
	\caption{The physical dynamic parameters of robotic manipulators \cite{WangAT}.}
\vspace*{-5pt}
	\begin{tabular}{c|c|c|c|c}
		\hline	\hline
		\small Robot &  $ m_{1i}, m_{2i} \ \text{(kg)} $  &  $ I_{1i}, I_{2i} \ (\text{kg} \text{m}^{2})  $   &  $ l_{1i}, l_{2i} \ \text{(m)} $   & $ l_{c1i}, l_{c2i} \ \text{(m)} $ \\
		\hline
		\small $ 1$   &  \hspace{-2em}\small 1.5, 1.3  & \hspace{-2em}  \small 0.50, 0.43  & \hspace{-1em}  \small 2.0, 2.0  &\small 1.00, 1.00 \\
		\hline
		\small $ 2$   & \hspace{-2em} \small 1.2, 1.5  & \hspace{-2em}  \small 0.53, 0.36 & \hspace{-1em}  \small 2.3, 1.7  &\small 1.15, 0.85 \\
		\hline
		\small $ 3 $  & \hspace{-2em} \small 1.2, 1.3  & \hspace{-2em}  \small 0.32, 0.52 & \hspace{-1em}  \small 1.8, 2.2  &\small 0.90, 1.10 \\
		\hline	
		\small $ 4 $  &\hspace{-2em} \small 1.8, 1.5  & \hspace{-2em}  \small 0.66, 0.45 & \hspace{-1em}  \small 2.1, 1.9  &\small 1.05, 0.95 \\
		\hline 	
		\small $ 5 $  &\hspace{-2em} \small 1.7, 1.6  & \hspace{-2em}  \small 0.56, 0.43  & \hspace{-1em}  \small 2.0, 1.8  &\small 1.00, 0.90 \\
		\hline	
		\small $ 6 $  &\hspace{-2em} \small 1.9, 1.3  & \hspace{-2em}  \small 0.46, 0.48  & \hspace{-1em}  \small 1.7, 2.1   &\small 0.85, 1.05 \\
		\hline
	\end{tabular}
\end{table}

Then, the forward kinematics for the $ i $th robot are  
\vspace*{-5pt}
\begin{equation}
x_{i} = S_{i}(q_{i})= \left[
\begin{array}{c}
c_{12i}l_{2i} v_{1i}+c_{1i}l_{1i}v_{1i} \\
s_{12i}l_{2i}v_{2i}+ s_{1i}l_{1i}v_{2i}
\end{array}
\right], \ \dot{x}_{i} = J_{i}(q_{i}) \dot{q}_{i},  \label{cartesian}   
\end{equation}
where $ c_{12i}=\text{cos}(q_{1i}+q_{2i}) $, $ c_{1i}=\text{cos}(q_{1i}) $, $s_{12i} =\text{sin}(q_{1i}+q_{2i}) $, $ s_{1i}=\text{sin}(q_{1i}) $, and $ v_{1i}$, $ v_{2i} $ are scaling factors \cite{LiangTcyber}. The overall Jacobian matrix from joint space to task space is expressed as
\vspace*{-5pt}  
\begin{equation}
J_{i}(q_{i})=\left[
\begin{array}{cc}
-s_{12i}l_{2i}v_{1i}-s_{1i}l_{1i}v_{1i} & -s_{12i}l_{2i} v_{1i} \\
c_{12i}l_{2i}v_{2i}+c_{1i}l_{1i}v_{2i} &  c_{12i}l_{2i} v_{2i} 
\end{array}
\right]. \label{Kinematicscartesian}
\end{equation}

Thus, $ \dot{x}_{i} = J_{i}(q_{i}) \dot{q}_{i} $ in (\ref{cartesian}) is written as the product of a known regressor matrix $ Z_{i}(q_{i},\dot{q}_{i}) $ and an unknown constant vector $ a_{i} $ 
\vspace*{-5pt}
\begin{align}
\dot{x}_{i} &=  \left[
\begin{array}{cccc}
-s_{1i}\dot{q}_{1i} & -s_{12i}(\dot{q}_{1i}+\dot{q}_{2i}) & 0  & 0 \\
0 &  0 &  c_{1i}\dot{q}_{1i} &  c_{12i} (\dot{q}_{1i}+\dot{q}_{2i}) 
\end{array} 
\right] \notag \\
& \ \ \ \times \left[
\begin{array}{c}
l_{1i}v_{1i}   \\
l_{1i}v_{2i} \\
l_{2i}v_{1i} \\
l_{2i}v_{2i}
\end{array} 
\right] =Z_{i}(q_{i},\dot{q}_{i})a_{i},  \label{cartesianadaptive}   
\end{align}
where $ a_{i1}=l_{1i}v_{1i} $, $ a_{i2}=l_{1i}v_{2i} $, $ a_{i3}=l_{2i}v_{1i} $, and $ a_{i4}=l_{2i}v_{2i} $.

The simulation results are presented in Figs. \ref{TaskspaceLeader_Estimation_Position}-\ref{Parameters}. Under the proposed finite-time distributed estimator in (\ref{EEstimator}), Fig. \ref{TaskspaceLeader_Estimation_Position} 
depicts the trajectories of the estimated position $ \chi_{i} $ and their tracking errors $ \tilde{\chi}_{i} $. Fig. \ref{Positionerror} shows the position trajectories of the end-effectors $ x_{i} $ and its tracking errors $ \tilde{x}_{i} $ under the proposed distributed controller (\ref{CController}) with (\ref{EEstimator}). The path of the six end-effectors is depicted in Fig. \ref{CircleNresult}, where all the robots can track the circle path in the task space. 
Fig. \ref{Parameters} shows that all adaptive parameters $ \hat{a}_{i}, \hat{\theta}_{i}, \hat{\varepsilon}_{i}, \hat{\kappa}_{i} $ are bounded. It can be concluded from Figs. \ref{TaskspaceLeader_Estimation_Position}-\ref{Parameters} that the task-space cooperative tracking can be achieved under the proposed distributed algorithm for networked manipulators irrespective of uncertain kinematics, dynamics, and time-varying actuator faults.

\vspace{-20pt}
\begin{figure}[h!] 
	\centering
	\hspace*{-0.5em}
	\begin{tabular}{cc}	
		
		\subfloat [Estimated position $ \chi_{i} $] 
		{\includegraphics[width=4.8cm,height=3.8cm]{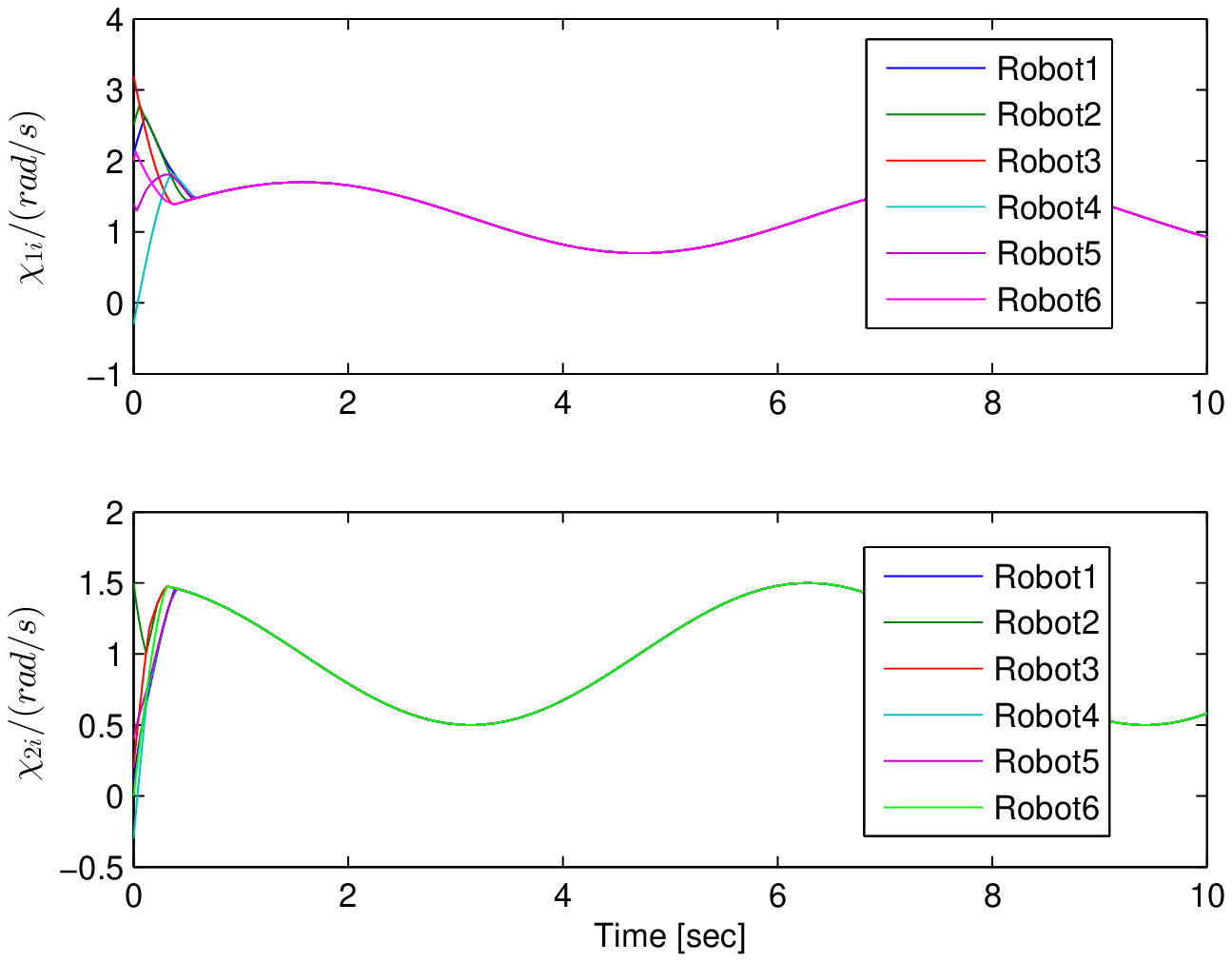}
			\label{L1}}
		
		\hspace*{-1.2em}		
		\subfloat [Estimated position error $ \tilde{\chi}_{i} $] 
		{\includegraphics[width=4.8cm,height=3.8cm]{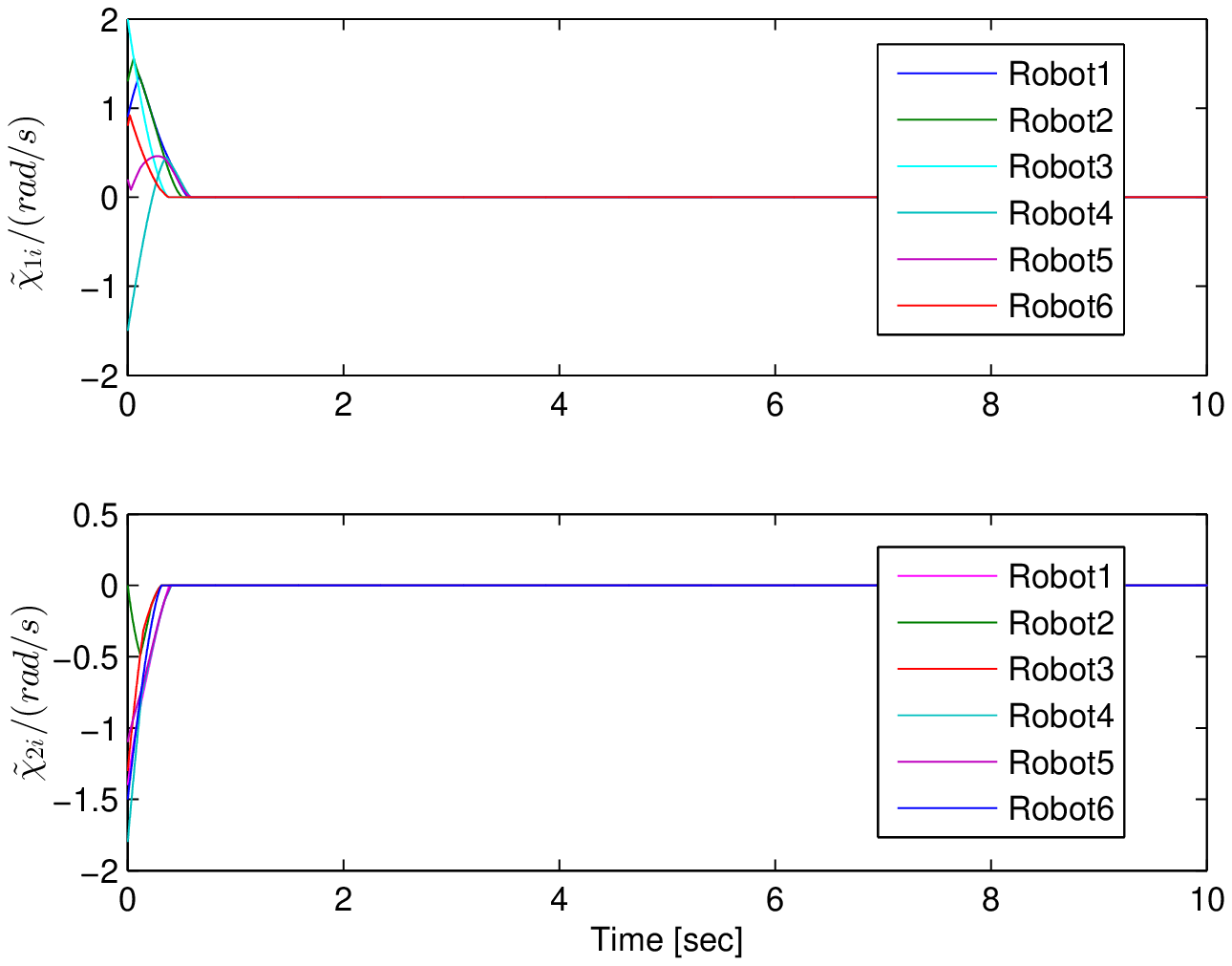}
			\label{L2}}
	\end{tabular}
	\vspace*{-5pt}
	\caption{The estimated position trajectories of the global task reference under the proposed distributed estimator (\ref{EEstimator}).}
	\label{TaskspaceLeader_Estimation_Position}
\end{figure}


\vspace{-28pt}
\begin{figure}[h!] 
	\centering
	\hspace*{-0.5em}
	\begin{tabular}{cc}	
		
		\subfloat [Position $ x_{i} $] 
		{\includegraphics[width=4.8cm,height=3.8cm]{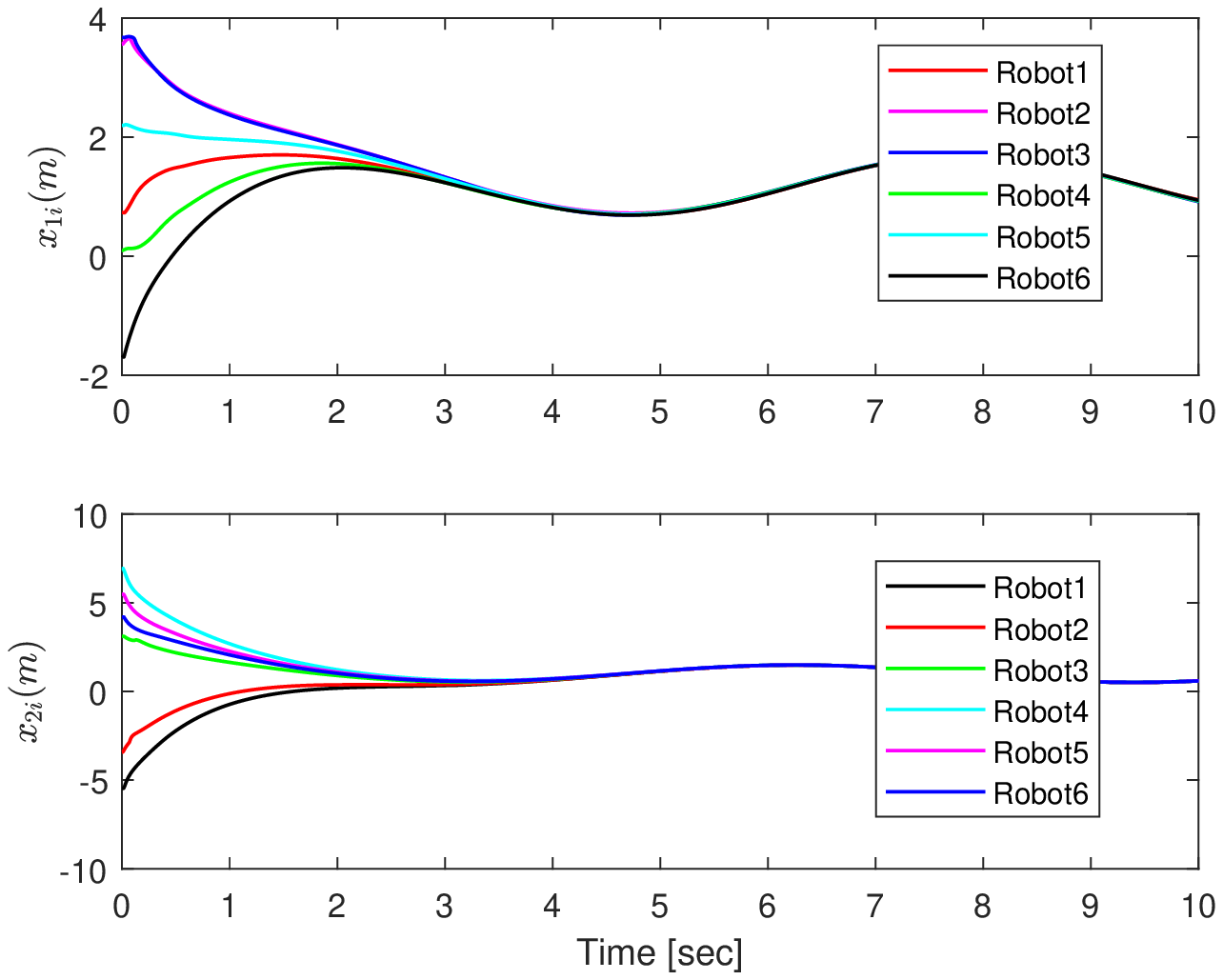}
			\label{L1}}
		
		\hspace*{-1.2em}		
		\subfloat [Position tracking error $ \tilde{x}_{i} $] 
		{\includegraphics[width=4.8cm,height=3.8cm]{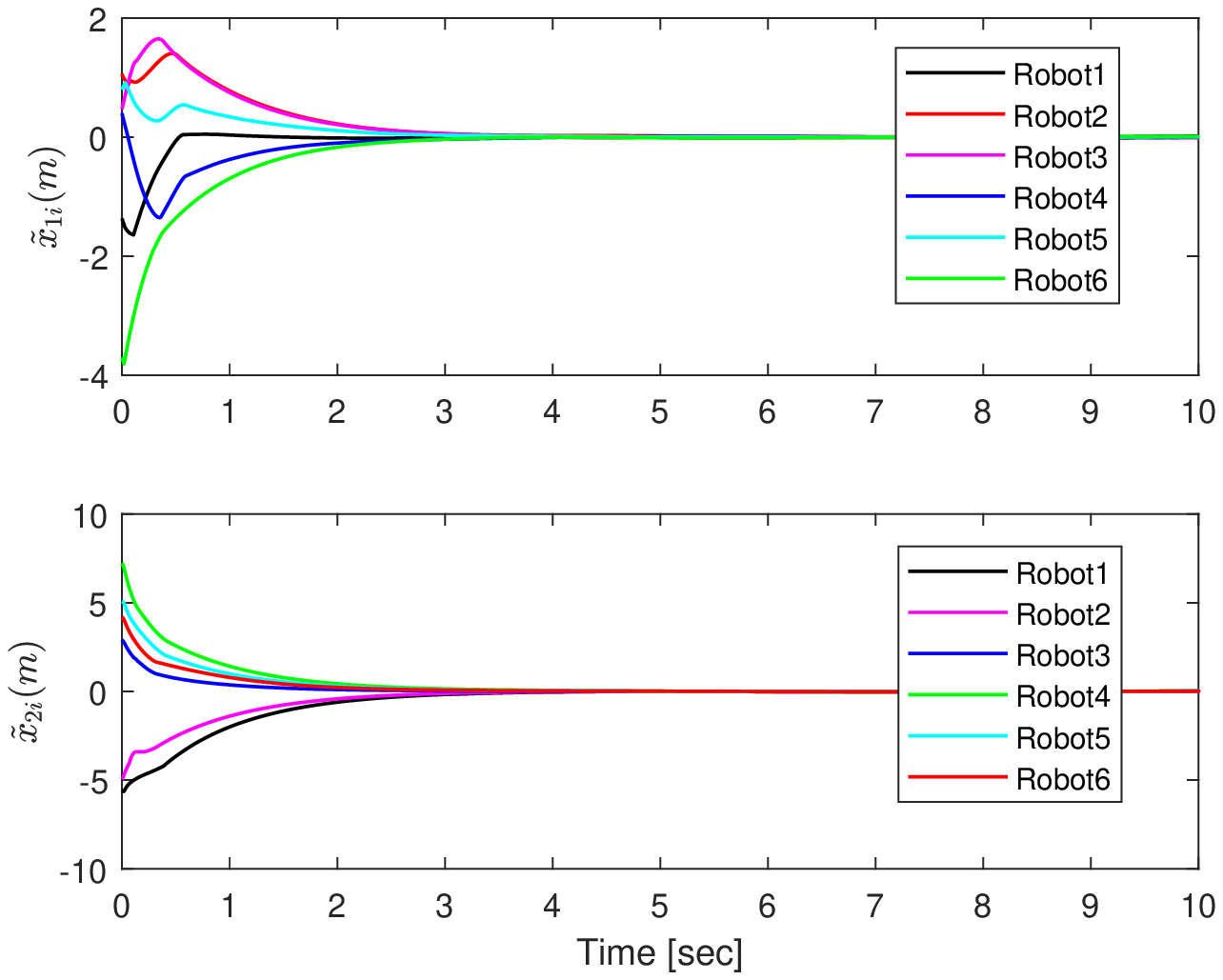}
			\label{L2}}
	\end{tabular}
	\vspace*{-5pt}
	\caption{The position trajectories of end-effectors under the proposed distributed control algorithm (\ref{CController}) with the distributed estimator (\ref{EEstimator}).}
	\label{Positionerror}
\end{figure}

\vspace{-20pt}
\begin{figure}[h!] 
\hspace{1.5em}
	{\includegraphics[width=7.8cm,height=5.5cm]{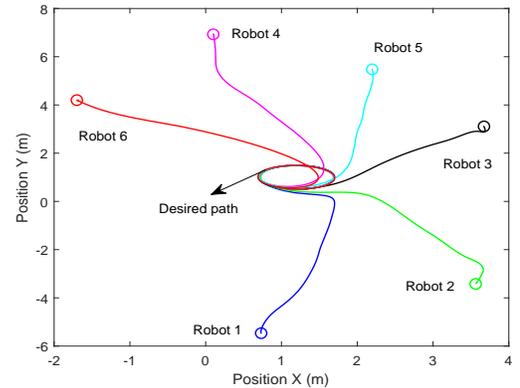}
		\label{ss2}}
		\vspace*{-3pt}
	\caption{The path of end-effectors $ x_{i} $ under (\ref{CController}) with (\ref{EEstimator}).	}
	\label{CircleNresult}
\end{figure}


%
%
%
%
%
%
%

\begin{figure*}[t!] 
	\centering
	\hspace*{-0.5em}
	\begin{tabular}{cccc}	
		\hspace*{-1.5em}		
		\subfloat [$ \hat{a}_{i} $]  {\includegraphics[width=5.0cm,height=3.3cm]{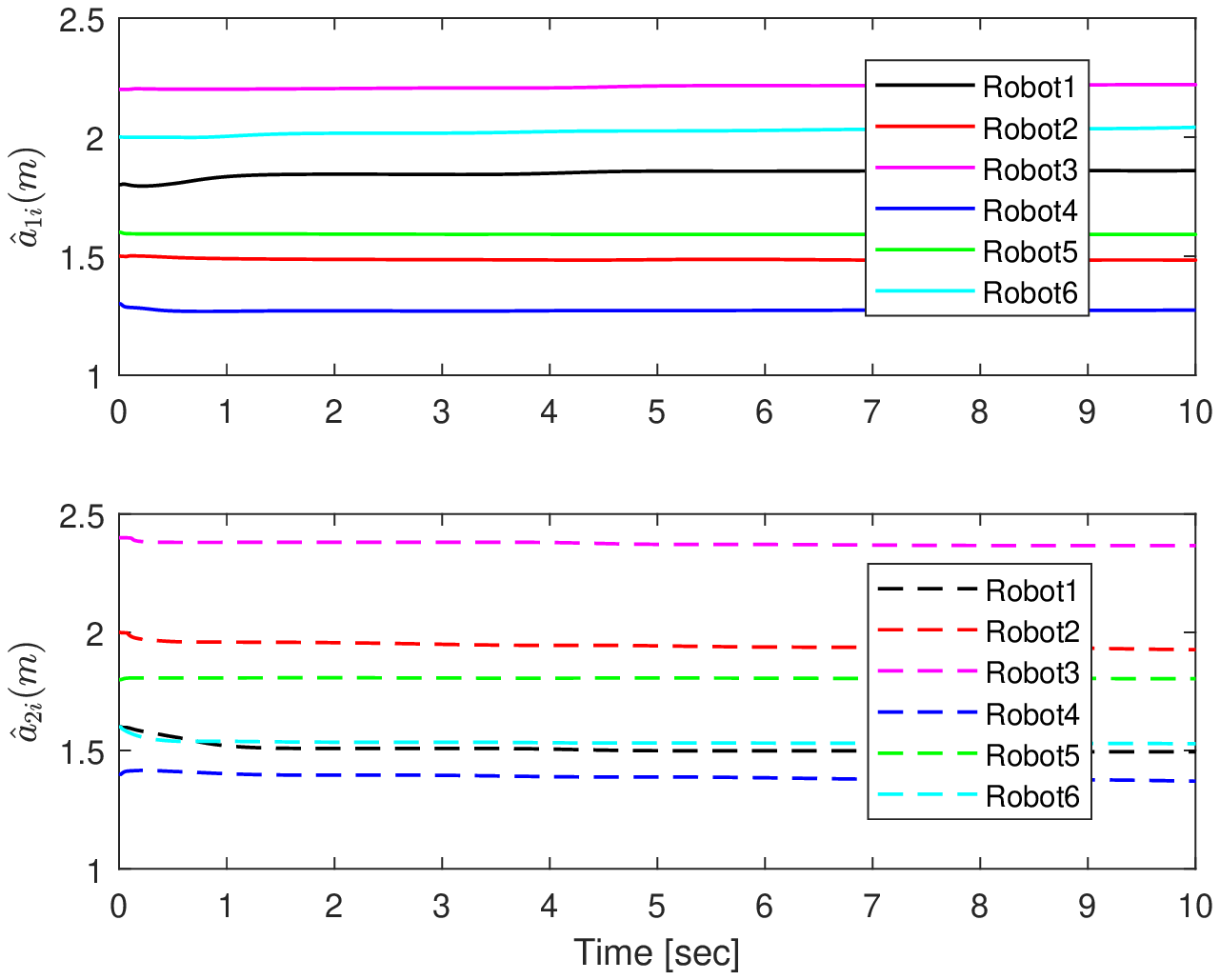} \label{E1}}
		
		\hspace*{-1.5em}		
		\subfloat [$ \hat{\theta}_{i} $ ] {\includegraphics[width=5.0cm,height=3.3cm]{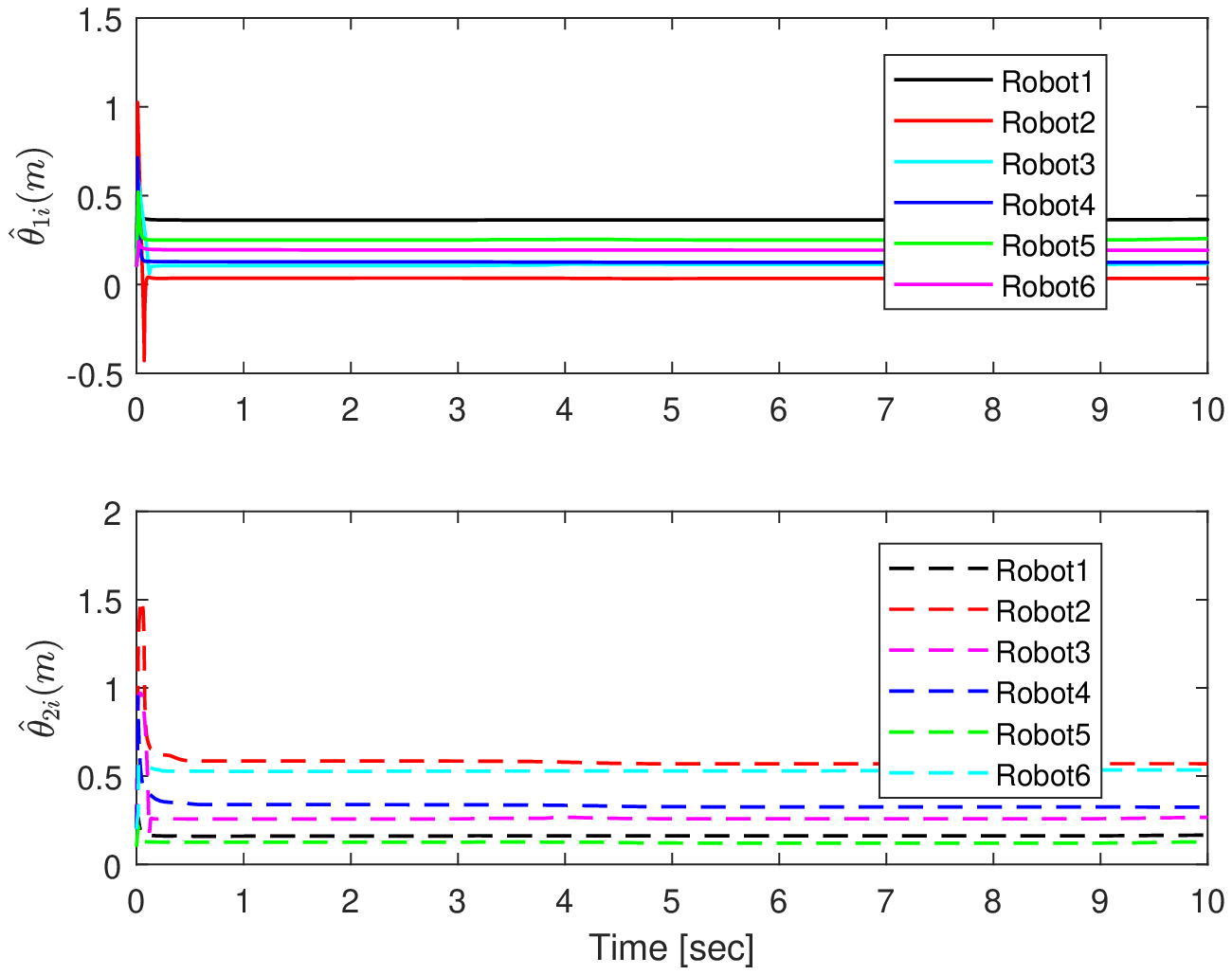} \label{E2}} 
		
		\hspace*{-1.5em}		
		\subfloat [$ \hat{\varepsilon}_{i} $] {\includegraphics[width=5.0cm,height=3.3cm]{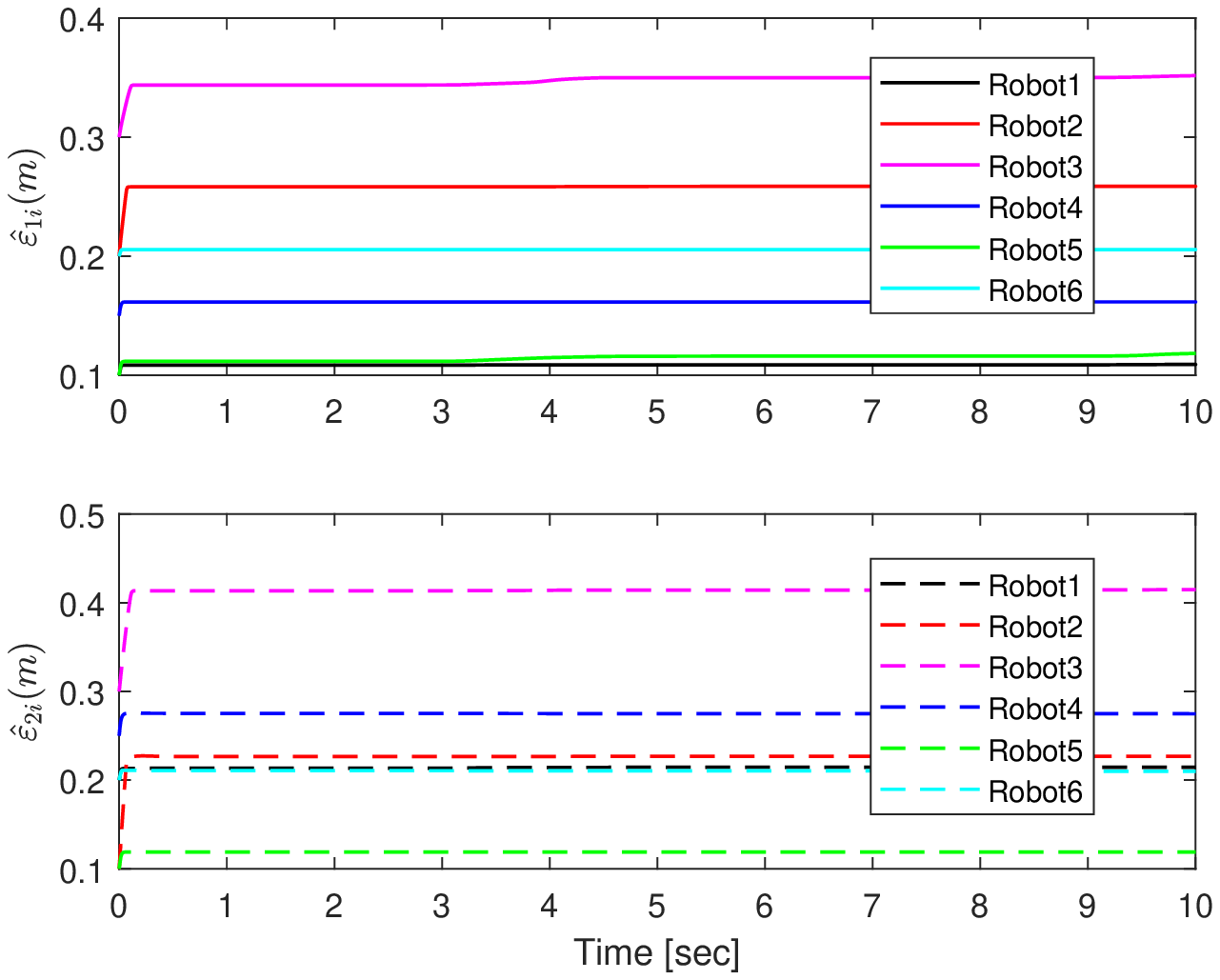} \label{E3}} 
		
		\hspace*{-1.5em}		
		\subfloat [$ \hat{\kappa}_{i} $] {\includegraphics[width=5.0cm,height=3.3cm]{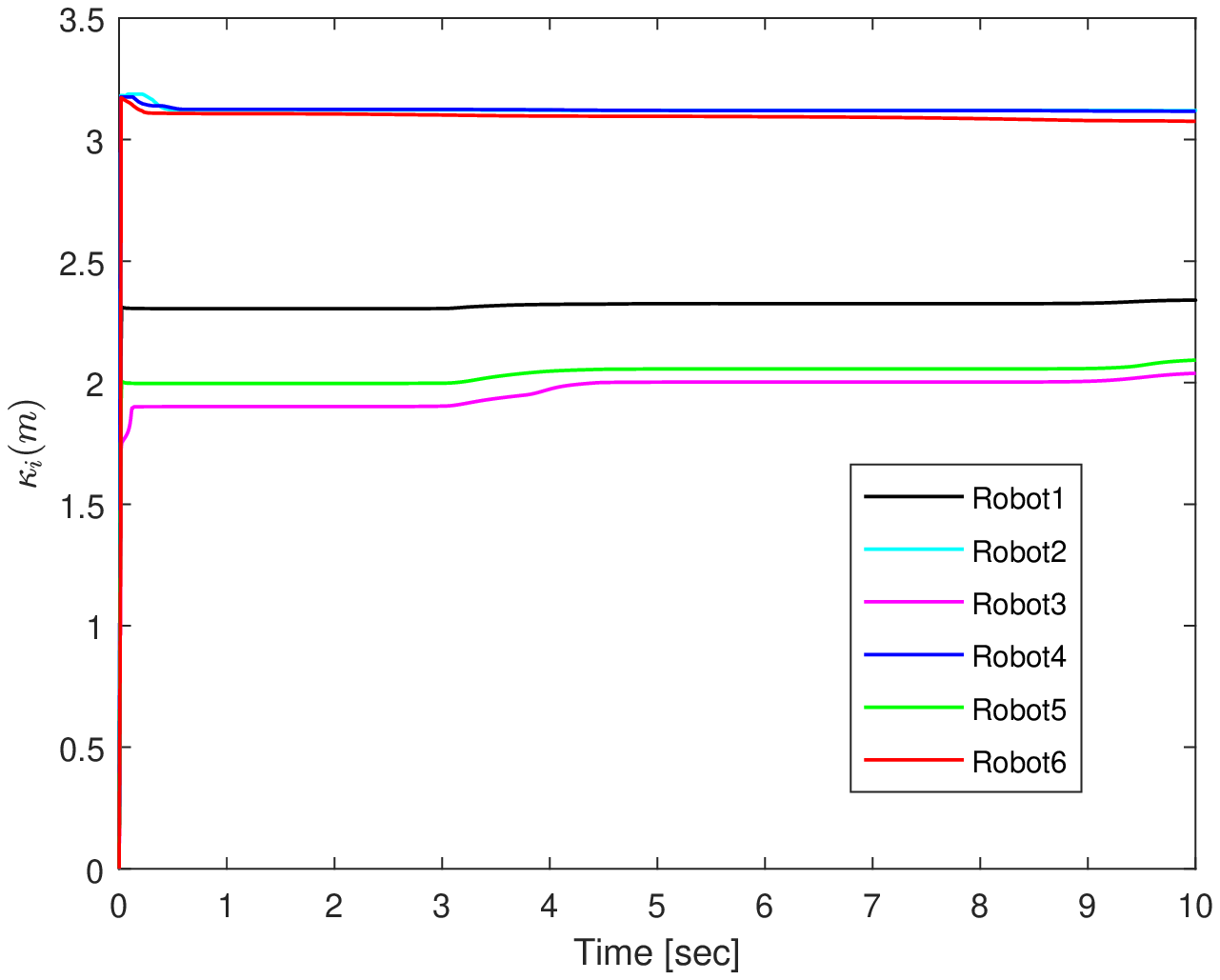} \label{E4}}
		
	\end{tabular}
	\vspace*{-3pt}
	\caption{The trajectories of parameters under the proposed distributed algorithm (\ref{CController}) with (\ref{EEstimator}).}
	\label{Parameters}
\end{figure*}

\vspace*{-15pt} 
\section{Conclusion}
\vspace*{-3pt}
In this paper, we considered the fault-tolerant formation tracking problem for nonlinear multi-agent systems with time-varying actuator faults over the directed graph. We proposed a distributed estimation and control framework by incorporating a distributed nonlinear estimator and a Nussbaum gain technique. Based on the proposed algorithm, the problem was solvable under completely unknown control coefficients. The proposed strategy was applied to task-space cooperative tracking of networked manipulators with unknown kinematics, dynamics, and actuator faults.


\vspace*{-5pt}
\begin{thebibliography}{99}
\setlength{\itemsep}{-0.05mm}
\vspace*{-3pt}
\bibitem{AI06TRO} Z. Han, K. Guo, L. Xie, Z. Lin, ``Integrated relative localization and leader-follower formation control,'' \emph{Automatica}, 64(1): 20--34, 2019.

\bibitem{Nigam12TCST}  S. Du, X. Sun, M. Cao, W. Wang, ``Pursuing an evader through cooperative relaying in agent surveillance networks,'' \emph{Automatica}, 83(9): 155--161, 2017. 

\bibitem{Sun15AT} R. Zheng, Y. Liu, D. Sun, ``Enclosing a target by nonholonomic mobile robots with bearing-only measurements,'' \emph{Automatica}, 53: 400--407, 2015.

 	
\bibitem{Ahn15AT} Z. Feng, G. Hu, Y. Sun, J. Soon, ``An overview of collaborative robotic manipulation in multi-robot systems,'' \emph{Annual Reviews in Control}, 49 (1): 113--127, 2020.

 
\bibitem{Lin16TAC} Z. Lin, L. Wang, Z. Han, M. Fu, ``A graph Laplacian approach to coordinate-free formation stabilization for directed networks,'' \emph{IEEE Trans. Autom. Control}, 61(5): 1269--1280, 2016. 


\bibitem{Dong17TAC} X. Dong, G. Hu, ``Time-Varying formation tracking for linear multi-agent systems with multiple leaders,'' \emph{IEEE Trans. Autom. Control}, 62(2): 3658--3664, 2017.

\bibitem{Dong16AT} X. Dong, Y. Li, G. Hu, Q. Li, Z. Ren, ``Time-varying formation tracking for UAV swarm systems with switching directed topologies,'' \textit{IEEE Trans. Neural Netw. Learn. Syst.,} 30(12): 3674--3685, 2019.
 

\bibitem{Liu18TRO} X. Liang, H. Wang, Y. Liu, W. Chen, T. Liu, ``Formation control of Nonholonomic mobile robots without position and velocity measurements,'' \emph{IEEE Trans. Robot.}, 34(2): 434--446, 2018.


 
\bibitem{Yang16IAC} H. Ma, G. Yang, ``Adaptive fault tolerant control of cooperative systems with actuator faults and unreliable interconnections,'' \textit{IEEE Trans. Autom. Control,} 61(11): 3240--3255, 2016.

\bibitem{Wen14AT} X. Li, G. Yang, ``Neural-network-based adaptive decentralized fault-tolerant control for a class of interconnected nonlinear systems,'' \textit{IEEE Trans. Neural Netw. Learn. Syst.,} 29(1): 144--155, 2018.

\bibitem{Jiang13FS} D. Ye, M. Chen, H. Yang, ``Distributed adaptive event-triggered fault-tolerant consensus of multi-agent systems with general linear dynamics,'' \textit{IEEE Trans. Cybern.,} 49(3): 757--767, 2019.

\bibitem{Zhao20Tcyber} J. Qin, G. Zhang, W. Zheng, Y. Kang, ``Neural network-based adaptive consensus control for a class of nonaffine nonlinear multi-agent systems with actuator faults,'' \textit{IEEE Trans. Neural Netw. Learn. Syst.,} 30(12): 3633--3644, 2019.



\bibitem{Wen16AT} Y. Wang, Y. Song, F. L. Lewis, C. Wen ``Fault-tolerant finite time consensus for multiple uncertain nonlinear mechanical systems under the single-way directed communication interactions and actuation failures,'' \textit{Automatica,} 64(1): 374--383, 2016. 

\bibitem{Gang16Tcyber} Z. Yu, Z. Liu, Y. Zhang, Y. Qu, C. Su, ``Distributed finite-time fault-tolerant containment control for multiple unmanned aerial vehicles,'' \textit{IEEE Trans. Neural Netw. Learn. Syst.,} 31(6): 2077--2091, 2020.

\bibitem{Lewis15IE} C. Deng, W. Che, P. Shi, ``Cooperative fault-tolerant output regulation for multi-agent systems by distributed learning control approach,'' \textit{IEEE Trans. Neural Netw. Learn. Syst.,} DOI: 10.1109/TNNLS.2019.2958151, 2019.
\bibitem{19ATChen} Z. Chen, ``Nussbaum functions in adaptive control with time-varying unknown control coefficients,'' \textit{Automatica,} 102: 72--77, 2019.


\bibitem{16KanACC} C. Ton, Z. Kan, E. A. Doucette, J. W. Curtis, S. S. Mehta, ``Leader-follower consensus with unknown control direction,'' \textit{the 2016 American Control Conference}, pp: 2820--2825, Boston, MA, July 6-8, 2016.

\bibitem{17HarisTAC} H. E. Psillakis, ``Consensus in networks of agents with unknown high-frequency gain signs and switching topologies,'' \textit{IEEE Trans. Automatic Co} 

\textit{ntrol,} 62(8): 3993--3998, 2017.

\bibitem{19WangTAC} Q. Wang, H. E. Psillakis, C. Sun, ``Cooperative control of multiple agents with unknown high-frequency gain signs under unbalanced and switching topologies,'' \textit{IEEE Trans. Automatic Control,} 64(6): 2495--2501, 2019.

%

\bibitem{14TACWen1} W. Chen, X. Li, W. Ren, C. Wen, ``Adaptive consensus of multi-agent systems with unknown identical control directions based on a novel nussbaum-type function,'' \textit{IEEE Trans. Automatic Control,} 59(7): 1887--1892, 2014. 



\bibitem{Chen17TAC} C. Chen, C. Wen, Z. Liu, K. Xie, Y. Zhang, C. L. P. Chen, ``Adaptive consensus of nonlinear multi-agent systems with nonidentical partially unknown control directions and bounded modelling errors,'' \textit{IEEE Trans. Automatic Control,} 62(9): 4654--4659, 2017.
 
\bibitem{Fan19TAC} B. Fan, Q. Yang, S. Jagannathan, Y. Sun, ``Output-constrained control of non-affine multi-agent systems with partially unknown control directions,'' \textit{IEEE Trans. Automatic Control,} 64(9): 3936--3942, 2019. 


\bibitem{15TACSu} Y. Su, ``Cooperative global output regulation of second-order nonlinear multi-agent systems with unknown control direction,'' \textit{IEEE Trans. Automatic Control,} 60(12): 3275--3280, 2015. 

\bibitem{Huang17TAC} T. Liu, J. Huang, ``Cooperative output regulation for a class of nonlinear multi-agent systems with unknown control directions subject to switching network,'' \textit{IEEE Trans. Automatic Control,} 63(3): 783--790, 2017. 

\bibitem{15TACLiu} 
Y. Wang, Y. Lei, Z. Guan, ``Distributed control of nonlinear multi-agent systems with unknown and nonidentical control directions via event-triggered communication,'' \emph{IEEE Trans. on Cybern.}, 50(5): 1820--2398, 2020.


 

	
	
	
	
\bibitem{LiuTRO} Y. Liu , N. Chopra, ``Controlled synchronization of heterogeneous robotics manipulators in task space,'' \emph{IEEE Trans. on Robotics}, 28(1): 268--275, 2012.
	
\bibitem{WangAT} H. Wang, ``Passivity based synchronization for networked robotic systems with uncertain kinematics and dynamics,'' \emph{Automatica}, 49: 755--761, 2013.
	
\bibitem{WangTAC} H. Wang, ``Task-Space synchronization of networked robotic systems with uncertain kinematics and dynamics,'' \emph{IEEE Trans. Autom. Control}, 58(12): 3169--3174, 2013. 
	
\bibitem{LiangTcyber} X. Liang, H. Wang, Y. Liu, W. Cheng, G. Hu, J. Zhao, ``Adaptive task-space cooperative tracking control of networked robotic manipulators without task-space measurements,'' \emph{IEEE Trans. on Cybern.}, 46(10): 2386--2398, 2016.
	
\bibitem{Zhang17Tcyber} B. Zhang, Y. Jia, ``Task-space synchronization of networked mechanical systems with uncertain parameters and communication delays,'' \emph{IEEE Trans. on Cybern.}, 47(8): 2386--2398, 2017.

\bibitem{He20TNNLS} W. He, Y. Sun, Z. Yan, C. Yang, Z. Li, O. Kaynak, ``Disturbance observer-based neural network control of cooperative multiple manipulators with input saturation,'' \textit{IEEE Trans. Neural Netw. Learn. Syst.,} 31(5): 1735--1745, 2020.

\bibitem{Cai16TAC} S. Wang, J. Huang, ``Adaptive leader-following consensus for multiple Euler-Lagrange systems with an uncertain leader system,'' \textit{IEEE Trans. Neural Netw. Learn. Syst.,} 30(7): 2188-2196, 2018. 



\bibitem{Hu16TCNS1} Z. Feng, G. Hu, W. Ren, W. E. Dixon, J. Mei, ``Distributed coordination of multiple Euler-Lagrange systems,'' \emph{IEEE Trans. Control Netw. Syst.,} 5(1): 55--66, 2018. 

\bibitem{Hu19AT} Z. Feng, G. Hu, ``Connectivity-preserving flocking for networked Lagrange systems with time-varying actuator faults,'' \textit{Automatica,} 109(1): 1--10, 2019. 

\bibitem{Hu19TCNS} Z. Feng, G. Hu, C. G. Cassandras, ``Finite-time distributed convex optimization for continuous-time multi-agent systems with disturbance rejection,'' \emph{IEEE Trans. Control Netw. Syst.,} 7(2): 686--698, 2020.



\bibitem{15ATZuo} W. Dong, ``On consensus algorithms of multiple uncertain mechanical systems with a reference trajectory,'' \emph{Automatica}, 47(9), 2023–2028, 2011.





\bibitem{Wen98TAC} X. Ye, J. Jing, ``Adaptive nonlinear design without a priori knowledge of control direction,'' \textit{IEEE Trans. Automatic Control,} 43: 1617--1621, 1998. 

\bibitem{Moreno12TAC} J. A. Moreno, M. Osorio, ``Strict lyapunov functions for the super-twisting algorithm,'' \emph{IEEE Trans. Autom. Control}, 57(4), 1035--1040, 2012. 

\bibitem{BookKhail} H. K. Khalil. \textit{Nonlinear Systems}, 3rd ed, Prentice-Hall, 2002.
 
\bibitem{Hu20CDC} Z. Feng, G. Hu, ``Fault-tolerant formation tracking of heterogeneous multi-agent systems with time-varying actuator faults,'' \emph{the 59th Conference on Decision and Control}, Republic of Korea, December 14-18, 2020.
 


\end{thebibliography}
\end{document}